\PassOptionsToPackage{table,dvipsnames}{xcolor}

\documentclass[acmsmall,screen,nonacm]{acmart}

\AtBeginDocument{%
  }

\setcopyright{acmlicensed}
\copyrightyear{2018}
\acmYear{2018}
\acmDOI{XXXXXXX.XXXXXXX}

\acmJournal{JACM}
\acmVolume{37}
\acmNumber{4}
\acmArticle{111}
\acmMonth{8}

\usepackage{booktabs}
\usepackage{multirow}
\usepackage[linesnumbered,ruled,vlined]{algorithm2e}
\usepackage{listings}
\usepackage{hyperref}
\usepackage{stmaryrd}
\usepackage{graphicx}
\usepackage{subcaption}
\usepackage[most]{tcolorbox}
\usepackage{cleveref}
\usepackage{dsfont} 
\usepackage{import} 
\usepackage{siunitx} 
\usepackage{caption} 
\usepackage{xspace}
\tcbuselibrary{most}
\usepackage[frozencache,cachedir=.]{minted}
\usepackage{tikz}
\usepackage{framed}
\usepackage{wrapfig}
\usepackage{array}
\usepackage{makecell}
\usepackage{diagbox}
\usepackage{mathpartir}
\usepackage[normalem]{ulem} 

\usepackage[T1]{fontenc}
\usepackage[utf8]{inputenc}

\usepackage{xcolor}
\DeclareUnicodeCharacter{00B7}{\ensuremath{\cdot}}
\DeclareUnicodeCharacter{00B9}{\ensuremath{{}^1}}
\DeclareUnicodeCharacter{00D7}{\ensuremath{\times}}
\DeclareUnicodeCharacter{03A6}{\ensuremath{\Phi}}
\DeclareUnicodeCharacter{03C4}{\ensuremath{\tau}}
\DeclareUnicodeCharacter{03C6}{\ensuremath{\varphi}}
\DeclareUnicodeCharacter{2014}{--}
\DeclareUnicodeCharacter{2019}{'}
\DeclareUnicodeCharacter{207B}{\ensuremath{{}^{-}}}
\DeclareUnicodeCharacter{2081}{\ensuremath{{}_1}}
\DeclareUnicodeCharacter{2082}{\ensuremath{{}_2}}
\DeclareUnicodeCharacter{2192}{\ensuremath{\to}}
\DeclareUnicodeCharacter{2194}{\ensuremath{\leftrightarrow}}
\DeclareUnicodeCharacter{2200}{\ensuremath{\forall}}
\DeclareUnicodeCharacter{2203}{\ensuremath{\exists}}
\DeclareUnicodeCharacter{2208}{\ensuremath{\in}}
\DeclareUnicodeCharacter{2218}{\ensuremath{\circ}}
\DeclareUnicodeCharacter{2227}{\ensuremath{\land}}
\DeclareUnicodeCharacter{2261}{\ensuremath{\equiv}}
\DeclareUnicodeCharacter{227A}{\ensuremath{\prec}}
\DeclareUnicodeCharacter{227B}{\ensuremath{\succ}}
\DeclareUnicodeCharacter{27E8}{\ensuremath{\langle}}
\DeclareUnicodeCharacter{27E9}{\ensuremath{\rangle}}
\DeclareUnicodeCharacter{FF1A}{:}

\usepackage{enumitem}
\setlist{leftmargin=7mm}

\definecolor{dg}{RGB}{54, 156, 90}
\definecolor{lightgreen}{rgb}{0.85,1,0.85}
\definecolor{lightred}{rgb}{1,0.85,0.85}

\tcbset{
    mybox/.style={colback=white, width=\textwidth, boxrule=0.2mm,title=Prompt, sharp corners,left=1mm, right=1mm, bottom=1mm},
}
\tcbset{enhanced,attach boxed title to top left={xshift=2mm,yshift=-2mm},coltitle=black,boxed title style={size=small,colback=white, boxrule=0.2mm}}

\newcommand\hl[1]{\tcbhighmath{#1}}
\tcbset{highlight math style={
    sharp corners,
    boxrule=0mm,
    colback=gray!20!white,
    left=0mm, right=0mm, top=0.0mm, bottom=0mm
}}

\newtcbox{\highlight}{on line,
    arc=0pt,
    colframe=yellow!50, 
    colback=yellow!50, 
    boxrule=0.0pt, 
    boxsep=0pt, 
    left=1pt, right=1pt, top=2pt, bottom=2pt 
}

\definecolor{keywordcolor}{rgb}{0.7, 0.1, 0.1}
\definecolor{tacticcolor}{rgb}{0.0, 0.1, 0.6}
\definecolor{commentcolor}{rgb}{0.4, 0.4, 0.4}
\definecolor{symbolcolor}{rgb}{0.0, 0.1, 0.6}
\definecolor{sortcolor}{rgb}{0.1, 0.5, 0.1}
\definecolor{attributecolor}{rgb}{0.7, 0.1, 0.1}

\newcommand\restr[2]{{
  \left.\kern-\null resize the bar with \right
  #1 
  \littletaller 
  \right|_{#2} 
  }}
\newcommand{\littletaller}{\mathchoice{\vphantom{\big|}}{}{}{}}

\newcommand{\toolName}{\textsc{just-tri-it}\xspace}

\newtheorem{assumption}{Assumption}

\newcommand{\cfg}{\Sigma}
\newcommand{\judg}[2]{\langle \cfg,\, #1 \rangle \Rightarrow #2}
\newcommand{\U}{\ensuremath{\mathcal{U}}}  
\newcommand{\A}{\ensuremath{\mathcal{A}}}  
\newcommand{\D}{\ensuremath{\mathcal{D}}}   
\newcommand{\True}{\ensuremath{\mathrm{True}}}
\newcommand{\False}{\ensuremath{\mathrm{False}}}

\begin{document}

\title[Reducing Hallucinations in LLM-Generated Code via Semantic Triangulation]{Reducing Hallucinations in LLM-Generated Code\\ via Semantic Triangulation}

\author{Yihan Dai}
\affiliation{%
  \institution{Key Lab of HCST (PKU), MOE; SCS, Peking University}
  \city{Beijing}
  \country{China}}
\email{yhdai25@stu.pku.edu.cn}

\author{Sijie Liang}
\affiliation{%
  \institution{Key Lab of HCST (PKU), MOE; SCS, Peking University}
  \city{Beijing}
  \country{China}
}
\email{2601213433@stu.pku.edu.cn}

\author{Haotian Xu}
\affiliation{%
 \institution{Key Lab of HCST (PKU), MOE; SCS, Peking University}
 \city{Beijing}
 \country{China}}
\email{xht@stu.pku.edu.cn}

\author{Peichu Xie}
\affiliation{%
  \institution{Independent}
  \city{Beijing}
  \country{China}}
\email{xie_peichu@hotmail.com}

\author{Sergey Mechtaev}
\authornote{Sergey Mechtaev is the corresponding author.}
\affiliation{%
  \institution{Key Lab of HCST (PKU), MOE; SCS, Peking University}
  \city{Beijing}
  \country{China}}
\email{mechtaev@pku.edu.cn}

\renewcommand{\shortauthors}{Dai et al.}

\begin{abstract}
Large language models (LLMs) can generate executable code from natural language descriptions, but the resulting programs frequently contain bugs due to hallucinations. In the absence of formal specifications, existing approaches attempt to assess correctness using LLM-generated proxies such as tests or auto-formalized specifications. However, these proxies are produced by the same imperfect models and thus often corroborate rather than catch errors, especially when the model exhibits correlated errors. We introduce semantic triangulation, a theory-grounded framework that decorrelates model errors by transforming the original problem into a dissociative variant---one likely requiring a fundamentally different algorithm---and checks consistency between independently sampled solutions to both problems. We identify theoretical requirements for this framework, and we prove that under a formal model of LLM hallucinations, these properties confer higher confidence in program correctness. We instantiate the framework through four concrete triangulation methods based on problem inversion, decomposition, and solution enumeration. Evaluated on LiveCodeBench and CodeElo across GPT-4o, DeepSeek-V3, and Gemini 2.5 Flash, our tool increases the probability of selecting a correct program by 16\% over baselines (test generation, metamorphic testing, and auto-formalized specifications) and achieves 7\% higher reliability and 7\% higher F1 score in selection-or-abstention scenarios, while being the only method that consistently handles inexact problems admitting multiple valid solutions.
\end{abstract}




\maketitle

\vspace{-1mm}

\section{Introduction}

Large language models (LLMs) turn natural language specifications into executable code, but this convenience comes at a cost: the generated programs often contain bugs~\cite{liu2024exploring} or vulnerabilities~\cite{pearce2025asleep} due to LLM hallucinations (also known as confabulations). This is exacerbated by the growing trend of ``vibe coding'', where developers engage in continuous prompting without thoroughly reviewing the generated code. In a black-box setting, code hallucinations can be automatically mitigated through two orthogonal strategies: (1) pre-generation input improving---such as prompt engineering via CoT~\cite{wei2022chain} or RAG~\cite{lewis2020retrieval}---and (2) post-generation output assessment, which examines generated candidates for correctness~\cite{feng2024don}. Here, we focus on the latter.

Identifying which, if any, programs from a pool of LLM-generated samples are correct is akin to a police detective trying to determine which suspects, if any, are actually innocent. The standard approach, plurality voting~\cite{hansen2002neural}, simply asks all generated programs for their ``story''---its input-output behaviour---and selects the most common answer. However, since LLMs are prone to generating correlated errors~\cite{kim2025correlated}, this is equivalent to interrogating a group of suspects who have already colluded on the same fake alibi. In such cases, plurality voting does not identify the truth; it merely amplifies their shared deception~\cite{tumer1996error}.

To counter this, previous methods employ alternative ``interrogation'' techniques, using generated tests~\cite{huang2023agentcoder,taherkhani2024valtest} or auto-formalized specifications~\cite{endres2024can,misu2024towards}  as proxies for truth. Yet, because these proxies are generated by the same underlying LLM, it is akin to bringing in a biased witness who shares the same flawed logic as the suspects; they will likely corroborate the false alibi, leading to an erroneous conclusion. Meanwhile, metamorphic testing approaches, such as paraphrasing the problem~\cite{wang2024validating}, aim to reveal inconsistencies by probing suspects with slightly altered scenarios. However, since the suspects have rehearsed their alibi thoroughly enough to withstand such variations---much like LLMs' robustness to minor perturbations, the rephrased questions elicit the same erroneous responses, just with different phrasing.

Our key insight is to tackle this problem much like skilled detectives who expose liars by approaching their narrative from fundamentally different, unexpected angles. In particular, we apply \emph{dissociative} problem transformations: perspectives on the same problem that are sufficiently different that an LLM, being an imperfect ``stochastic parrot''~\cite{shani2025tokens,bender2021dangers}, struggles to maintain consistent deception across them. Then, we check the consistency between responses to these dissociated problems using traditional formal methods. While an innocent suspect's account, i.e. correct program behaviour, holds up under any line of questioning because it reflects the truth, hallucinated programs, like liars caught off guard, will betray themselves through contradiction.

This paper introduces \emph{semantic triangulation}, a general, theory-grounded framework that operationalises the above intuition. Given an original coding task $d$, a semantic triangulation method first applies a \emph{dissociative} transformation to obtain a non-trivially different task $d'$. The transformation must ensure that solutions to $d$ and $d'$ are linked through a binary relation over programs (a hyperproperty) $\phi$. Intuitively, $\phi$ captures how a correct solution to $d$ should relate to a correct solution to $d'$---for instance, if $d'$ is obtained by swapping the input and output, $\phi$ may require that one program is the inverse of the other. Formally, $\phi$ must satisfy two properties: (1) \emph{bijection-inducing}, meaning it establishes a one-to-one correspondence between program equivalence classes, so that each solution for $d$---whether correct or erroneous---maps to a unique solution for $d'$; and (2) \emph{correctness-coupling}, meaning this correspondence preserves correctness: a correct solution to $d$ maps to a correct solution to $d'$. We then compare independently sampled solutions to $d$ and $d'$; when a solution to $d'$ agrees with a solution to $d$ as checked via $\phi$, it serves as a \emph{plausibility witness}---independent corroboration that the solution to $d$ is not a hallucination.

Returning to our police detective analogy, each of these theoretical requirements plays a distinct role in exposing deception. The \emph{dissociative} property ensures that the transformed problem approaches suspects from unexpected angles wholly unrelated to their rehearsed deception. The \emph{bijection-inducing} property enables detectives to detect even subtle inconsistencies in the suspects' stories. Finally, \emph{correctness-coupling} ensures that independently fabricated stories are unlikely to corroborate each other by coincidence. Together, these requirements yield a theoretical result: under our formal model of LLM hallucinations, an appropriately designed semantic triangulation provably increases the probability of identifying a correct program compared to plurality voting.

We instantiate our framework through four concrete semantic triangulation methods, employing transformations that invert problems, decompose them into sub-problems, and enumerate their solution spaces. The dissociative character of these transformations stems from a key observation: each transformed problem typically requires the LLM to implement a fundamentally different algorithm---for instance, inversion demands a parser where the original required a printer. Because the resulting programs embody distinct computational strategies, an LLM that hallucinated one is unlikely to produce a matching hallucination for the other, making corroboration across the transformation a strong signal of correctness. We further provide machine-checked proofs (in Lean) that all four methods satisfy the theoretical requirements of our framework.

We implemented all triangulation methods in a tool, dubbed \toolName, and evaluated it on LiveCodeBench~\cite{jain2024livecodebench} and a subset of CodeElo~\cite{quan2025codeelobenchmarkingcompetitionlevelcode} that includes \emph{inexact} problems---tasks where multiple distinct outputs are valid for the same input. We evaluate \toolName across three state-of-the-art models from different families: GPT-4o, DeepSeek-V3, and Gemini Flash 2.5. Our experiments show that a successful triangulation increases the probability of correctness by 16\% over baselines (test generation, metamorphic testing, and auto-formalized Hoare-style specifications). For the task of selection-or-abstention---selecting a solution only when sufficiently confident, and abstaining otherwise---\toolName achieves 7\% higher reliability and 7\% higher abstention F1 than the strongest baseline, which selects only high-confidence solutions with sampling probability $\geq 0.5$. Remarkably, it maintains this reliability while being able to select solutions with sampling probabilities as low as 0.07, thereby enabling reliable outputs on harder tasks. Crucially, \toolName is the only method that handles inexact problems consistently. Finally, we show that the requirements of our framework are not satisfied by prior methods such as auto-formalization into Hoare logic, while an ablation study confirms that these requirements are essential for our method's reliability.

In summary, the paper makes the following contributions:
\begin{itemize}
\item Semantic triangulation, a theory-grounded hallucination detection framework that enhances confidence of LLM-generated code by checking cross-task solution consistency.
\item Practical instances of triangulation, with their design rigorously proven in Lean.
\item Evaluation on LiveCodeBench and CodeElo showing that triangulation outperforms previous methods in correct code selection decisions on instances with low probability of correct solutions, and when the task permits multiple non-equivalent solutions.
\end{itemize}

All code, data, and mechanized proofs are available at \url{https://github.com/msv-lab/just-tri-it}.

\begin{figure}[t]
  \centering
  {\relscale{0.7}\def\svgwidth{\textwidth}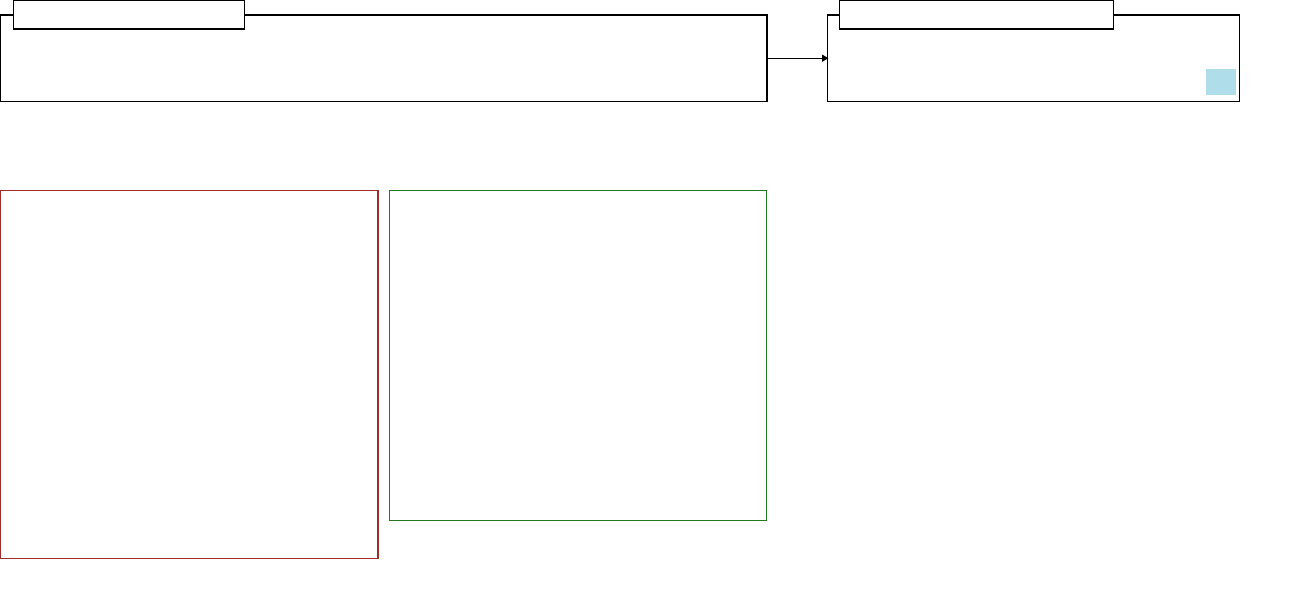}
  \caption{Identifying a correct solution is hard when the model's success probability is low (0.07 in the figure), and it makes correlated errors (the dominant error has the probability 0.23). Auto-inferred Hoare-style specifications tend to be overly permissive, failing to distinguish correct from incorrect solutions. The notation $p \sim m(\,\cdot \mid d)$ means the program $p$ is sampled from a conditional distribution induced by an LLM $m$ given a problem description $d$. $\mathds{P}([p])$ is the probability mass of the equivalence class of $p$, estimated via sampling.\label{fig:problem}}
\end{figure}%


\section{Overview}

To illustrate the challenge of assessing correctness of LLM-generated programs, we use the problem 1999D\footnote{The problem is simplified for presentation; more details are available in \Cref{sec:motivating_details}.} from the CodeElo benchmark~\cite{quan2025codeelobenchmarkingcompetitionlevelcode}, denoted as $d$ and shown in \Cref{fig:problem}. The problem asks to replace each \mintinline{Python}{"?"} in a string $s$ with a lowercase letter so that a target string $t$ becomes a subsequence of $s$, or to report that no such replacement exists. For example, \mintinline{Python}{p("a?a", "ab")} should return \mintinline{Python}{"aba"}. Notably, the problem is \emph{inexact}: multiple solutions may be valid for the same input. For example, given the input \mintinline{Python}{p("??", "a")}, both \mintinline{Python}{"aa"} and \mintinline{Python}{"ab"} are acceptable answers.

GPT-4o has a low success rate on this problem: the estimated probability of sampling a correct solution, equivalent to $p_2$ in \Cref{fig:problem}, is only 0.07, while the dominant incorrect behaviour has the probability 0.23 $\gg$ 0.07. The sample $p_1$ in \Cref{fig:problem} exemplifies this error: it fails to advance the index into the target string $t$ during the main loop, causing, for instance, \mintinline{Python}{p("ab??e", "abcde")} to erroneously return \mintinline{Python}{None}.

Plurality voting~\cite{hansen2002neural} selects the most frequent semantic equivalence class among sampled solutions; majority voting, a stricter variant, requires the selected class to have probability at least 0.5. Both are ineffective due to the tendency of LLMs to produce correlated errors, like $p_1$ in our example,---a phenomenon observed not only across samples from the same model but even across different models~\cite{kim2025correlated}. In such cases, plurality consensus amplifies rather than mitigates mistakes, preferring the dominant error $p_1$ over the correct sample $p_2$, an effect long recognized in the ensemble learning literature~\cite{tumer1996error}. Inexact problems further exacerbate the issue: since multiple correct solutions may exist, plurality voting may favor a single dominant error even when the total probability mass across all correct equivalence classes is higher.

\begin{wrapfigure}{r}{0.42\textwidth}
  \vspace{-3mm}
  \includegraphics[width=0.41\textwidth]{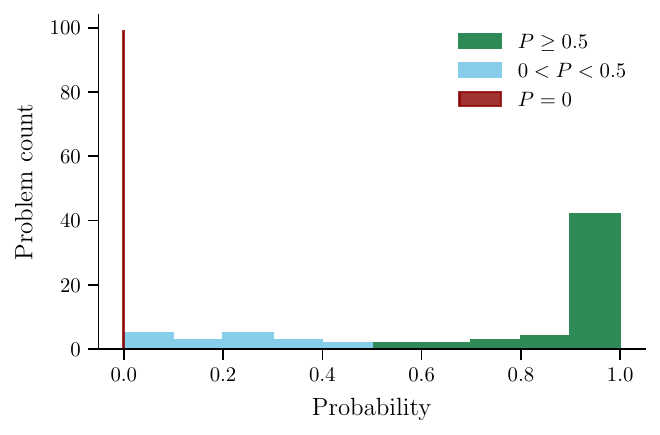}
  \vspace{-3mm}
  \caption{Probabilities of sampling correct solutions from GPT‑4o for LiveCodeBench‑v6 problems (estimated over 100 trials).\label{fig:prob_correct_dist}}
  \vspace{-3mm}
\end{wrapfigure}

Alternative techniques---checking consistency with generated tests~\cite{chen2022codet}, auto-formalized specifications~\cite{endres2024can,misu2024towards}, and metamorphic testing via paraphrasing~\cite{wang2024validating}---all fail to identify a correct solution on this example. We discuss the failure of auto-formalized Hoare-style specifications~\cite{hoare1969axiomatic} in detail, as it reveals a conceptual issue. For this problem, GPT-4o generates a postcondition equivalent to $q_1$ in \Cref{fig:problem} with probability 0.73. This postcondition implements an incomplete check for the answer \mintinline{Python}{None}, and thus holds for both the incorrect sample $p_1$ and the correct sample $p_2$, failing to distinguish between them.
This failure has two causes. First, complete behavioral specifications are
inherently hard to write, so real-world specifications tend to be overly
permissive---as evidenced by benchmarks such as
SV-COMP~\cite{beyer2012competition}---and models reproduce this tendency. Second,
and more fundamentally, the semantics of a Hoare triple does not require the
postcondition to capture the intended behavior exactly: $\{\mathrm{pre}\}\, p\,
\{\mathrm{post}\}$ holds for any postcondition that overapproximates the behavior
of $p$, up to the trivial $\mathrm{True}$. Errors in specifications thus need not
mirror errors in programs: a single incorrect or incomplete specification accepts
many incorrect programs, making spurious matches between LLM-generated
specifications and programs highly probable.

\begin{figure}[t]
  \centering
  \includegraphics[width=\textwidth]{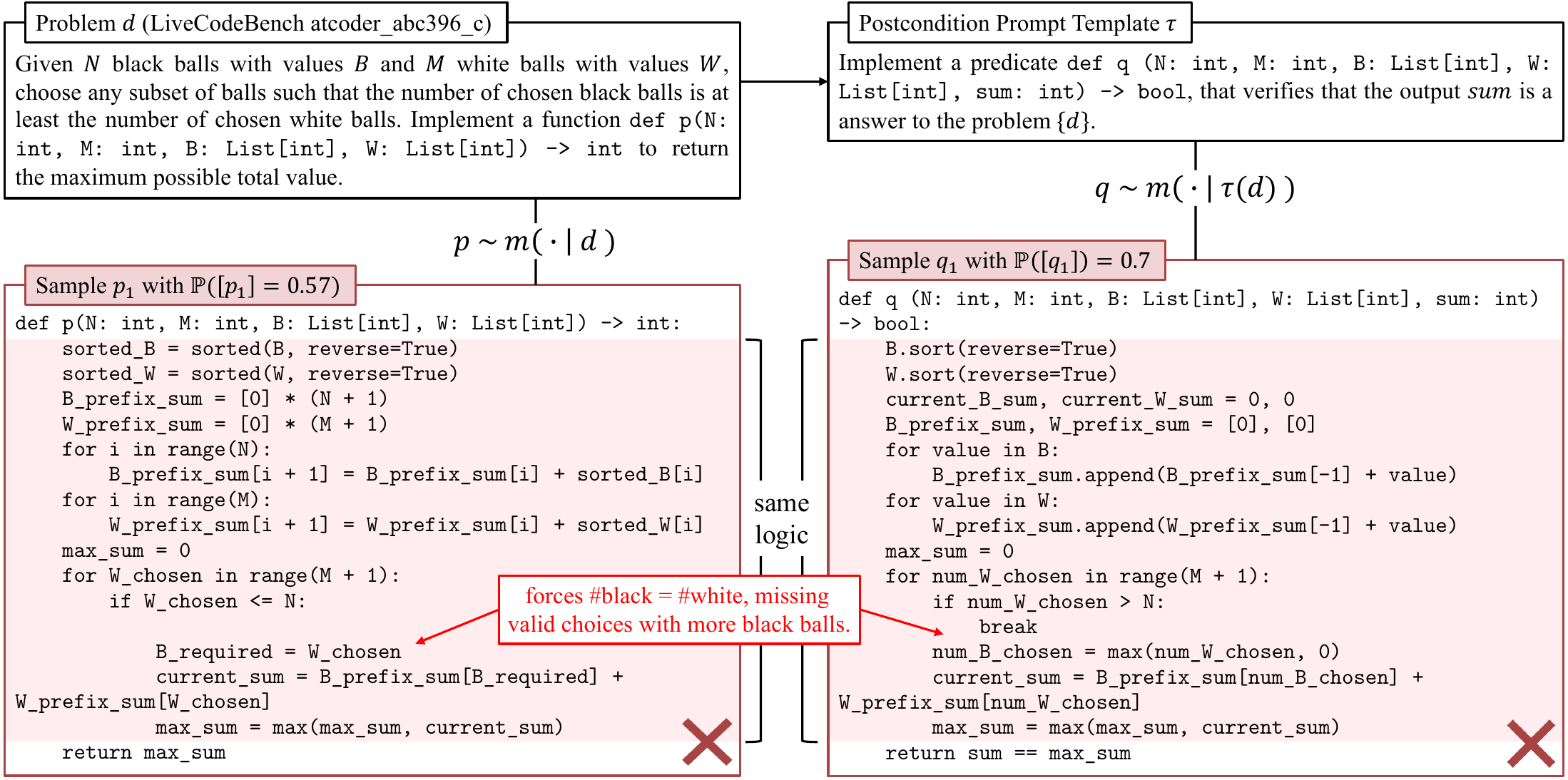}
  \caption{Hoare-style postconditions can degenerate into a functional reimplementation of the verified program. In this example, the generated postcondition $q_1$ follows the same algorithmic structure as sampled solution $p_1$; as a result, $q_1$ successfully verifies $p_1$ despite both being incorrect.\label{fig:postcondition}}
\end{figure}

Overpermissiveness is not the only failure mode. \Cref{fig:postcondition}
shows a LiveCodeBench task in which a sampled program and an independently
generated postcondition repeat the same semantic error. In this task, the constraint requires selecting at least as many black balls as white
balls, whereas $p_1$ always selects equal numbers of each. The postcondition $q_1$
reimplements the same flawed calculation; hence, it accepts the output of $p_1$ even
though neither program captures the intended semantics. Thus, specification can suffer from correlated errors in the same way as plurality.

The situation shown in \Cref{fig:problem} is not uncommon. \Cref{fig:prob_correct_dist} reports the estimated probability of sampling a correct solution from GPT-4o across tasks in LiveCodeBench~\cite{jain2024livecodebench}. The model fails to solve most tasks, and for a large fraction of solvable tasks, the probability of a correct solution is below 0.5. No existing method can both abstain from unsolvable tasks and reliably identify low-probability correct solutions on solvable ones, posing a significant challenge to the reliability of LLM-generated code.

To understand this limitation and find a path forward, we recast prior approaches within a unifying perspective: they sample auxiliary artifacts $w$ (such as other programs, tests, specifications, etc.)---which we refer to as the \emph{plausibility witnesses}---and check their agreement with sampled programs $p$. These agreement checks yield matching pairs that are then used to form consensus. This leads to the following fundamental yet overlooked question:

\begin{center}
\vspace{1mm}
\emph{What plausibility witnesses provide higher confidence in program correctness?}
\vspace{1mm}
\end{center}

\noindent More formally, it asks to identify the type of witnesses that maximize the probability that the program is correct under an agreement, $\mathds{P}(p\text{ is correct w.r.t. }d \mid \mathrm{agree}(p, w)),$ where $d$ is a problem description, and ``$\mathrm{agree}$'' may denote (i) semantic equivalence if $w$ is a program, i.e. another solution to $d$, (ii) test execution if $w$ is a test suite, or (iii) a verification procedure if $w$ is a formal specification.

\subsection{High-Confidence Plausibility Witnesses via Semantic Triangulation}

Semantic triangulation is the first approach designed from first principle to answer the above question. Specifically, it applies a dissociative transformation to the problem in order to change the surface representation or structure of the task sufficiently to decorrelate model errors, but does so in a way that maintains an exact, verifiable mapping between solutions before and after transformation. This is feasible in the programming domain, where problems have well-defined semantics and program behavior can be reasoned about with precision.

Let $D$ be a space of problem descriptions, and $P$ be a space of programs. We define:

\begin{framed}
    \begin{minipage}[b]{0.55\textwidth}
      \begin{definition}[Semantic Triangulation]\label{def:triangulation}
        A semantic triangulation $(\tau, \phi)$ consists of a dissociative problem transformation $\tau: D \to D$ and a predicate on pairs of programs $\phi: P\times P \to \{0,1\}$ that induces a bijection between semantic equivalence classes of programs such that, for any problem description $d\in D$, it maps correct solutions for $d$ to correct solutions for $\tau(d)$.
      \end{definition}
    \end{minipage}%
    \hfill
    \begin{minipage}[b]{0.4\textwidth}
      \centering
      {\scriptsize
      \begin{tikzpicture}[>=stealth, node distance=0.8cm]
        \node (D) at (0,2.5) {Problem description $d$};
        \node (p) at (-1.8,0) {Solution $p$};
        \node (pprime) at (1.8,0) {Solution $q$};
        \node (Dprime) at (0.9,1.2) {$d'$};
        \draw (D) -- (p) node[midway, left] {$p \sim m(\,\cdot \mid d)$};
        \draw (p) -- (pprime) node[midway, below] {Hyperproperty $\hl{\phi}$};
        \draw (D) -- (Dprime) node[midway, right] {$d' = \hl{\tau}(d)$};
        \draw (Dprime) -- (pprime) node[midway, right] {$q\sim m(\,\cdot \mid d')$};
      \end{tikzpicture}
      }
      \vspace{-5mm}
    \end{minipage}
  \end{framed}
  
\noindent where by inducing a bijection we mean that there exists a one-to-one correspondence between equivalence classes of programs such that $\phi$ holds exactly for those pairs $(p,q)$ where the class of $p$ is mapped to the class of $q$. Intuitively, this definition can be interpreted as follows:
\begin{itemize}
\item \emph{Error decorrelation through bijection}: A bijective mapping between equivalence classes guarantees that each error in solutions to the original problem maps to a unique error (one-to-one) in solutions to the transformed problem. Thus, the probability that two independently generated incorrect programs match exactly under a dissociative transformation is low.
\item \emph{Correctness-coupling}: The requirement that the bijection preserves correctness increases the probability of correctness under agreement via $\phi$ since in well-trained models, even in the presence of correlated errors, the correct class still carries a non-negligible signal — it is expected to be more likely than an arbitrarily chosen error class; otherwise the model would effectively output noise.
\end{itemize}

We theoretically justify this definition using a mathematical model tailored to LLM hallucinations. In particular, to capture LLM's lack of genuine understanding of the problem's semantics, which is supported by empirical evidence~\cite{bender2021dangers,shani2025tokens}, we posit that for any given problem, there exists a semantically different problem that the model cannot distinguish (Assumption~\ref{assume:parrot}). Furthermore, drawing on empirical findings of widespread error correlation~\cite{kim2025correlated}, we assume that the LLM produces solutions with imbalanced distributions over equivalence classes of errors (Assumption~\ref{assume:correlated}). Under these assumptions, and using the rearrangement inequality~\cite{hardy1952inequalities}, we prove that an appropriately designed triangulation strictly improves upon plurality voting (Proposition~\ref{prop:triangulation_vs_plurality} and Proposition~\ref{prop:triangulation_vs_plurality_easy}):
\begin{equation*}
\mathds{E}_{d\sim\mathrm{Unif}(D)} \bigl[ \mathds{P}(p\text{ is correct w.r.t. }d \mid \phi(p, q')) - \mathds{P}(p \text{ is correct w.r.t. }d \mid p \equiv q) \bigr] > 0,
\end{equation*}
where $p$ and $q$ are independently sampled solutions for a uniformly chosen problem, $q'$ is a sampled solution to the transformed problem, and $\equiv$ is semantic program equivalence.

\subsection{Inversion-Based Dissociative Problem Transformation}
\label{sec:enum_sinv_example}

Semantic triangulation requires that its problem transformation $\tau$ is dissociative, meaning it changes the task sufficiently that the model does not precisely associate it with the original. Since deriving such a transformation directly from the model weights is infeasible given the complexity of modern LLMs, we instead design it based on intuition and validate the design empirically (\Cref{sec:rq1}).



\begin{figure}[t]
  \centering
  {\relscale{0.75}\def\svgwidth{\textwidth}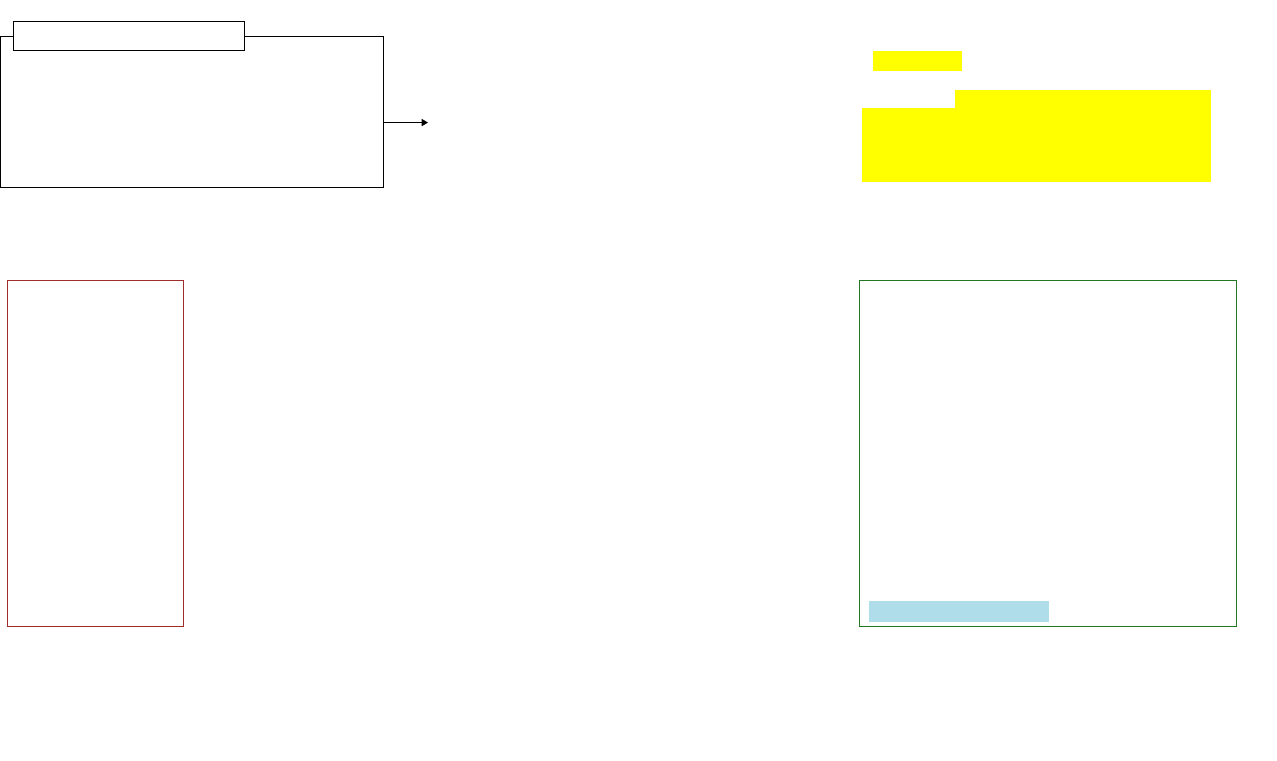}
  \caption{\toolName transforms the program $d$ into two derived problems: $d^{\prec}$, which enumerates all valid outputs for a given input, and $d^{\succ}$, which computes all second-argument values that produce a given output under a given first argument. It then triangulates candidate solutions to $d^{\prec}$ against those to $d^{\succ}$ via ENUM-SINV, identifying a reliable enumerator $e_1$, which is used to discard buggy $p_1$, and select correct $p_2$ via FWD-ENUM.\label{fig:enumsinv}}
  \vspace{-4mm}
\end{figure}%

A natural candidate for a dissociative transformation is problem inversion, since an inverse problem will likely require a different algorithm, e.g. a parsing algorithm is substantially different from that of a printer for the same format, making perfectly matching mistakes in their implementations less likely. However, the challenge of applying problem inversion in practice is that many problems are hard to invert. Formally, inverting a function requires it to be bijective, while most practical problems require implementing non-bijective functions (e.g., squaring is not bijective since both $2$ and $-2$ map to $4$). Although this can be addressed by returning all values that map to a given output, the approach has two limitations. First, a full inversion, e.g., enumerating the values of both $s$ and $t$ for the problem in \Cref{fig:problem}, is often intractable. Second, inexact problems further exacerbate the situation: since multiple valid solutions exist, any concrete solution will fail a bijection-inducing check with a correct inverse, which is a requirement of semantic triangulation (see \Cref{sec:inexact}).

Our tool \toolName identifies the low-probability correct solution $p_2$ and discards the dominant error $p_1$ shown in \Cref{fig:problem}. Initially, it uses an LLM to automatically transform the original problem $d$ into two derived problems shown in \Cref{fig:enumsinv}: an \emph{answer enumeration problem} $d^{\prec}$ and a \emph{set-valued partial inverse problem} $d^{\succ}$. Considering that $d$ is an inexact problem, an enumerator $e$ solving $d^{\prec}$ must, for a given input \mintinline{Python}{e("??", "a")}, return all valid substitutions satisfying the requirements of $d$: \mintinline{Python}{["aa", "ab", ...]}. The set-valued partial inverse $d^{\succ}$ flips the roles of the second argument $t$ and the return value, requiring the computation of all second-argument values that produce a given output under a given first argument. For example, given input \mintinline{Python}{q("??", "ab")}, an inverse $q$ solving $d^{\succ}$ must produce all $t$ that would require such a substitution: \mintinline{Python}{["a", "b", "ab"]}.

After transforming the problems, \toolName samples candidate solutions to both and proceeds as follows. First, it triangulates enumerators (solutions to $d^{\prec}$) against set-valued partial inverses (solutions to $d^{\succ}$) using the hyperproperty ENUM-SINV shown in \Cref{fig:enumsinv}. Concretely, for each candidate enumerator $e$ and each candidate inverse $q$, ENUM-SINV checks whether the two programs agree: for a sampled input, every output listed by $e$ must map back through $q$ to a set containing the original second argument, and vice versa. Since ENUM-SINV is bijection-inducing, a match can only occur when both $e$ and $q$ implement the exact same underlying input-output relation---just viewed from opposite directions. In our example, the enumerator $e_1$ and inverse $q_1$ pass this check, while other candidate pairs do not, because their independent errors are inconsistent under inversion. Second, it uses the identified enumerator $e_1$ to check solutions to the original problem using the hyperproperty FWD-ENUM. This hyperproperty simply checks the membership of the output of $p$ in the output of $e$. Since FWD-ENUM holds only for $p_2$, it is returned as the plausible solution, while the dominant error $p_1$ is discarded.

The success of this method stems from a combination of factors. First, the transformations are dissociative: as can be seen in \Cref{fig:problem} and \Cref{fig:enumsinv}, all problems require substantially different algorithms, making it non-trivial to produce perfectly matching bugs across them. Second, in contrast to the Hoare triples in \Cref{fig:problem}, ENUM-SINV is bijection-inducing, which we prove in Lean, meaning it can detect even minor inconsistencies between solutions. Thus, while each problem is solved correctly with low probability, only correct solutions match under the hyperproperties.
  
Crucially, all components of this triangulation scheme play distinct roles:

\begin{itemize}
\item A set-valued partial inverse problem $d^{\succ}$ and ENUM-SINV are needed because FWD-ENUM is not bijection-inducing. Like the Hoare triples in \Cref{fig:problem}, it permits enumerators that output overapproximate sets, resulting in spurious matches. As we show experimentally via an ablation, non-bijection-inducing properties do not confer high reliability.
\item An answer enumeration problem $d^{\prec}$ and FWD-ENUM are needed because solutions to the original problem cannot be directly triangulated against set-valued inverses: since the problem is inexact, a correct solution to $d$ does not map to a single correct solution to $d^{\succ}$, and therefore will fail any bijection-inducing, correctness-coupling property.
\end{itemize}

\section{Background, Notation and Basic Assumptions}
\label{sec:background}

We use $\mathds{P}(\cdot)$ to denote probability and $\mathds{E}[\cdot]$ for expectation. Conditional probability $\mathds{P}(A\mid B)$ represents the probability of $A$ given that $B$ holds. We use $d \sim \mathrm{Unif}(D)$ for uniform sampling. Let $D$ be a space of problem descriptions, $P$ be a space of programs. Let $m$ be a code generation model, and $m(\,\cdot\mid d)$ be a distribution of programs conditioned on the description $d$. We denote programs independently sampled from this distribution as $p, q,...\sim m(\,\cdot\mid d)$. The support of a distribution is the set of all programs with non-zero probability. Let $\equiv$ be the (partial) semantic program equivalence relation~\cite{ciobacua2016language}, $\mathcal{P}/{\equiv}$ be the set of all equivalence classes $[p_1]_{\equiv}, [p_2]_{\equiv}, ...$. Since semantic program equivalence is a classical relation, we omit $\equiv$ when denoting equivalence classes w.r.t this relation, e.g. $[p_1]$. Given $[p]$, we call $p$ its representative. We use $\mathrm{id}_A$ to denote the identity function on the domain $A$.

The plurality method selects a representative of the most probable equivalence class. Formally, if we sample $n>2$ programs $S\triangleq \{p_1, p_2, \dots, p_n\}$, that is, $p_1, p_2, \dots, p_n \sim m(\,\cdot \mid d),$ then plurality will select a program $p^\ast \triangleq \arg\max_{p_i \in S} | [p_i] \cap S|$.  This method captures the intuition that a well-trained model assigns higher probability to semantically correct solutions. Random sample consensus (RANSAC)~\cite{fischler1981random} is a classical algorithm designed to detect outliers, which was adapted to eliminate hallucinated program-test pairs in CodeT~\cite{chen2022codet}. In this context, this algorithm is applied to sets of programs $\{p_1, ..., p_n\}$ and witnesses $\{w_1, ..., w_n\}$ to maximize $|P_\mathrm{a}| \cdot |W_\mathrm{a}|$ where all program in $P_\mathrm{a}$ agree with all witnesses in $W_\mathrm{a}$. Importantly, plurality is a special case of RANSAC consensus, where equivalence between two programs directly corresponds to their agreement.

The modulus of an array $x$ indexed by $I$ is defined as $\|x\|_{\ell^2(I)}\triangleq\sqrt{\sum_{i\in I}x_i^2}.$

Since prompts are written in natural language, they are prone to imprecision and ambiguity, often allowing conflicting interpretations~\cite{jia2025automated}. We categorize causes of this uncertainty into two groups. We say that a problem description is \emph{ambiguous} if it admits unintended interpretations. For example, in the task ``find three distinct elements of a given list that sum to zero'', the term ``distinct'' may refer either to distinct indices or to distinct values, and it is likely to be a fault in the description. In contrast, a problem description is \emph{inexact} if it omits unimportant details. For example, the specification ``return a positive element of a given list'' allows one to choose a positive element arbitrarily, which is likely valid and even desirable. Indeed, Quan et al.~\cite{quan2025codeelobenchmarkingcompetitionlevelcode} found that even in programming contests, where descriptions are more precise than in typical user prompts, ``about 30\% of problems do not have unique correct outputs.'' We rely on the following assumption:

\begin{assumption}\label{assumption:inexact}
  All descriptions $D$ are unambiguous, but some are inexact.
\end{assumption}

\noindent This relaxes assumptions of previous works. For example, the uncertainty measure by Valentin et al.~\cite{valentin2025estimating} assumes that ``if two LLM-generated programs behave differently on the same input, at least one must be incorrect'', thus requiring prompts that are both unambiguous and exact.

Following standard practice in program verification, we model programs as mathematical functions (please see \Cref{sec:threats} for practical implications):

\begin{assumption}\label{assumption:pure}
  Relevant behavior of solutions to $D$ can be represented as mathematical functions $p: I \rightarrow O$, where $I$ is a space of program inputs, and $O$ is a space of program outputs. 
\end{assumption}

\begin{definition}[Semantic Program Equivalence]
  Programs $p$ and $q$ operating on the same domains are semantically equivalent, which is denoted as $p \equiv q$, iff $\forall i\in I.\ p(i) = q(i).$
\end{definition}

\noindent Given that (partial) program equivalence is undecidable, we estimate it by comparing programs' outputs on LLM-generated inputs provided that the programs terminate within a timeout.

We rely on the existence (without requiring explicit knowledge) of a semantic function $\llbracket \cdot \rrbracket : \mathrm{NL} \to 2^{I \times O}$ that for each problem description returns the required relation between inputs and outputs. For example, $\llbracket ``\text{multiply two numbers}" \rrbracket = \{ ((2, 2), 4), ((2, 3), 6), \cdots \}.$ Inexactness is defined as

\begin{definition}[Exactness]\label{def:exact}
  A problem description $d$ is \emph{exact} if and only if for any $i\in I$, $o, o'\in O$, $(i, o) \in \llbracket d \rrbracket \wedge (i, o') \in \llbracket d \rrbracket \Rightarrow o = o';$ Otherwise, it is \emph{inexact}.
\end{definition}

\begin{example}[Inexact Problem Description]
  \begin{align*}
    \llbracket ``\text{return}\text{ a positive element of a given set}" \rrbracket =\ \{\ &(\{1, -1, 3\}, {\color{red} 1}),\ (\{1, -1, 3\}, {\color{red} 3}), \cdots \}.
  \end{align*}
\end{example}

To simplify our definitions, we assume that all inputs are valid and all outputs are reachable:
\[
\forall i\in I\ \exists o\in O.\ (i, o) \in \llbracket d \rrbracket,\quad \forall o\in O\ \exists i\in I.\ (i, o) \in \llbracket d \rrbracket.
\]
In practice, identifying valid inputs is undecidable; \Cref{sec:implementation} discusses our practical mitigation.

\begin{definition}[Program Correctness w.r.t. Problem Description]\label{def:program_correctness}
  A program $p\in P$ is \emph{correct} w.r.t. a problem description $d\in D$ iff $\forall\ i\in I. (i, p(i)) \in \llbracket d \rrbracket,$ which is denoted as $p \vdash \llbracket d \rrbracket.$ It is \emph{incorrect} iff $\exists i\in I. (i, p(i)) \notin \llbracket d \rrbracket,$ which is denoted as $p \nvdash \llbracket d \rrbracket.$
\end{definition}

Hoare logic~\cite{hoare1969axiomatic} is a framework for specifying expected program behavior. Its central construct, the Hoare triple $\{\mathit{Pre}\}c\{\mathit{Post}\}$, expresses that, if the program $c$ starts in a state satisfying the precondition $\mathit{Pre}$, then---assuming termination---it will end in a state satisfying the postcondition $\mathit{Post}$. Pre- and postconditions are predicates over program states---mappings from variables to values. Since we assume that programs are pure functions and all inputs are valid, such specifications reduce to $\{\mathrm{True}\}o=c(i)\{\mathit{Post}(i,o)\}$, determined by postconditions $\mathit{Post}$.




\section{Theoretical Basis of Triangulation}
\label{sec:theory}

This section formally justifies the design of semantic triangulation (\Cref{def:triangulation}), along with providing a conceptual account of the underlying improvement. Plurality, test generation, and autoformalization can be unified under the concept of plausibility witnesses. Let $p\sim m(\,\cdot\mid d)$ be a program generated for a description $d$. Suppose there exists a procedure for sampling artifacts $w$ and a criterion ``$\mathrm{agree}$'' for assessing their agreement with $p$. We call such $w$ plausibility witnesses. Based on agreeing pairs $(p, w)$, algorithms such as RANSAC form consensus over the entire sample, and plurality voting is a special case of this algorithm. Thus, a fundamental question arises:

\begin{definition}[Confidence-Enhancing Plausibility Witness Problem]\label{def:witness_problem}
  For a sampled program $p$, what type of plausibility witness $w$ and the corresponding agreement criterion provide higher $\mathds{P}(p\vdash \llbracket d \rrbracket \mid \mathrm{agree}(p, w)),$ the probability that the program is correct under an agreement?
\end{definition}

\subsection{Semantic Triangulation Provably Enhances Confidence Under Agreement}

We prove that triangulation yields higher-confidence witnesses than plurality in a mathematical model of LLM hallucinations, which rigorously explains the mechanism of mitigating hallucinations, informing our practical implementation. To streamline the formal analysis, we adopt several simplifying conventions that
  do not restrict the practical applicability of our approach:
\begin{itemize}
\item to avoid modeling linguistic variations, all descriptions in $D$ are semantically distinct;
\item to facilitate equivalence checking, all programs $P$ operate on exactly the same domains;
\item there is a finite number of equivalence classes $[p_1], [p_2], \cdots, [p_n]$ in $P/\equiv$;
\item the existence of $\phi$ does not require a closed-form representation.
\end{itemize}

Plurality voting operates under the assumption that each problem has a single solution. Indeed, if there are multiple correct solutions to a problem, plurality may prefer high-probability errors even if the total probability mass of all correct solutions is higher. Since our theoretical justification involves a comparison with plurality, in this section we assume that each $d$ has exactly one solution. In the subsequent section, we show how triangulation alleviates this limitation in practice.

Intuitively, triangulation yields higher-confidence witnesses than plurality by mitigating correlated errors. This arises because LLMs, as ``stochastic parrots,'' lack precise semantic understanding and therefore fail to preserve errors under non-trivial transformations. To formalize this intuition, inspired by work on spurious correlations~\cite{ye2024spurious}, we assume that LLMs rely on shallow features preventing them from capturing subtle aspects of problem semantics. At the same time, when a model lacks the ability to discriminate between any problem instances, confidence improvement becomes infeasible. The following assumption reflects these intuitions:

\begin{assumption}[Stochastic Parrots]\label{assume:parrot}
  Let $m$ be a model that defines a distribution of programs conditioned on problem descriptions $m(p \mid d)$. We say that $m$ is a stochastic parrot if for any description $d$, there is a description $d'$ such that $m(\,\cdot \mid d) = m(\,\cdot \mid d')$ and $\llbracket d \rrbracket \neq \llbracket d' \rrbracket$, and a description $d''$ such that $m(\,\cdot \mid d) \neq m(\,\cdot \mid d'')$ and $\llbracket d \rrbracket \neq \llbracket d'' \rrbracket$.
\end{assumption}

\noindent If $m$ is a stochastic parrot, the equivalence relation $\simeq$ defined as $d \simeq d' \Leftrightarrow m(\,\cdot\mid d) = m(\,\cdot\mid d')$ partitions $D$ into at least two equivalence classes containing at least two elements each. We refer to this partitioning $D/\simeq$ as the \emph{hallucination pattern} of $m$.

To formalize the notion of correlated errors, we say that all programs in the same incorrect equivalence class have the same errors, while different classes have different errors. Then, correlated errors across samples for a given problem effectively mean that some equivalence classes are more likely than others. For simplicity, assume that all error classes have different probabilities:

\begin{assumption}[Correlated Errors]\label{assume:correlated}
Let $m$ be a model that defines a distribution of programs conditioned on problem descriptions $m(p \mid d)$. We say that its errors correlate if for each $d$, all classes of programs incorrect w.r.t. $d$ have different probabilities.
\end{assumption}

In our proof, we transform problems by exploiting the hallucination pattern of a stochastic parrot. We pose a requirement for triangulation that it must not make problems arbitrarily harder to solve, i.e., the probabilities of correct solutions for the original problem $d$ and the transformed problem $d'$ must differ from each other less than the probabilities of errors. In the first proposition, operating under a weak \Cref{assume:correlated} that merely establishes that the probabilities of errors are different, we assume that the probabilities of correct solutions to $d$ and $d'$ are the same\footnote{This assumption will be relaxed in \Cref{prop:triangulation_vs_plurality_easy}}:

\begin{figure}[t]
  \centering
  \hspace{8mm}{\relscale{0.75}\def\svgwidth{0.75\textwidth}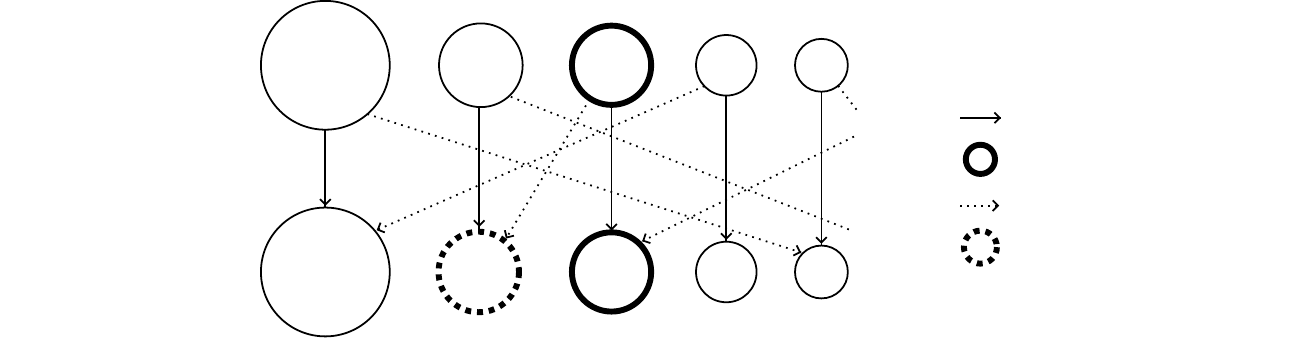}
  \caption{Plurality suffers from correlated errors. Triangulation rearranges the mapping of programs to their witnesses so that large classes of program errors are matched with small classes of witness errors and vice versa, which decreases the probability of matching bugs as per the rearrangement inequality.\label{fig:rearrangement}}
\end{figure}%

\begin{proposition}\label{prop:triangulation_vs_plurality}
Let a code generation model $m$ be a stochastic parrot with correlated errors that hallucinates on problems with equal probability of correct solutions. There exists semantic triangulation $(\tau, \phi)$ such that for $p, q\sim m(\,\cdot\mid d)$ and $q'\sim m(\,\cdot\mid \tau(d)),$
  \begin{equation}\label{eq:triangulation_vs_plurality}
    \mathds{E}_{d\sim\mathrm{Unif}(D)} \bigl[ \mathds{P}(p \vdash \llbracket d \rrbracket \mid \phi(p, q')) - \mathds{P}(p \vdash \llbracket d \rrbracket \mid p \equiv q) \bigr] > 0.
\end{equation}
\end{proposition}

\begin{proof}

We construct $\tau$ so that it exploits the stochastic parrot assumption. For each $[d]_\simeq$ in the $m$'s hallucination pattern, we define $\tau$ as a permutation on $[d]_\simeq$ that has no fixed-points.

Let $\phi \triangleq \{ (p, q) \mid d\in D, p \vdash \llbracket d \rrbracket, q \vdash \llbracket \tau(d) \rrbracket \}.$ Since $\tau$ is a permutation on semantically distinct problems, $\phi$ induces a permutation on equivalence classes $[p_1], [p_2], \cdots, [p_n]$. We represent this permutation on the indices $\{1, ..., n\}$ as $\sigma$. By construction, $\phi$ is correctness-coupling, i.e. if $[p_c]$ is correct w.r.t. $d$, then $[p_{\sigma(c)}]$ is correct w.r.t. $\tau(d)$. Therefore, $(\tau, \phi)$ is semantic triangulation. Apart from that, since $\tau$ has no fixed-points, therefore $\sigma$ has no fixed-points. The effect of this permutation $\sigma$, induced by $\phi$, is illustrated in \Cref{fig:rearrangement} using dashed arrows.

Let $\pi_i \triangleq \mathds{P}(p \in [p_i]),$ for $p \sim m(\,\cdot \mid d).$ We can rewrite $\mathds{P}(p \vdash \llbracket d \rrbracket \mid p \equiv q)$ and $\mathds{P}(p \vdash \llbracket d \rrbracket \mid \phi(p, q))$ as $\frac{\pi_c^2}{\sum\limits_{i}\pi_i^2},$ and $\frac{\pi_c \pi_{\sigma(c)}}{\sum\limits_{i}\pi_i \pi_{\sigma(i)}},$ respectively. Assume $C_d$ are the indices of correct solutions to problems in $[d]_\simeq$, and $B_d$ are the remaining indices of buggy solutions. Then, the expectation in \Cref{eq:triangulation_vs_plurality} for each class $[d]_\simeq$, which we refer to as $\mathds{E}_{d\sim\mathrm{Unif}([d]_\simeq)}(\Delta \mathds{P})$, is expressed as
\begin{equation}\label{eq:delta_elaborated}
\mathds{E}_{d\sim\mathrm{Unif}([d]_\simeq)} \Biggl[ \frac{\pi_c \pi_{\sigma(c)}\sum\limits_{i}\pi_i^2 - \pi_c^2 \sum\limits_{i}\pi_i \pi_{\sigma(i)}}{\sum\limits_{i}\pi_i^2\cdot\sum\limits_{i}\pi_i \pi_{\sigma(i)}} \Biggr] = \frac{\sum\limits_{c\in C_d} \pi_c \pi_{\sigma(c)}\cdot\sum\limits_{i}\pi_i^2 - \sum\limits_{c\in C_d} \pi_c^2 \cdot \sum\limits_{i}\pi_i \pi_{\sigma(i)}}{| [d]_\simeq | \sum\limits_{i}\pi_i^2\cdot\sum\limits_{i}\pi_i \pi_{\sigma(i)}}
\end{equation}

The reduction of correlated errors through triangulation can be explained by the rearrangement inequality~\cite{hardy1952inequalities}. This inequality captures the principle that matching large terms with large terms and small terms with small terms maximizes the sum of products. Specifically, for any real numbers $\pi_1, \pi_2, \dots, \pi_n$ and any permutation $\sigma$, $\sum_{i=1}^n \pi_i^2 \ \ge\ \sum_{i=1}^n \pi_i \pi_{\sigma(i)}$. Moreover, if all $\pi_i$ are different and $\sigma$ has no fixed points, the inequality is guaranteed to be strict. See Appendix~\ref{app:strict-rearrangement} for a proof of the strictness condition.

In our case, $\sigma$ does not preserve any index by construction, and all $\pi_b$ for $b\in B_d$ are distinct due to \Cref{assume:correlated}. Apart from that, since $\sigma$ is a permutation on $[1..n]$ and on $C_d$, therefore it is also a permutation on $B_d$, which contains at least two elements due to \Cref{assume:parrot}. Therefore, for the top part of \Cref{eq:delta_elaborated}, after canceling repeated terms, the following holds:
\begin{equation*}
\sum\limits_{c\in C_d} \pi_c \pi_{\sigma(c)}\cdot \hl{\sum\limits_{b \in B_d}\pi_b^2} - \sum\limits_{c\in C_d} \pi_c^2 \cdot \sum\limits_{b \in B_d}\pi_b \pi_{\sigma(b)} > \sum\limits_{c\in C_d} \pi_c \pi_{\sigma(c)}\cdot\hl{\sum\limits_{b \in B_d}\pi_b \pi_{\sigma(b)}} - \sum\limits_{c\in C_d}\pi_c^2\cdot\sum\limits_{b\in B_d}\pi_b \pi_{\sigma(b)}
\end{equation*}
Since we assumed that all correct solutions to problems in $[d]_\simeq$ have equal probabilities, the RHS of this inequality is zero, therefore $\mathds{E}_{d\sim\mathrm{Unif}(D)}(\Delta \mathds{P}) > 0$.
\end{proof}


The preceding proof constructs $\tau$ such that the transformed problem instances are equally difficult to solve as the originals, where ``difficulty'' is defined as the probability of sampling a correct solution. This contrasts with our empirical observations, in which non-semantics-preserving transformations may increase task difficulty. To reconcile this, we introduce the dissociative hallucination pattern assumption that enables us to construct $\tau$ in such a way that it does not make problems significantly harder, while still preventing the model from precisely associating their semantics:

\begin{assumption}[Dissociative Hallucination Pattern]\label{assume:easy_pattern}
  Let $m$ be a stochastic parrot, $D/\simeq$ be its hallucination pattern, $[p_i] \in P/\equiv$ be all program equivalence classes, $c\in C_d$ be all indexes such that $[p_c]$ is correct w.r.t. at least one problem in $[d]_\simeq$, $B_d$ be the remaining indexes, and $\pi_i \triangleq \mathds{P}(p \in [p_i]).$ We say that $D/\simeq$ is dissociative if there exists $\tau$ defined as a permutation on each $[d]_\simeq$, with $\sigma$ being the corresponding permutation on the indexes of a correct solution induced by $\tau$, that the normalized difference between probabilities of correct solutions for the original programs and correct solutions for the transformed problems is lower than difference between probabilities of corresponding errors:
\begin{equation}\label{eq:dissociative-pattern}
 \sum_{c\in C_d}\left(\frac{\pi_c-\pi_{\sigma(c)}}{\|\pi\|_{\ell^2(C_d)}}\right)^2
<\sum_{b\in B_d}\left(\frac{\pi_b-\pi_{\sigma(b)}}{\|\pi\|_{\ell^2(B_d)}}\right)^2
\end{equation}
\end{assumption}

\noindent The justification of this assumption is that if a well-trained model does not distinguish between two problems and produces high-probability bugs, we can still expect that the probabilities of correct solutions to these problems in the distribution should not vary as much as probabilities of errors. Under this assumption, we can prove a more general statement:

\begin{proposition}\label{prop:triangulation_vs_plurality_easy}
Let the model $m$ be a stochastic parrot with a dissociative hallucination pattern. There exists semantic triangulation $(\tau, \phi)$ such that for $p, q\sim m(\,\cdot\mid d)$ and $q'\sim m(\,\cdot\mid \tau(d)),$
  \begin{equation*}
    \mathds{E}_{d\sim\mathrm{Unif}(D)} \bigl[ \mathds{P}(p \vdash \llbracket d \rrbracket \mid \phi(p, q')) - \mathds{P}(p \vdash \llbracket d \rrbracket \mid p \equiv q) \bigr] > 0.
\end{equation*}
\end{proposition}

\begin{proof}
Continuing the derivation for $\mathds{E}_{d\sim\mathrm{Unif}([d]_\simeq)}(\Delta \mathds{P})$ in \Cref{eq:delta_elaborated} leads to the equation in \Cref{assume:easy_pattern}, as elaborated in \Cref{sec:proof_prop_general}.
\end{proof}

\section{Practical Design of Triangulation}
\label{sec:practical}

Our practical design of semantic triangulation methods is based on the insights of our theoretical analysis (\Cref{assume:easy_pattern}) that require that the transformed problems are not significantly harder to solve, while it is hard to exactly match errors in their implementations. In particular, we transform a problem into one that requires a different algorithm while still being ``natural'' for an LLM to solve. Our transformations involve algorithm inversion/partial inversion, answer enumeration and problem decomposition. Specifically, we introduce three inversion-based methods consisting of problem transformations and corresponding hyperproperties: FWD-INV, FWD-SINV, and ENUM-SINV, and a meta triangulation method STREAM for stateless stream-processing programs. All transformations are performed automatically by prompting an LLM (see, e.g., \Cref{fig:sinv_transformation_prompt}).

To specify the intended effect of transformations, we use an imaginary procedure $\llparenthesis\ \cdot\ \rrparenthesis$ that, given a binary relation between inputs and outputs, converts it into natural language text.

\begin{definition}[Inverse Problem]
  Let $d$ be a problem such that $\llbracket d \rrbracket\subset I\times O$. The corresponding inverse problem is $d^{-1} \triangleq \llparenthesis\ \{(o, i) \mid (i, o) \in \llbracket d \rrbracket\}\ \rrparenthesis.$
\end{definition}

\begin{proposition}[Full FWD-INV]\label{prop:full_fwd_inv}
  Let $\tau$ transform each problem $d$ into $d^{-1}$, $p$ is a sample for $d$, $q$ is for $d^{-1}$, and $\phi\triangleq q \circ p \equiv \mathrm{id_I} \wedge \forall o, o'\in O.\ o \neq o' \Rightarrow q(o) \neq q(o')$. If $\tau$ is dissociative, then $(\tau, \phi)$ is semantic triangulation (\Cref{def:triangulation}). We refer to this triangulation method as Full FWD-INV.
\end{proposition}

\begin{proof}
  A proof that $\phi$ is bijection-inducing and correctness-coupling is mechanized in Lean.
\end{proof}

\begin{figure}[t]
    \centering
        {\scriptsize
      \begin{tcolorbox}[mybox,title={SINV problem transformation prompt template}]
      \begin{minted}[escapeinside=||]{text}
You are given a programming problem that requires implementing the function {ORIG_SIGN}. Rewrite this problem
so that it instead requires implementing the set-valued inverse function {SINV_SIGN}. Given the desired output
value {NEW_ARG} (corresponding to the original function's return value), the new function should return an
exhaustive list of values for the parameter {INV_ARG} such that if the original function were called with
any of these values (and the other parameters unchanged), it would produce {NEW_ARG} as the result.
      \end{minted}
      \end{tcolorbox}
    }
    \vspace{-4mm}
\caption{\label{fig:sinv_transformation_prompt}A fragment of prompt template used for set-valued inverse problem transformation. \texttt{ORIG\_SIGN} is the typed signature of the original function, \texttt{INV\_SIGN} is the signature of the inverse function, \texttt{NEW\_ARG} is the argument of the inverse function representing the output of the original function, \texttt{INV\_ARG} is the argument for inversion. The argument to invert (\texttt{INV\_ARG}) is chosen by an LLM using a separate prompt.}
\vspace{-3mm}
\end{figure}

Inversion may exploit dissociative hallucination patterns, because forward and inverse problems often require different algorithms. However, full invertion is often infeasible. Consider the problem ``given a graph and two of its nodes, compute the length of the shortest path between them''. A full inversion requires producing a pair consisting of a graph and two nodes whose shortest-path distance matches a given value---such outputs would likely fail to align with the behavior of a forward program, or else require exhaustively enumerating all such pairs, which is infeasible.

To alleviate this problem, we apply a more natural and practical partial inversion, that is inverting an algorithm w.r.t. a single parameter. For instance, ``given a graph and a target distance, enumerate all pairs of nodes whose shortest-path length is exactly that distance''. This formulation is easier for an LLM to handle, while still satisfying the properties required for semantic triangulation:

\begin{definition}[Partial Inverse Problem]
  Let $d$ be a problem such that $\llbracket d \rrbracket\subset (I_1 \times I_2) \times O$. The corresponding partial inverse problem w.r.t. the first argument is $$d^{-1}_1 \triangleq \llparenthesis\ \{((o, i_2), i_1) \mid ((i_1, i_2), o) \in \llbracket d \rrbracket\}\ \rrparenthesis.$$
\end{definition}

\noindent The inversion with respect to other arguments is defined accordingly. In practice, we prompt an LLM to decide which argument to use for inverting the problem.

\begin{proposition}[Partial FWD-INV]\label{prop:partial_fwd_inv}
  Let $\tau$ transform each problem $d$ into $d^{-1}_1$, and $$\phi\triangleq \forall (i_1, i_2)\in I.\ q(p(i_1, i_2), i_2) = i_1 \wedge\ \forall (o, i_2), (o', \_)\in O'.\ o \neq o' \Rightarrow q(o, i_2) \neq q(o', i_2),$$ where $p$ is a sample for $d$, $q$ is for $d^{-1}_1$ and $O'$ refers to all inputs to $q$. If $\tau$ is dissociative, then $(\tau, \phi)$ is semantic triangulation. We refer to this triangulation method as Partial FWD-INV.
\end{proposition}

\begin{proof}
  A proof that $\phi$ is bijection-inducing and correctness-coupling is mechanized in Lean.
\end{proof}

Although the problem inversion defined above is likely to change the solution algorithms, while keeping the problem tractable for an LLM, its key limitation is that for the hyperproperty to faithfully hold, i.e., hold on pairs of correct solution-witness pairs, the problem description must permit bijective functions as valid solutions. This, however, often does not hold in practice. For example, for the problem ``compute the square of a given number'', a correct solution maps both $-i$ and $i$ to $i^2$, implying that only incorrect forward and inverse implementations will match.

To alleviate this challenge, we employ set-valued inverses that return the exhaustive list of values for a given input argument that would result in a given output:

\begin{definition}[Set-valued Inverse Problem]
  Let $d$ be a problem such that $\llbracket d \rrbracket\subset I\times O$. The corresponding set-valued inverse problem is $$d^{\succ} \triangleq \llparenthesis\ \bigl\{(o, \{i \mid (i, o) \in \llbracket d \rrbracket \}) \mid (\_, o) \in \llbracket d \rrbracket \bigr\}\ \rrparenthesis.$$
\end{definition}

\begin{proposition}[Full FWD-SINV]\label{prop:full_fwd_sinv}
  Let $\tau$ transform each problem $d$ into $d^{\succ}$, and $$\phi\triangleq \forall i \in I.\ \underbrace{i \in q(p(i))}_{L_1} \wedge \underbrace{\forall i' \in q(p(i)).\ p(i') = p(i)}_{L_2} \wedge\ \underbrace{\forall o, o'\in O.\ o \neq o' \Rightarrow q(o) \cap q(o') = \emptyset}_{L_3},$$ where $p$ is a sample for $d$ and $q$ is for $d^{\succ}$. If $\tau$ is dissociative, then $(\tau, \phi)$ is semantic triangulation. We refer to this triangulation method as Full FWD-SINV.
\end{proposition}

\begin{proof}
  A proof that $\phi$ is bijection-inducing and correctness-coupling is mechanized in Lean.
\end{proof}


As with FWD-INV, we define a partial version of FWD-SINV:

\begin{definition}[Partial Set-valued Inverse Problem]
  Let $d$ be a problem such that $\llbracket d \rrbracket\subset (I_1 \times I_2) \times O$. The corresponding partial set-valued inverse problem w.r.t. the first argument is $$d^{\succ}_1 \triangleq \llparenthesis\ \bigl\{\bigl((o, i_2), \{i_1 \mid ((i_1, i_2), o) \in \llbracket d \rrbracket \}\bigr) \mid ((\_, i_2), o) \in \llbracket d \rrbracket \bigr\}\ \rrparenthesis.$$
\end{definition}

\begin{proposition}[Partial FWD-SINV]\label{prop:partial_fwd_sinv}
  Let $\tau$ transform each problem $d$ into $d^{\succ}_1$, and
\begin{align*}
\phi\ \triangleq&\ \forall (i_1, i_2) \in I.\ i_1 \in q(p(i_1, i_2), i_2) \wedge \forall i_1' \in q(p(i_1, i_2), i_2).\ p(i_1', i_2) = p(i_1, i_2)\\
\wedge&\ \forall (o, i_2), (o', \_)\in O'.\ o \neq o' \Rightarrow q(o, i_2) \cap q(o', i_2) = \emptyset,
\end{align*}
where $p$ is a sample for $d$, $q$ is for $d^{\succ}_1$ and $O'$ refers to all inputs to $q$. If $\tau$ is dissociative, then $(\tau, \phi)$ is semantic triangulation. We refer to this triangulation method as Partial FWD-SINV.
\end{proposition}

\begin{proof}
  A proof that $\phi$ is bijection-inducing and correctness-coupling is mechanized in Lean.
\end{proof}

Automating such a transformation with an LLM is straightforward, as it often reduces to a simple linguistic reformulation and does not require the reasoning to actually solve the problem. A fragment of the prompt used for Partial FWD-SINV transformation is given in \Cref{fig:sinv_transformation_prompt}.

\subsection{Triangulation for Inexact Problems}
\label{sec:inexact}

\begin{figure}[t]
    \centering
    \begin{subfigure}[t]{0.32\textwidth}
    {\scriptsize
      \begin{tcolorbox}[mybox,title={Original problem $d$}]
Implement \mintinline{python}{p(i: int) -> int} that returns an integer strictly greater than i by at most 2.
      \end{tcolorbox}
      \begin{tcolorbox}[mybox,title=Sample $p$]
        \begin{minted}[escapeinside=||]{python}
def p(i: int) -> int:
  return i+1
        \end{minted}
    \end{tcolorbox}
    }
    \end{subfigure}%
    \hfill
\begin{subfigure}[t]{0.32\textwidth}
    {\scriptsize
      \begin{tcolorbox}[mybox,title={Inverse problem $d^{-1}$}]
Implement \mintinline{python}{q(i: int) -> int} that returns an integer strictly smaller than i by at most 2.
      \end{tcolorbox}
      \begin{tcolorbox}[mybox,title=Sample $q$]
        \begin{minted}[escapeinside=||]{python}
def q(i: int) -> int:
  return i-2
        \end{minted}
    \end{tcolorbox}
    }
\end{subfigure}
    \hfill
\begin{subfigure}[t]{0.32\textwidth}
    {\scriptsize
      \begin{tcolorbox}[mybox,title={Set-valued inverse problem $d^{\succ}$}]
Implement \mintinline{python}{q(i: int) -> list[int]} that returns all integers strictly smaller than i by at most 2.
      \end{tcolorbox}
      \begin{tcolorbox}[mybox,title=Sample $q'$]
        \begin{minted}[escapeinside=||]{python}
def q(i: int) -> list[int]:
  return [i-1, i-2]
        \end{minted}
    \end{tcolorbox}
    }
\end{subfigure}
\vspace{-3mm}
\caption{\label{fig:inexact_challenge}It is challenging to achieve consensus for problems that permit multiple non-equivalent solutions.}
\end{figure}

Inexact problems permit multiple non-equivalent solutions (\Cref{def:exact}), e.g., the task ``generate a script that extracts all customers from a CSV'' can be interpreted as saving customer data to a file or printing it to stdout, returning it in sorted order or not, and deciding which specific information to extract. Some details are likely irrelevant to the user's goal. We speculate that in practice, most code generation prompts are inexact. In contrast, to the best of our knowledge, existing techniques to establish consensus or control uncertainty, sometimes explicitly~\cite{valentin2025estimating}, assume that only one solution is correct, since forming a consensus under inexact problem descriptions is challenging.

\Cref{fig:inexact_challenge} illustrates this challenge. The original problem admits at least two non‑equivalent solutions: $i \mapsto i+1$ and $i \mapsto i+2$ (for simplicity, we omit solutions whose logic varies depending on the input). This undermines plurality, as the vote may select a high‑probability error even when correct solutions have greater total probability mass. This also hinders the above triangulation methods. For example, composing the inverse program $q$ with the forward program $p$ in \Cref{fig:inexact_challenge} does not yield an identity function as required by FWD-INV, despite both programs being correct individually. Similarly, in the context of FWD-SINV, the clause $\forall i' \in q'(p(i)).\ p(i') = p(i)$ in \Cref{prop:full_fwd_sinv} fails, as for $i=1$, $q'(p(i))$ is $\{0, 1\}$, but $p(0) \neq p(1)$.

Triangulation is the first to alleviate this problem; it generates a program to exhaustively enumerate valid answers for each input, and triangulate this enumerator against a set-valued inverse.

\begin{definition}[Answer Enumeration Problem]
  Let $d$ be a problem such that $\llbracket d \rrbracket\subset I\times O$. The corresponding answer enumeration problem is $d^{\prec} \triangleq \llparenthesis\ \bigl\{(i, \{o \mid (i, o) \in \llbracket d \rrbracket \}) \mid (i, \_) \in \llbracket d \rrbracket \bigr\}\ \rrparenthesis.$
\end{definition}

An enumerator is inefficient as a solution to the original problem; to address it, we sample solutions to three problems: the original problem $d$, the enumeration problem $d^{\prec}$, and the set-valued inverse problem $d^{\succ}$. After triangulating a solution to $d^{\prec}$ with a solution to $d^{\succ}$, we triangulate a solution to $d$ with $d^{\prec}$. Note that triangulating with $d^{\succ}$ cannot be omitted, since triangulating only $d$ with $d^{\prec}$ lacks the bijective property of semantic triangulation and therefore does not sufficiently decorrelate errors. We call such a triangulation scheme a \emph{cascade}.

\begin{proposition}[Full ENUM-SINV]\label{prop:full_enum_sinv}
  Let $\tau_1$ transform each problem $d$ into $d^{\prec}$, and $\tau_2$ transform each problem $d$ into $d^{\succ}$. Let a hyperproperty $\phi$ defined as
\[
\forall i \in I.\ \forall o \in p(i).\ i \in q(o) \wedge \forall o \in O.\ \forall i \in q(o).\ o \in p(i)
\]
where $p$ is a sample for $d^{\prec}$ and $q$ is for $d^{\succ}$. If $\tau_2 \circ \tau_1^{-1}$ is dissociative, then $(\tau_2 \circ \tau_1^{-1}, \phi)$ is semantic triangulation. We refer to this triangulation method as Full ENUM-SINV cascade.
\end{proposition}

\begin{proof}
  A proof that $\phi$ is bijection-inducing and correctness-coupling is mechanized in Lean.
\end{proof}

To illustrate this hyperproperty, refer again to \Cref{fig:inexact_challenge}. Assume an incorrect enumerator $p' \triangleq i \mapsto [i+1]$ misses $i+2$ in its output. Although $L_1$ will hold for any input to $p'$, the clause $L_2$ will fail. For example, for $o=2$, $q'(o)$ will return $\{0, 1\}$, but $2\notin p'(0)$. At the second phase, solutions $p$ to $d$ are checked against triangulated enumerators $p''$ via FWD-ENUM: $\forall i\in I.\ p(i) \in p''(i).$

Similar to previous triangulation methods, we define a partial inversion variant of ENUM-SINV:

\begin{proposition}[Partial ENUM-SINV]\label{prop:partial_enum_sinv}
  Let $\tau_1$ transform each problem $d$ into $d^{\prec}$, and $\tau_2$ transform each problem $d$ into $d^{\succ}_1$. Let the hyperproperty $\phi$ be defined as
\[
\forall (i_1, i_2) \in I.\ \forall o \in p(i_1, i_2).\ i_1 \in q(o, i_2) \wedge \forall (o, i_2) \in O'.\ \forall i_1 \in q(o, i_2).\ o \in p(i_1, i_2)
\]
where $p$ is a sample for $d^{\prec}$, $q$ is for $d^{\succ}_1$, and $O'$ refers to all inputs to $q$. If $\tau_2 \circ \tau_1^{-1}$ is dissociative, then $(\tau_2 \circ \tau_1^{-1}, \phi)$ is semantic triangulation. We refer to this method as Partial ENUM-SINV cascade.
\end{proposition}

\begin{proof}
  A proof that $\phi$ is bijection-inducing and correctness-coupling is mechanized in Lean.
\end{proof}


\begin{figure}[t]
\centering
\begin{subfigure}[c]{0.45\textwidth}
\centering
\scalebox{0.7}{
\begin{tabular}{cccc}
\toprule
\diagbox{GT}{R} & \makecell{Correctly \\ selected} & \makecell{Incorrectly \\ selected} & Abstained \\
\midrule
No abstention & \cellcolor{lightgreen}\( N_1 \) & \cellcolor{lightred}\( N_2 \) & \cellcolor{lightred}\( N_3 \) \\
Abstention & & \cellcolor{lightred}\( N_4 \) & \cellcolor{lightgreen}\( N_5 \) \\
\bottomrule
\end{tabular}
}
\caption{{\footnotesize Abstention confusion matrix. ``GT'' is the ground-truth, where ``Abstention'' indicates tasks with no correct solutions in the sample; ``R'' is the sample consensus response. When GT is abstention, the response can only be either ``Abstained'' or ``Incorrectly selected''.}}
\label{fig:sub-table}
\end{subfigure}
\begin{subfigure}[c]{0.5\textwidth}
\centering
\footnotesize
\begin{align*}
\text{Reliable Accuracy} &= \frac{N_1}{N_1 + N_2 + N_4} \\
\text{Overall Accuracy} &= \frac{N_1 + N_5}{N_1 + N_2 + N_3 + N_4 + N_5} \\
\text{Abstention Rate} &= \frac{N_3 + N_5}{N_1 + N_2 + N_3 + N_4 + N_5} \\
\end{align*}
\vspace{-2.3\baselineskip}
\begin{gather*}
\text{Precision}_\mathrm{abs} = \frac{N_5}{N_3 + N_5}, \quad
\text{Recall}_\mathrm{abs} = \frac{N_5}{N_2 + N_4 + N_5} \\
\text{Abstention F1-score} =  \frac{2 \cdot\text{Precision}_\mathrm{abs} \cdot \text{Recall}_\mathrm{abs}}{\text{Precision}_\mathrm{abs} + \text{Recall}_\mathrm{abs}}
\end{gather*}
\normalsize
\label{fig:sub-formulas}
\end{subfigure}
\vspace{-\baselineskip}
\caption{Summary of abstention evaluation metrics. $N_i$ are decision counts across evaluation dataset that are classified along two dimensions GT and R. For example, $N_2$ is the number of cases where the correct solution exists but the method selects a wrong one. \label{fig:abstention_matrix}}
\vspace{-4mm}
\end{figure}

\subsection{Implementation}
\label{sec:implementation}

We implemented triangulation in a tool called \toolName for Python programs. To form consensus over the entire sample from individual agreements, we apply RANSAC in a manner similar to CodeT. For stateless stream-processing problems, \toolName additionally applies the domain-specific decomposition method STREAM; its construction and theoretical justification are available in Appendix~\ref{app:problem-decomposition}. In addition, we implemented autoformalization into Hoare Logic–style postconditions, expressed as boolean Python functions to avoid limitations of low-resource specification languages.

Formal verification of Python code is hampered by the gap between capabilities of SOTA LLM-based code generation and automated formal verification. Even on simpler benchmarks like HumanEval~\cite{chen2021evaluating}---where SOTA LLMs nearly solve all tasks---these LLMs automate verification of correct solutions with the SOTA Python verifier Nagini~\cite{eilers2018nagini} in only 15\% of cases~\cite{shefer2025can}. To obtain enough signal on harder benchmarks, we approximate property validity by testing properties on LLM-generated inputs. To generate inputs, we use the similar method as EvalPlus~\cite{liu2023your} to craft prompts that sample until each problem reaches high coverage, but instead of supplying a ground-truth solution, which is not available in practice, we provide the problem description in the input generation prompt. Also, we use the same input sets for all baselines to ensure fairness. 

\toolName implements practical mitigations to two fundamental problems: (1) in practice, many inputs are invalid, e.g., negative numbers for Fibonacci, and (2) some transformed problems, e.g., a set-valued inverse, may require enumerating infinite or intractable sets. First, our code-generation prompt instructs programs to validate inputs and raise \texttt{ValueError} on invalid cases following common Python practice. Second, this prompt requests enumerating only a tractable subset when the answer set is infinite or too large, and marking the output accordingly.

Checking hyperproperties by execution must handle two sources of uncertainty: the valid-input domain may be unknown, and transformed programs may compute infinite or intractable sets. We therefore use an extended evaluation semantics with three special outcomes. A crash or timeout yields the demonic value \(\mathcal{D}\) and rejects the check; an invalid input yields \(\mathcal{U}\); and an unresolvable check---for example, membership outside a marked finite subset of an infinite set---yields the angelic value \(\mathcal{A}\). The latter is inconclusive rather than positive evidence. A universally quantified hyperproperty is accepted only when every checked branch is either \(\mathrm{True}\) or \(\mathcal{A}\), and fewer than a threshold \(T\) of its branches are angelic. Thus, the checker tolerates a bounded amount of uncertainty without certifying a property whose apparent validity rests mainly on invalid inputs or incomplete enumeration. The complete operational semantics and an example derivation are given in \Cref{app:property-checking}.

The same semantics governs differential testing for semantic equivalence between sampled programs: an invalid-input failure (\(\mathcal{U}\)) is tolerated angelically, so a pair is not deemed inequivalent merely because one program rejects an input, whereas other crashes and timeouts are demonic. In our experiments, the angelic treatment of failing executions was exercised in 42.4\% of pairwise equivalence checks; it does not, however, inflate equivalence decisions: checks involving a crashing or failing variant concluded equivalence in only 19.1\% of cases, compared to 32.6\% across all checks, because programs that fail on some inputs are generally more likely to be incorrect and thus tend to disagree with other programs on further inputs.


\section{Evaluation}
\label{sec:evaluation}

Our experiments address the following research questions:
\begin{description}
\item[RQ1:] Does triangulation yield higher-confidence plausibility witnesses than baselines?
\item[RQ2:] Does triangulation enable better program selection and abstention than baselines?
\item[RQ3:] Does triangulation handle inexact tasks that permit non-equivalent correct solutions?
\item[RQ4:] (Ablation) What aspects of triangulation contribute to its performance?
\item[RQ5:] What are the causes of failures of semantic triangulation?
\item[RQ6:] What are the time and token cost of semantic triangulation in comparison with baselines?
\item[RQ7:] Do the empirical results support the theoretical assumptions for semantic triangulation?
\end{description}

We selected three SOTA LLMs from different families, GPT-4o~\cite{hurst2024gpt}, a popular closed-weight model, DeepSeek-V3~\cite{liu2024deepseek}, a popular open-weight model, and Gemini-2.5-Flash, a reasoning model.

\begin{table}[t]
  \centering
  \caption{\emph{Agreements methods} check agreements with plausibility witnesses; \emph{Dissoc.}~= dissociative (witnesses likely require fundamentally different algorithm), \emph{Biject.}~= bijection-inducing. \emph{Consensus methods} form selection/abstention decisions by aggregating agreements across the entire sample.}
  \vspace{-1mm}
  \label{tab:method-properties}
  \footnotesize
  \setlength{\tabcolsep}{5pt}
  \begin{tabular}{p{7cm}ccl}
  \toprule
  \multirow{2}{*}{\textbf{Agreement method for $\mathrm{agree}(p, w)$}} &
  \multicolumn{2}{c}{\textbf{Properties}} & \multirow{2}{*}{\textbf{Consensus method}} \\
  \cmidrule(lr){2-3}
   & \textbf{Dissoc.} & \textbf{Biject.} & \\
  \midrule
  \multirow{2}{=}{\textbf{Another solution}: $w$ is another sampled solution to the original problem; $\mathit{agree}(p,w)$ is semantic equivalence}          & $\times$ & $\checkmark$ & \textbf{Plurality}          \\
                 & $\times$ & $\checkmark$ & \textbf{Majority0.5}
  \\
  \textbf{Tests}: $w$ is a generated test; $\mathit{agree}(p,w)$ means $p$ passes $w$        & $-$ & $\times$     & \textbf{CodeT}~\cite{chen2022codet}: Tests + RANSAC  \\
  \textbf{Postcondition}: $w$ is an LLM-generated postcondition; $\mathit{agree}(p,w)$ means $p$ satisfies $w$         & $\checkmark$ & $\times$     & \textbf{Postcondition} + RANSAC   \\
  \textbf{Syntactic}: metamorphic testing via paraphrasing  & $\times$ & $\checkmark$ & \textbf{Syntactic} + RANSAC  \\
  \textbf{OffByOne}: metamorphic testing via off-by-one  & $\times$ & $\checkmark$ & \textbf{OffByOne} + RANSAC \\
  \midrule
  Is triangulated via \textbf{FWD-INV}    & $\checkmark$ & $\checkmark$ &
  \multirow{3}{*}{\textbf{\toolName}} \\
  Is triangulated via \textbf{FWD-SINV}       & $\checkmark$ & $\checkmark$ & \\
  Is triangulated via \textbf{ENUM-SINV} & $\checkmark$ & $\checkmark$ & \\
  \bottomrule
  \end{tabular}
  \vspace{-4mm}
\end{table}

\Cref{tab:method-properties} summarizes the methods used in our evaluation, all implemented as described in \Cref{sec:implementation}. Each approach is evaluated in two modes. First, as an \emph{individual agreement method} in RQ1, e.g., \emph{Postcondition} checks whether generated code satisfies a generated postcondition. Since FWD-INV, FWD-SINV, and ENUM-SINV represent different ways of obtaining such agreement, we evaluate them separately here. Second, as a \emph{sample consensus method} in RQ2 and RQ3, where selection and abstention decisions are formed by aggregating agreements across an entire sample, typically using the RANSAC procedure described in \Cref{sec:background}. Because the three triangulation methods target different problem types, we combine them into a single tool, \toolName, for those evaluations. It sequentially tries ENUM-SINV, FWD-SINV, FWD-INV, with STREAM automatically enabled when the LLM classifies the problem as stream-processing. It forms consensus via RANSAC and abstains if no triangulation succeeds. \emph{Majority0.5} selects a solution whose estimated probability is at least $0.5$ and abstains if no such solution exists. \emph{Syntactic} is a variant of the metamorphic testing approach of Wang et al.~\cite{wang2024validating}: it translates problem statements from English to Chinese and checks equivalence of the resulting solutions. \emph{OffByOne} is a metamorphic testing approach that applies off-by-one perturbations to problem specifications, introducing trivial semantic changes, and verifies that solutions change accordingly.

Important hyperparameters involve the number of samples, which we set to 30, as at this rate the entropy of distribution of equivalence classes stabilizes, and the threshold for the tolerance of angelic values during property checking (\Cref{sec:implementation}), which we set to $\frac{1}{3}$ based on a pilot study on MBPP~\cite{austin2021program}. We use the default temperature 1.0 for sampling. Refer to \Cref{sec:entropy_stabilize} for details.

Since traditional code generation benchmarks HumanEval~\cite{chen2021evaluating} and MBPP~\cite{austin2021program} are reaching their saturation, we evaluated the above configurations on newer, more challenging benchmarks:

\begin{itemize}
  \item LiveCodeBench (LCB)~\cite{jain2024livecodebench} is a popular, challenging code generation dataset; to ensure contamination-free evaluation, we used the most recent, sixth segment of LiveCodeBench, including 175 problems published between Feb-Apr 2025. This benchmark itself provides oracle tests, so we use them for correctness checking as ground truth. Its limitations are that it explicitly excludes inexact problems due to the difficulty of testing their solutions.
  \item CodeElo-Inexact (CEI): to address the LCB limitation, we selected inexact problems from CodeElo~\cite{quan2025codeelobenchmarkingcompetitionlevelcode}, which originally contains 408 tasks. We collected all problems whose inexactness is explicitly stated in the problem description (e.g., "If there are multiple answers, print any of them"), yielding 64 problems. We then excluded problems whose inexactness is only due to case-insensitivity (e.g., "YES" vs. "Yes"), leaving the final 31 instances. A limitation of CodeElo is that it does not provide correctness oracles, suggesting the use of a closed third-party evaluation platform, which harms reproducibility. To address it, we manually wrote correctness judges for these 31 problems. As a sanity check, we collected pairs of correct and incorrect submissions for them from CodeForces database~\cite{codeforces}, and ensured that our judges discriminate them.
\end{itemize}

\subsection{RQ1: Confidence-Enhancing Plausibility Witnesses}
\label{sec:rq1}

\begin{wrapfigure}{r}{0.45\textwidth}
\vspace{-3mm}
  \includegraphics[width=0.44\textwidth]{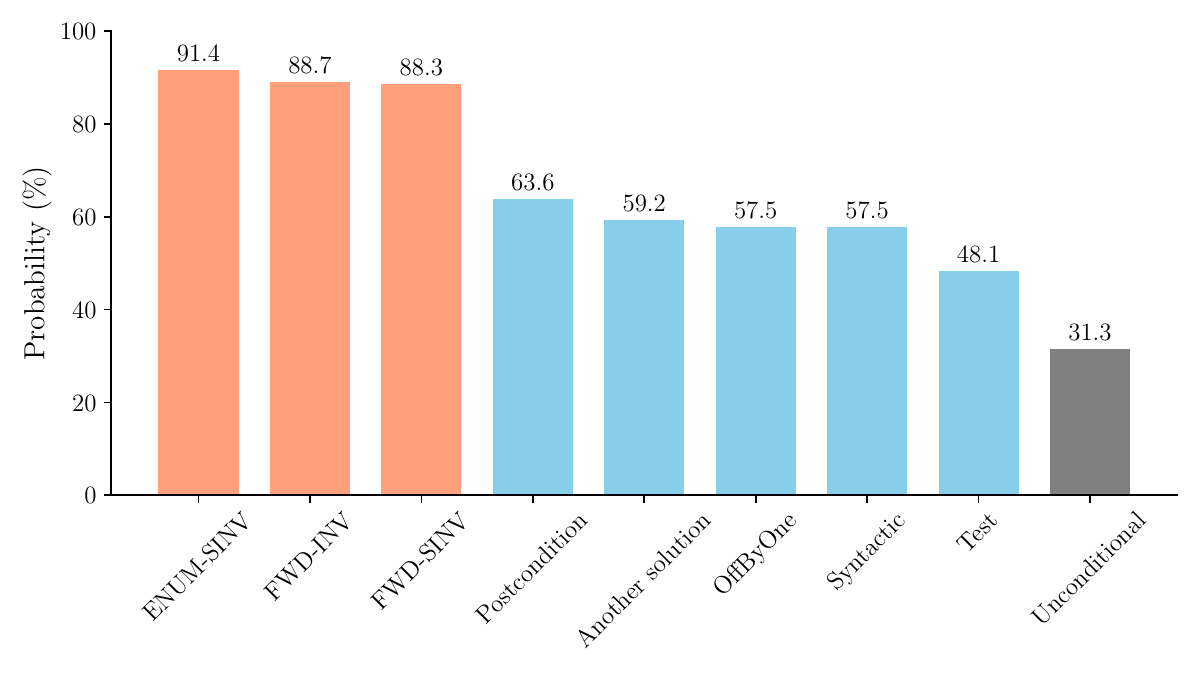}
  \vspace{-5mm}
  \caption{\label{fig:conditional_prob}The probability of GPT-4o sample correctness under agreement across 175 LCB problems. The horizontal axis lists the agreement methods $\mathrm{agree}(p, w)$ described in \Cref{tab:method-properties}. ``Unconditional'' refers to the probability of correctness irrespective of an agreement.}
\end{wrapfigure}

To identify which type of plausibility witness confers greater confidence in program correctness (as posed in \Cref{def:witness_problem}), we evaluate conditional probability of correctness given agreement $\mathds{P}(p\vdash \llbracket d \rrbracket \mid \mathrm{agree}(p, w))$. For \toolName, we report results separately for each triangulation scheme. Empirical estimates of this probability for GPT-4o on 175 LCB tasks are given in \Cref{fig:conditional_prob}. For details of DeepSeek-V3 and Gemini 2.5 Flash results, refer to \Cref{sec:deepseek-V3}.

Although all witnesses increase the probability of correctness in comparison to ``Unconditional'' (no witnesses), triangulation methods consistently outperform the strongest baselines, e.g. Hoare-style postconditions in the case of GPT-4o, by 16\% on average across the evaluated models, confirming our theoretical analysis in \Cref{sec:theory}. We used a task-level paired nonparametric bootstrap with 10,000 resamples over all task-run results. The three methods outperformed postcondition under the 95\% percentile bootstrap criterion: ENUM-SINV (\(\Delta=0.167\), 95\% CI \([0.099, 0.230]\)), FWD-SINV (\(\Delta=0.140\), 95\% CI \([0.090, 0.187]\)), and FWD-INV (\(\Delta=0.159\), 95\% CI \([0.037, 0.227]\)). All confidence intervals exclude zero, indicating statistically significant improvements over postcondition.

Importantly, methods that apply non-dissociative transformations, \emph{Syntactic} (translation to Chinese) and OffByOne (trivial semantic transformation), yield little improvement, as they do not disrupt error correlation. The substantial gap between these baselines and our triangulation methods provides empirical evidence that inversion-based transformations are indeed dissociative.

\begin{tcolorbox}[colback=gray!5, colframe=black!20, arc=2pt, boxrule=0.3mm, breakable, sharp corners, left=2pt, right=2pt, top=2pt, bottom=2pt]
\textbf{RQ1:} Semantic triangulation yields the most confidence-enhancing plausibility witnesses, outperforming the strongest baseline by 16\% on average.
\end{tcolorbox}

\subsection{RQ2: Selection \& Abstention}

\begin{figure*}[t]
    \centering
    \vspace{-3mm} 
    \begin{subfigure}{0.25\textwidth}
        \includegraphics[width=\linewidth]{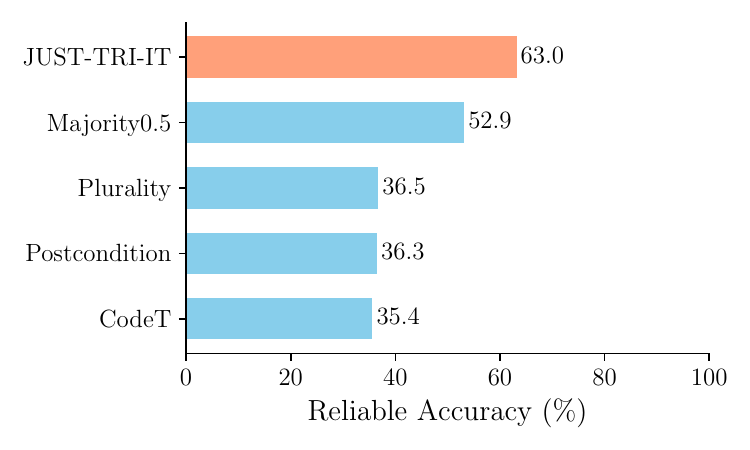}
        \label{fig:subfig_a}
    \end{subfigure}\hfill
    \begin{subfigure}{0.25\textwidth}
        \includegraphics[width=\linewidth]{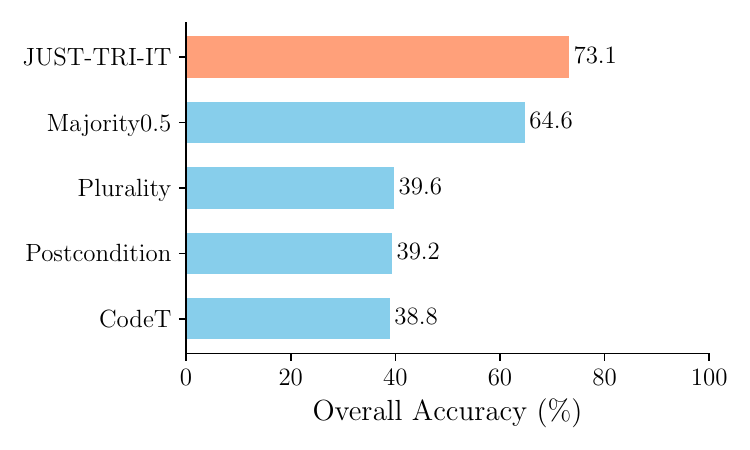}
        \label{fig:subfig_b}
    \end{subfigure}\hfill
    \begin{subfigure}{0.25\textwidth}
        \includegraphics[width=\linewidth]{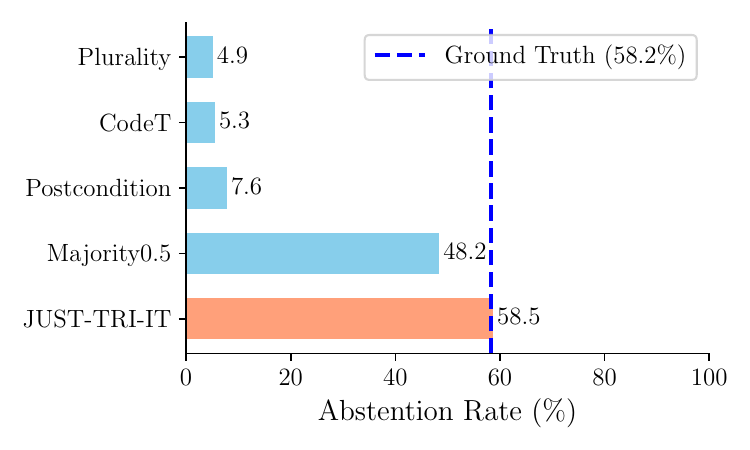}
        \label{fig:subfig_b}
    \end{subfigure}\hfill
    \begin{subfigure}{0.25\textwidth}
        \includegraphics[width=\linewidth]{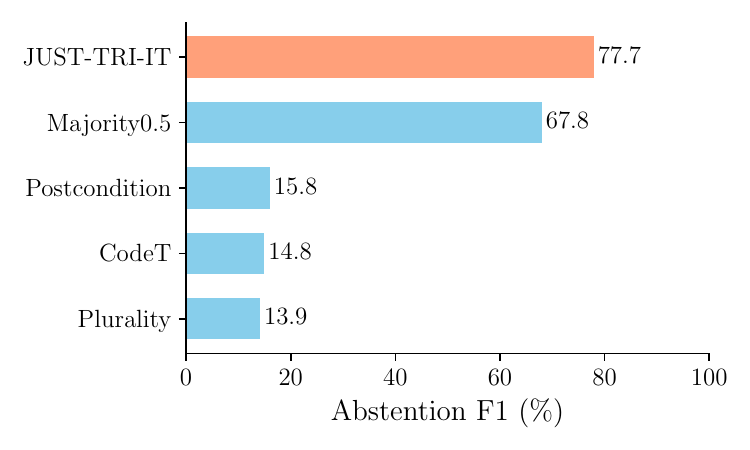}
        \label{fig:subfig_b}
    \end{subfigure}\\[-4mm]
    \begin{subfigure}{0.25\textwidth}
        \includegraphics[width=\linewidth]{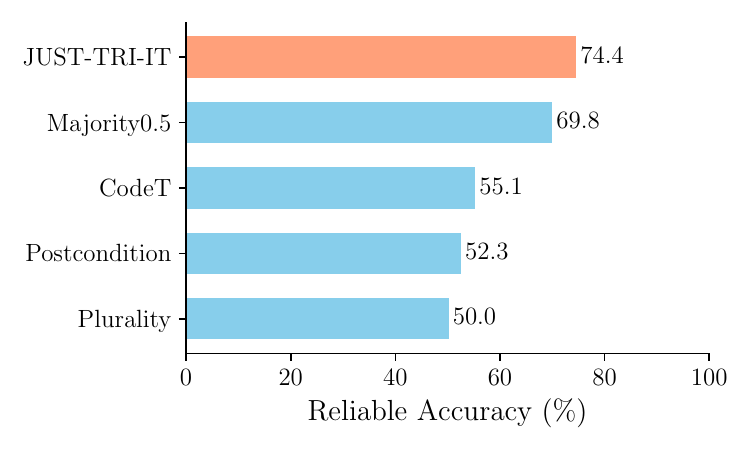}
        \label{fig:subfig_a}
    \end{subfigure}\hfill
    \begin{subfigure}{0.25\textwidth}
        \includegraphics[width=\linewidth]{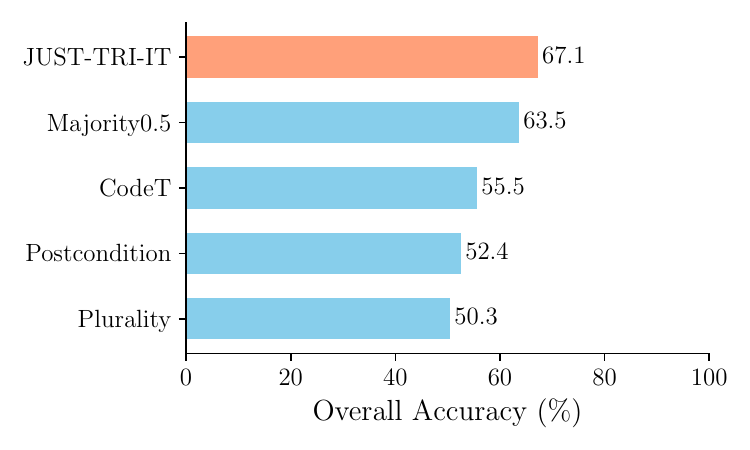}
        \label{fig:subfig_b}
    \end{subfigure}\hfill
    \begin{subfigure}{0.25\textwidth}
        \includegraphics[width=\linewidth]{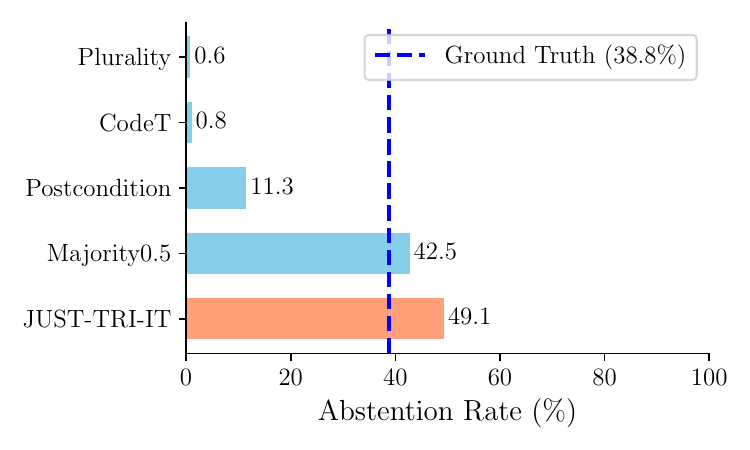}
        \label{fig:subfig_b}
    \end{subfigure}\hfill
    \begin{subfigure}{0.25\textwidth}
        \includegraphics[width=\linewidth]{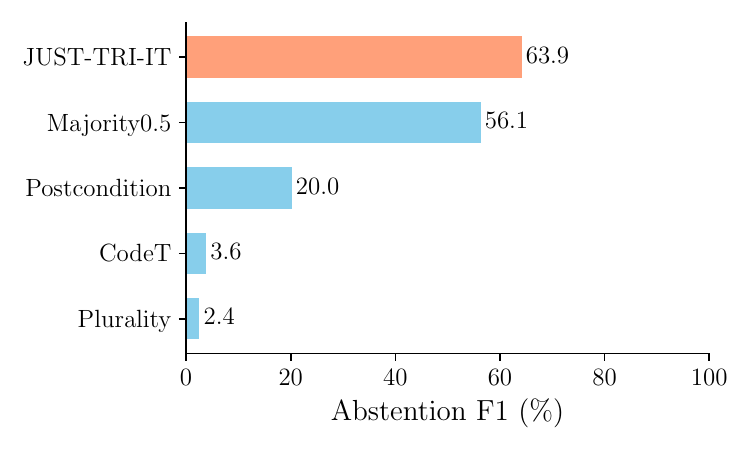}
        \label{fig:subfig_b}
    \end{subfigure}
\\[-4mm]
    \begin{subfigure}{0.25\textwidth}
        \includegraphics[width=\linewidth]{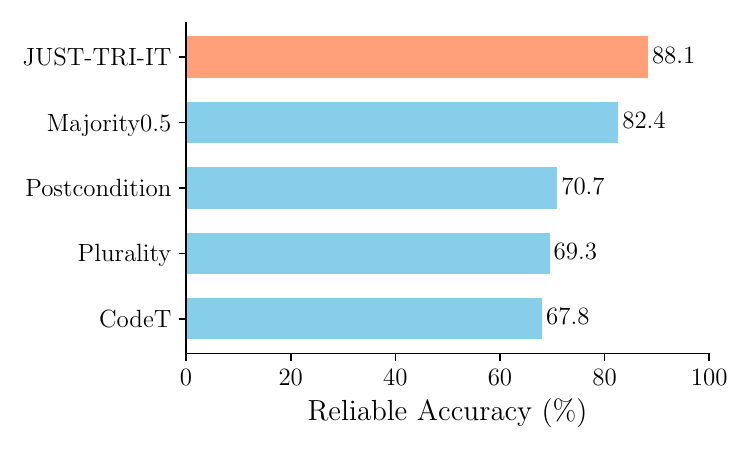}
        \label{fig:subfig_a}
    \end{subfigure}\hfill
    \begin{subfigure}{0.25\textwidth}
        \includegraphics[width=\linewidth]{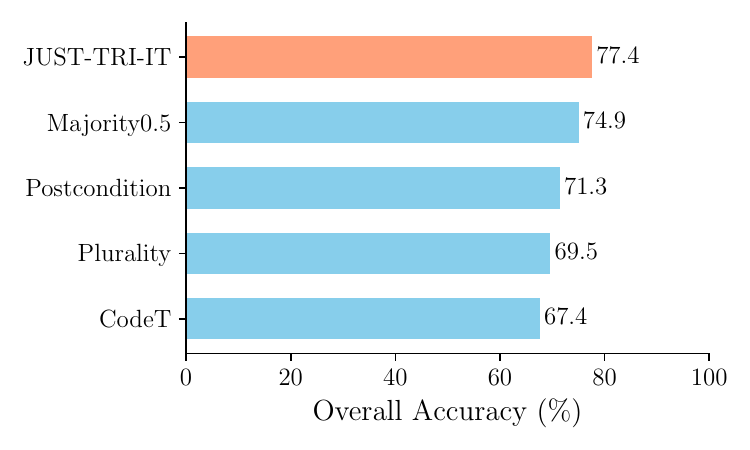}
        \label{fig:subfig_b}
    \end{subfigure}\hfill
    \begin{subfigure}{0.25\textwidth}
        \includegraphics[width=\linewidth]{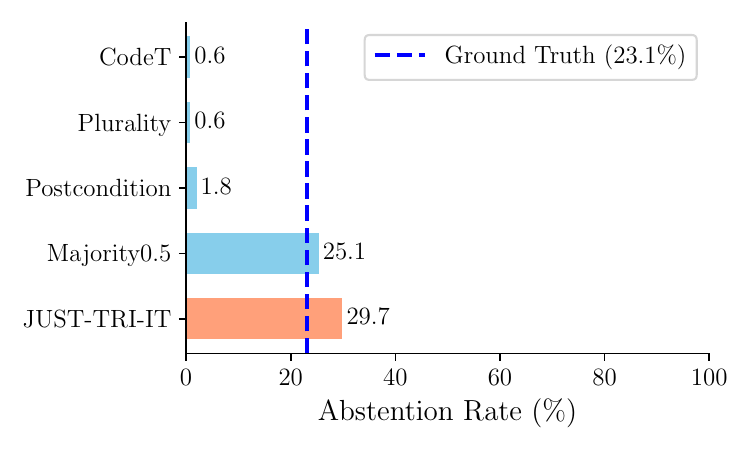}
        \label{fig:subfig_b}
    \end{subfigure}\hfill
    \begin{subfigure}{0.25\textwidth}
        \includegraphics[width=\linewidth]{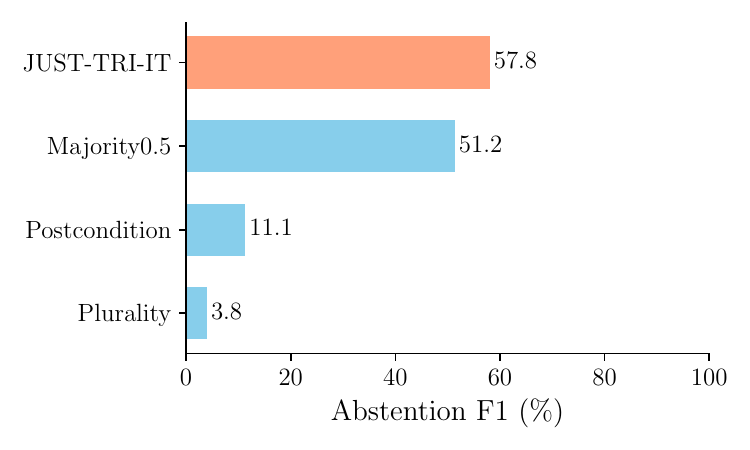}
        \label{fig:subfig_b}
    \end{subfigure}    
    \vspace{-12mm}
    \caption{\label{fig:abstention_lcb}Abstention measures on LCB with GPT-4o (top), DeepSeek-V3 (middle), and Gemini 2.5 Flash (bottom).}
    \vspace{-3mm}
\end{figure*}

While witnesses that bolster confidence in program correctness are valuable, their practical impact is limited if they exist only for cases where LLMs already solve problems with high confidence---for example, by producing correct programs with high probability. Thus, we evaluate how semantic triangulation compares with other techniques in terms of their selection and abstention decisions. Standard measures like Pass@k~\cite{chen2021evaluating} do not take into account reliability of selection and abstention decisions. Instead, we adopted established metrics for evaluating abstention~\cite{kim2025correlated} (\Cref{fig:abstention_matrix}):

\begin{itemize}
    \item \emph{Reliable Accuracy} the proportion of correctly solved tasks among those where the method did not abstain. On problems on which the technique does not abstain, it is identical to Pass@1.
    \item \emph{Overall Accuracy} is a comprehensive measure, incorporating abstention decision.
    \item \emph{Abstention Rate} represents the proportion of tasks in which the method chooses to abstain.
    \item \emph{Abstention F1-score} is a composite metric balancing abstention precision and recall.
\end{itemize}

\Cref{fig:abstention_lcb} summarizes abstention measures for the main configurations\footnote{We excluded Syntactic and OffByOne, as they showed no significant difference compared to plurality.}. Only Majority0.5 and \toolName achieve abstention rate comparable to the ground truth, with \toolName having higher abstention rate, and 7\% higher reliable accuracy on average. The key difference between Majority0.5 and \toolName is that Majority0.5 selects only high-confidence solutions, whereas \toolName can identify low-probability but correct answers. On LCB tasks with GPT-4o, \toolName correctly selected solutions for five problems where Majority0.5 abstained, indicating the potential to reliably solve hard problems. Accordingly, \toolName achieves a 8\% higher abstention F1 score. We also used a task-level paired nonparametric bootstrap with 10,000 resamples over paired task-run results. Under the 95\% percentile bootstrap criterion, \toolName outperformed Majority0.5 in both reliable accuracy (\(\Delta=0.075\), 95\% CI \([0.025, 0.126]\)) and abstention F1 (\(\Delta=0.094\), 95\% CI \([0.032, 0.156]\)). Both confidence intervals exclude zero, indicating statistically significant improvements over Majority0.5.

\begin{tcolorbox}[colback=gray!5, colframe=black!20, arc=2pt, boxrule=0.3mm, breakable, sharp corners, left=2pt, right=2pt, top=2pt, bottom=2pt]
\textbf{RQ2:} \toolName's reliability is, on average, 7\% higher than that of the method selecting only solutions with the probability $\geq 0.5$, while also being able to select correct solutions for tasks solved with only low probability, resulting in a 8\% higher abstention F1 score.
\end{tcolorbox}

\subsection{RQ3: Sample Consensus on Inexact Problems}

Forming consensus is challenging on inexact problems that admit multiple non-equivalent solutions. For example, for CEI problems, GPT-4o produces a correct solution with an average probability of only 0.1. \Cref{fig:abstention_cei} reports the abstention metrics for these models. Compared with the previous experiment, postconditions outperform plurality voting, likely because plurality is degraded by the presence of non-equivalent correct solutions. For Majority0.5, the models exhibited varied behavior on this benchmark: while it remains competitive for GPT-4o and Gemini, its reliability drops to 12.5\% for DeepSeek-V3, whose output distribution contains more correlated errors. In contrast, \toolName displays the most consistent performance, achieving the highest reliable accuracy among the evaluated methods for all three models---10\% higher than Majority0.5 on average---while maintaining comparable abstention F1. Under the task-level paired bootstrap (10,000 resamples), \toolName's improvements on CEI are statistically significant over Plurality (reliable accuracy \(\Delta=0.296\), 95\% CI \([0.165, 0.440]\); abstention F1 \(\Delta=0.662\), 95\% CI \([0.513, 0.745]\)) and Postcondition (\(\Delta=0.237\), 95\% CI \([0.111, 0.379]\); \(\Delta=0.431\), 95\% CI \([0.278, 0.586]\)). The improvement over Majority0.5, while directionally consistent across all three models, is not statistically significant at this benchmark size, as 31 tasks provide insufficient statistical power for this comparison. Remarkably, \toolName maintains high reliability while still selecting low-confidence solutions with sampling probabilities as low as 0.07, as shown in \Cref{fig:enumsinv}.

\begin{tcolorbox}[colback=gray!5, colframe=black!20, arc=2pt, boxrule=0.3mm, breakable, sharp corners, left=2pt, right=2pt, top=2pt, bottom=2pt]
\textbf{RQ3:} On challenging inexact problems, \toolName showed the most consistent performance across models, achieving on average 10\% higher reliable accuracy than a method restricted to selecting high-confidence solutions (probability threshold 0.5), while maintaining comparable abstention F1.
\end{tcolorbox}
\vspace{-2mm}

\begin{figure*}[t]
    \centering
    \vspace{-3mm} 
    \begin{subfigure}{0.25\textwidth}
        \includegraphics[width=\linewidth]{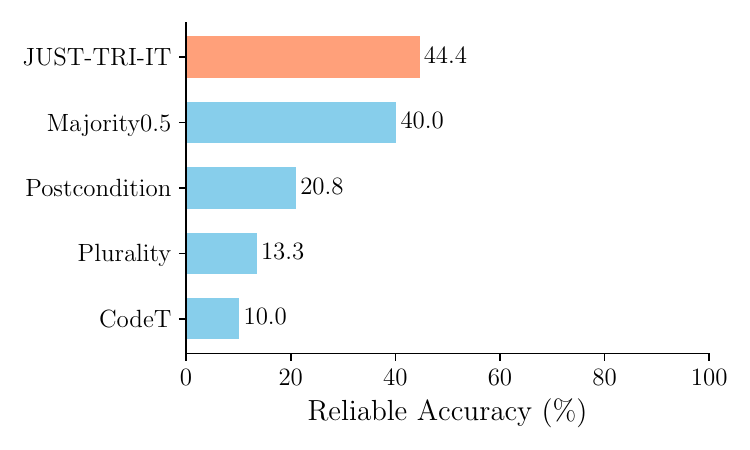}
        \label{fig:subfig_a}
    \end{subfigure}\hfill
    \begin{subfigure}{0.25\textwidth}
        \includegraphics[width=\linewidth]{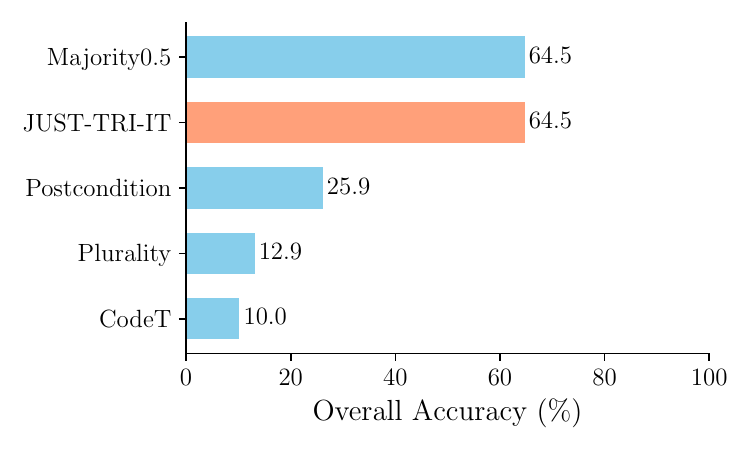}
        \label{fig:subfig_b}
    \end{subfigure}\hfill
    \begin{subfigure}{0.25\textwidth}
        \includegraphics[width=\linewidth]{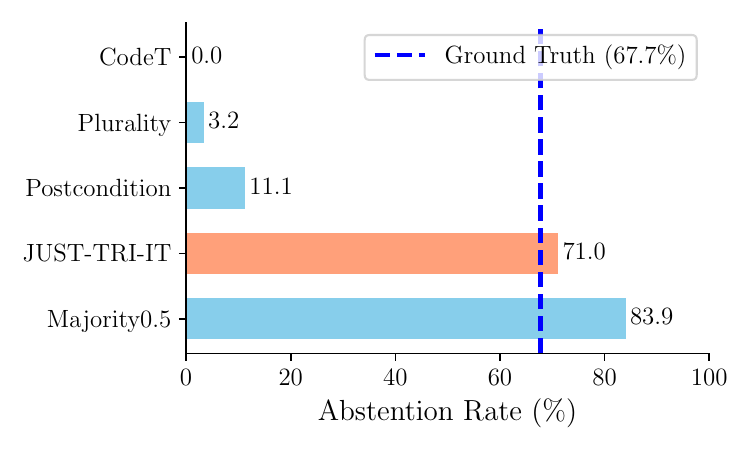}
        \label{fig:subfig_b}
    \end{subfigure}\hfill
    \begin{subfigure}{0.25\textwidth}
        \includegraphics[width=\linewidth]{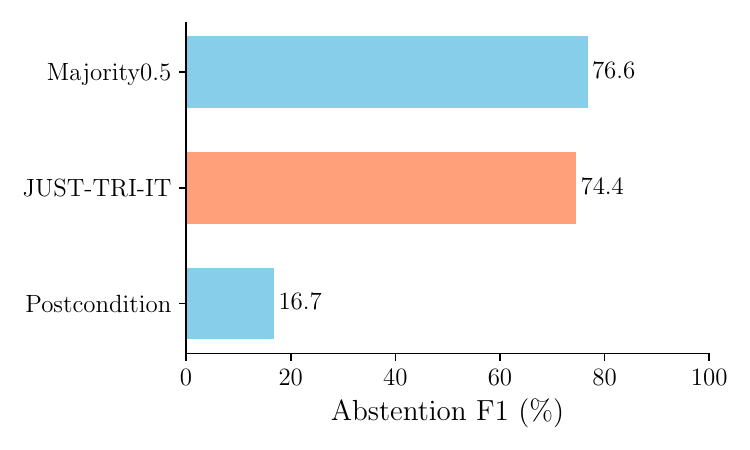}
        \label{fig:subfig_b}
    \end{subfigure}\\[-4mm]
    \begin{subfigure}{0.25\textwidth}
        \includegraphics[width=\linewidth]{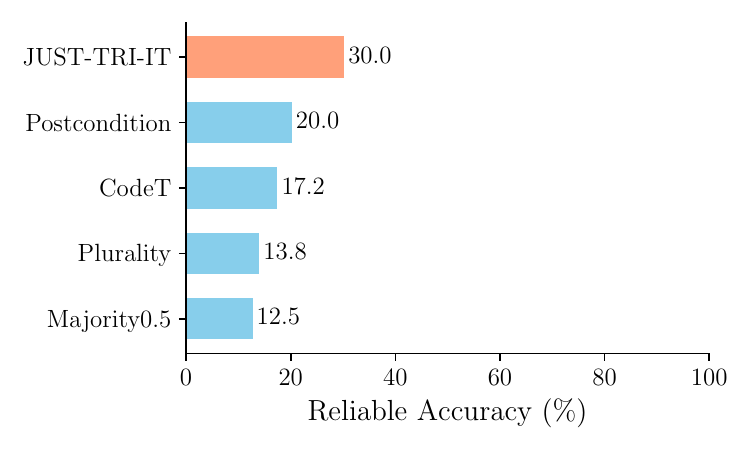}
        \label{fig:subfig_a}
    \end{subfigure}\hfill
    \begin{subfigure}{0.25\textwidth}
        \includegraphics[width=\linewidth]{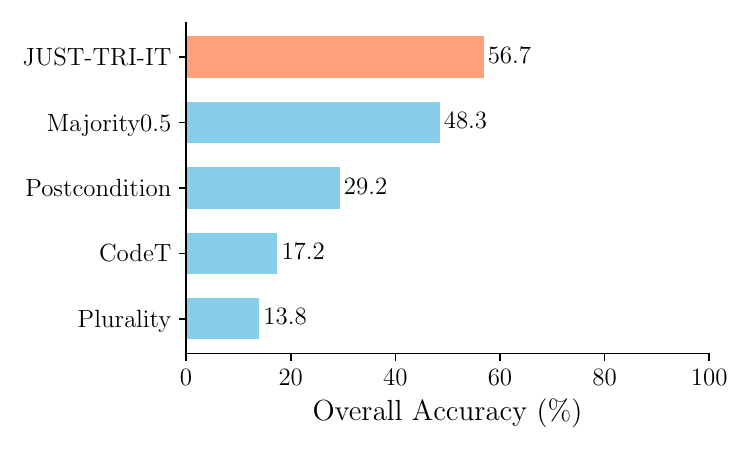}
        \label{fig:subfig_b}
    \end{subfigure}\hfill
    \begin{subfigure}{0.25\textwidth}
        \includegraphics[width=\linewidth]{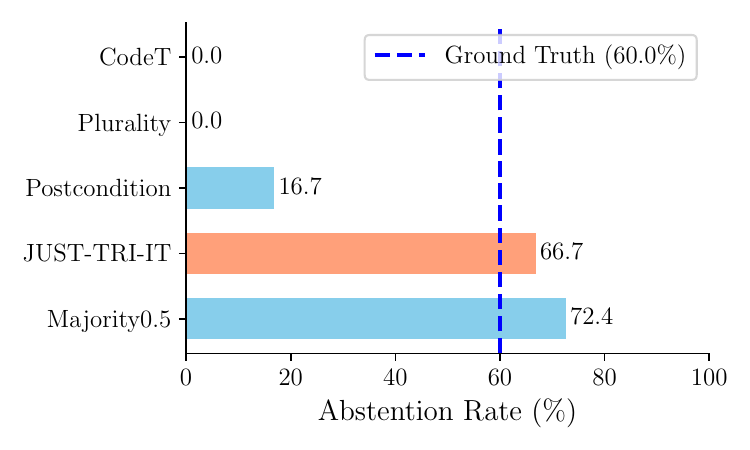}
        \label{fig:subfig_b}
    \end{subfigure}\hfill
    \begin{subfigure}{0.25\textwidth}
        \includegraphics[width=\linewidth]{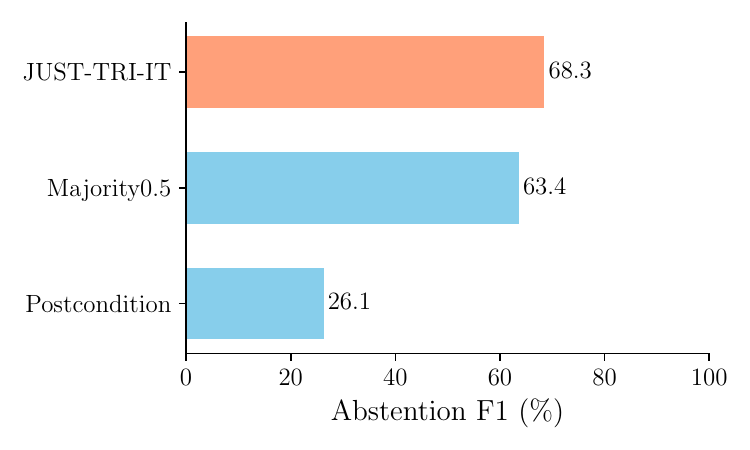}
        \label{fig:subfig_b}
    \end{subfigure}\\[-4mm]
    \begin{subfigure}{0.25\textwidth}
        \includegraphics[width=\linewidth]{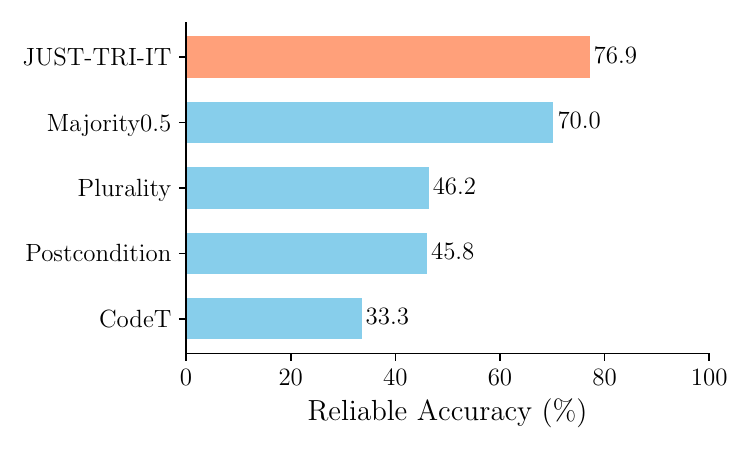}
        \label{fig:subfig_a}
    \end{subfigure}\hfill
    \begin{subfigure}{0.25\textwidth}
        \includegraphics[width=\linewidth]{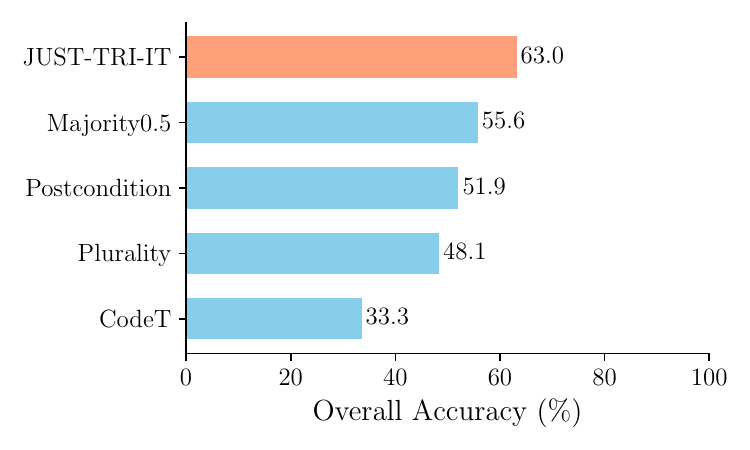}
        \label{fig:subfig_b}
    \end{subfigure}\hfill
    \begin{subfigure}{0.25\textwidth}
        \includegraphics[width=\linewidth]{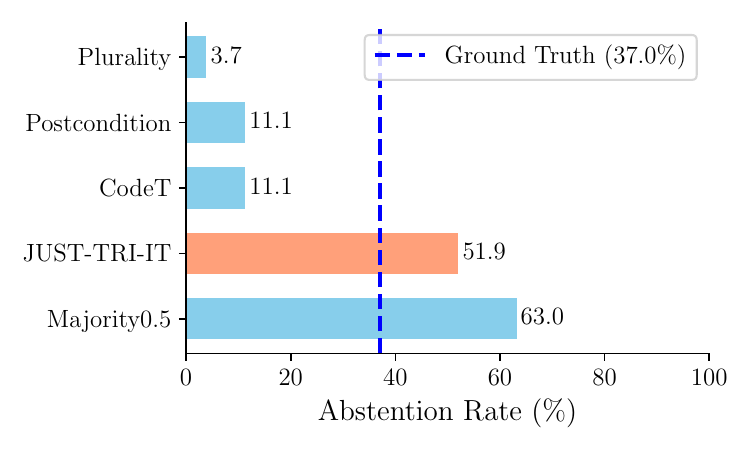}
        \label{fig:subfig_b}
    \end{subfigure}\hfill
    \begin{subfigure}{0.25\textwidth}
        \includegraphics[width=\linewidth]{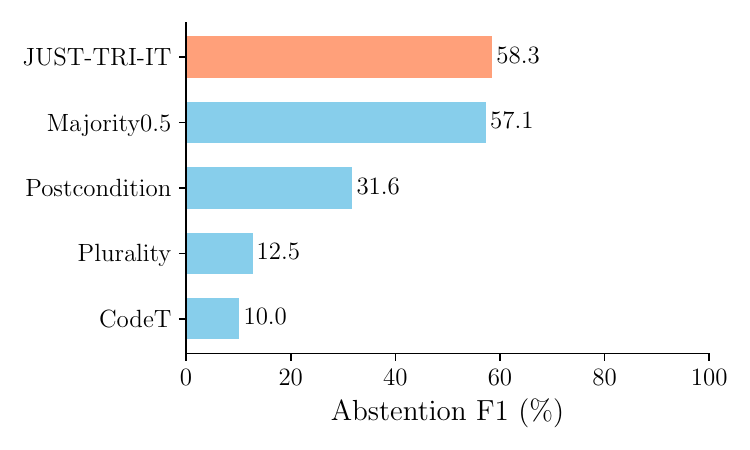}
        \label{fig:subfig_b}
    \end{subfigure}
    \vspace{-12mm}
    \caption{\label{fig:abstention_cei}Abstention measures on CEI with GPT-4o (top), DeepSeek-V3 (middle) and Gemini 2.5 Flash (bottom).}
    \vspace{-4mm}
\end{figure*}

\subsection{RQ4: Ablation Study}

To assess the contribution of each triangulation scheme to \toolName's performance, we counted the correct selections attributed to each scheme on the combined LCB and CEI results for GPT-4o (\Cref{fig:venn}). FWD-SINV yielded the most correct selections---likely because most benchmark problems are exact but do not specify one-to-one input--output mappings---yet every scheme uniquely solved at least one problem.

Our theoretical analysis in \Cref{sec:theory} shows the importance of a bijection between equivalence classes of programs and their witnesses (\Cref{def:triangulation}). To confirm this empirically, we investigated mappings from programs to generated Hoare-style postconditions, since they do not guarantee such a bijection. Among generated postconditions for LCB that matched at least one program, 15\% failed to distinguish a correct sample from an incorrect one; 36\% of these cases led to an incorrect selection.

\begin{wrapfigure}{r}{0.33\textwidth}
  \vspace{-6mm}
  \includegraphics[width=0.31\textwidth]{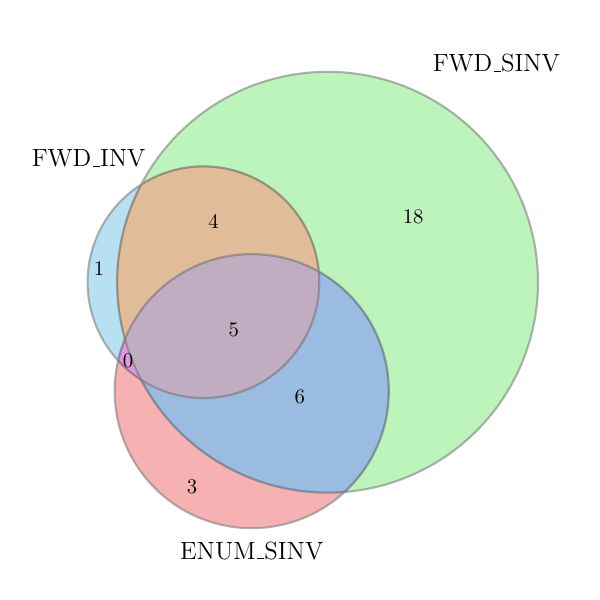}
  \vspace{-5mm}
  \caption{Contribution over all LCB + CEI problems for GPT-4o, where a correct selection was made by at least one triangulation scheme.\label{fig:venn}}
  \vspace{-20mm}
\end{wrapfigure}


Semantic triangulation prevents such situations by design. To investigate the impact of bijectiveness on \toolName, we removed right-hand-side conjuncts from FWD-INV, FWD-SINV and ENUM-SINV hyperproperties (e.g., $\forall (o, i_2), (o', \_)\in O'. ...$ from \Cref{prop:partial_fwd_inv}), breaking \Cref{def:triangulation}, and reran the benchmarks. As a result, the probability of correctness under agreement reduced by 23\% on average, which led to the drop of reliable accuracy by 16\%.

\vspace{-2mm}
\begin{tcolorbox}[colback=gray!5, colframe=black!20, arc=2pt, boxrule=0.3mm, breakable, sharp corners, left=2pt, right=2pt, top=2pt, bottom=2pt]
\textbf{RQ4:} While each triangulation method contributed to the selection of correct solutions, FWD-SINV was the most effective. Postconditions often failed to differentiate between correct and incorrect solutions due to the lack of a bijective mapping ensured by semantic triangulation.
\end{tcolorbox}

\subsection{RQ5: Analysis of Failures}
\label{sec:failure_analysis}

To analyse semantic triangulation's failures to select a correct solution or abstain correctly, we used an LLM to classify them, and then manually inspected and corrected the classification results.

False positive errors (N2 \& N4) are the type of failures that affect reliability. They account for 34\% of all failure cases and fall into three categories: correlated errors between the solution to the transformed problem and the original problem (17\%), weak input coverage (15\%) and an incorrect solution to the original problem matched a solution to an incorrectly transformed problem (2\%). The first category was often caused by the fact that LLM solved the inverse problem by using a forward solution as a subcomponent.


False negative errors (N3 in \Cref{fig:abstention_matrix}) account for about 66\% of all failures and fall into two categories: the failure to correctly solve the transformed problem (51\%) and incorrect transformation (15\%). As mentioned in our theoretical analysis, the transformed problem should not be excessively difficult compared with the original one. However, since we cannot strictly control the difficulty of each transformed problem, some problems may become more challenging after transformation thus exceeding the model's capability. As for the second, the primary issue lies in the loss of critical information during problem transformation. For instance, ``the maximum value at any time'' in the original problem description may be oversimplified to ``the maximum value'' in the transformed one, leading the LLM to misinterpret it as ``the maximum value at the final time''. Crucially, these transformation errors did not compromise reliability, but merely led to extra abstentions.

\vspace{-2mm}
\begin{tcolorbox}[colback=gray!5, colframe=black!20, arc=2pt, boxrule=0.3mm, breakable, sharp corners, left=2pt, right=2pt, top=2pt, bottom=2pt]
\textbf{RQ5:} Failures of semantic triangulation are dominated by false negatives (66\%), due to extra difficulty and information loss introduced by transformation, while false positives mainly arise from correlated errors and weak input coverage.
\end{tcolorbox}

\subsection{RQ6: Time and Token Cost}

\begin{wraptable}{r}{0.58\textwidth}
    \vspace{-9mm}
    \centering
    \footnotesize
    \setlength{\tabcolsep}{2.5pt}
    \caption{Average time and tokens across different models.\label{tab:cost}}
    \vspace{-2mm}
    \begin{tabular}{lrrrrrr}
        \toprule
        & \multicolumn{2}{c}{GPT-4o}
        & \multicolumn{2}{c}{DeepSeek-V3}
        & \multicolumn{2}{c}{Gemini-2.5-Flash} \\
        Method
        & Time(s) & Tokens
        & Time(s) & Tokens
        & Time(s) & Tokens \\
        \midrule
        Plurality     & 24 & 8835 & 21 & 3957 & 13 & 7158 \\
        Majority0.5   & 24 & 8835 & 21 & 3957 & 13 & 7158 \\
        CodeT         & 87 & 10798 & 35 & 10165 & 39 & 8259 \\
        Syntactic     & 39 & 20902 & 30 & 42396 & 26 & 64059 \\
        OffByOne      & 34 & 17561 & 27 & 35575 & 21 & 57579 \\
        Postcondition & 56 & 23073 & 31 & 40339 & 26 & 59160 \\
        FWD\_INV      & 38 & 34390 & 34 & 54530 & 38 & 71650 \\
        FWD\_SINV     & 132 & 28701 & 83 & 53228 & 97 & 64343 \\
        ENUM\_SINV    & 183 & 39630 & 76 & 51310 & 88 & 88821 \\
        \bottomrule
    \end{tabular}
    \vspace{-4mm}
\end{wraptable}

\Cref{tab:cost} compares average time overhead and token usage among various methods and models. Since the latter six methods require generating additional artifacts beyond the original solutions and tests compared to the first three, they exhibit at least a twofold increase in token consumption. In terms of time, three semantic triangulation methods are the most computationally expensive, with the majority of the time spent on property validity checking. Furthermore, the average execution time exhibits an increasing trend from FWD\_INV to FWD\_SINV and ENUM\_SINV, which directly corresponds to the escalating complexity of hyperproperty checking across these three methods. Overall, there is a trade-off: plurality is the cheapest but lowest-precision method, whereas triangulation spends tokens to improve reliability. Crucially, plurality cannot buy back this precision with a larger token budget — more samples only reinforce the dominant error.

\begin{tcolorbox}[colback=gray!5, colframe=black!20, arc=2pt, boxrule=0.3mm, breakable, sharp corners, left=2pt, right=2pt, top=2pt, bottom=2pt]
\textbf{RQ6:} Triangulation methods increase code reliability at the expense of a 2.5$\times$ increase in time and 4.9$\times$ as many tokens as plurality. Compared with autoformalized specifications, triangulation only moderately increases the time by 1.5$\times$ and consumes 2.1$\times$ more tokens.
\end{tcolorbox}

\subsection{RQ7: Empirical Validation of Theoretical Assumptions}

Our theoretical analysis relies on assumptions: the correlated errors (\Cref{assume:correlated}), and a more general dissociative hallucination pattern (\Cref{assume:easy_pattern}). To empirically validate these assumptions, we analyze program-equivalence-class distributions estimated from task runs, reported in \Cref{tab:empirical_val}.

To analyze error correlation, we computed normalized entropy of the distribution over incorrect equivalence classes (1 is for uniform distribution). The estimated values---0.645 for GPT-4o, 0.569 for DeepSeek-V3, and 0.397 for Gemini-2.5-Flash---are substantially below one, indicating moderately concentrated errors. Across model--benchmark task instances, the largest observed equivalence class is incorrect in 52\% of cases. Moreover, among instances in which at least one correct solution is sampled, an incorrect equivalence class has a higher estimated probability than every correct class in 12\% of cases. These observations provide empirical support for \Cref{assume:correlated}: errors are often concentrated in a small number of high-probability equivalence classes rather than being distributed uniformly.

\begin{wraptable}{r}{0.6\textwidth}
\vspace{-5mm}
\centering
\scriptsize
\setlength{\tabcolsep}{4pt}
\caption{Empirical evidence for \Cref{assume:easy_pattern}. Each transformation cell reports \(V_C < V_B\), the empirical left- and right-hand sides of \Cref{eq:dissociative-pattern}, aggregated over LCB and CEI.\label{tab:empirical_val}}
\vspace{-2mm}
\begin{tabular}{lccc}
\toprule
Source & FWD-INV & FWD-SINV & ENUM-SINV\\
\midrule
GPT-4o           & 0.0000 < 0.0141 & 0.0015 < 0.0227 & 0.0000 = 0.0000 \\
DeepSeek-V3      & 0.0052 < 0.0070 & 0.0136 < 0.0745 & 0.0499 < 0.0569 \\
Gemini-2.5-Flash & 0.0000 < 0.7288 & 0.0154 < 0.1028 & 0.0071 < 0.0800 \\
\bottomrule
\end{tabular}
\vspace{-5mm}
\end{wraptable}

To assess the dissociative hallucination pattern, for each triangulation transformation we estimate the two sides of \Cref{eq:dissociative-pattern}: the normalized probability variation over correct equivalence classes, \(V_C\), and that over incorrect equivalence classes, \(V_B\). The first three columns of \Cref{tab:empirical_val} report these values in the form \(V_C < V_B\), estimated over the observed agreement pairs. The inequality holds for eight of the nine model--transformation configurations, indicating that transformations typically perturb the probabilities of incorrect classes more strongly than those of correct classes. For GPT-4o with ENUM-SINV, both quantities are zero, because it yielded only 16 valid task transformations, each with a single agreement, so all paired normalized masses are identical and both estimated variations are zero. Overall, these results provide empirical support for \Cref{assume:easy_pattern}.

\begin{tcolorbox}[colback=gray!5, colframe=black!20, arc=2pt, boxrule=0.3mm, breakable, sharp corners, left=2pt, right=2pt, top=2pt, bottom=2pt]
\textbf{RQ7:} The empirical results support both theoretical assumptions: errors concentrate in a few high-probability equivalence classes (entropy 0.397--0.645, well below uniform), and transformations perturb incorrect classes more strongly than correct ones in eight of nine model--transformation configurations.
\end{tcolorbox}

\section{Threats to Validity \& Discussion}
\label{sec:threats}


\paragraph{External validity.}
LiveCodeBench and CodeElo are challenging benchmarks widely used to evaluate the capabilities of state-of-the-art models. A potential threat to external validity is limited generalization beyond programming-contest problems that form the basis of these benchmarks. In particular, we model program behaviour as mathematical functions, which is suitable for, e.g., algorithm implementations or components of larger systems. For whole projects that maintain state and interact with users or external systems, a potentially suitable model would represent behavior as sets of execution traces or transition systems, capturing sequences of requests, side effects, and state updates. Extending semantic triangulation to such models---for instance, by defining hyperproperties over traces rather than input-output pairs---is a promising direction for future work. We evaluated \toolName on three models from different families, including a reasoning model, addressing generalizability.

\paragraph{Transformation correctness and applicability.}
Our problem transformations are automatically generated via LLM prompting. Although transformation is significantly easier than solving the original task---often amounting to simple linguistic reformulation---it can still introduce errors. However, our analysis in \Cref{sec:failure_analysis} shows that in the case of such mistakes, when a transformation is incorrect, triangulation usually simply fails to find agreement, leading to abstentions, and thus not compromising reliability. Only in 2\% of cases it led to the selection of an unreliable program. A potential approach to mitigate this is to fine-tune a specialized problem-transformation model to ensure high transformation accuracy.

While our triangulation methods rely on partial problem inversion, not all problems are naturally inverted. For example, problems requiring implementing boolean predicates may have intractable input domain to enumerate for each outcome. Similarly, methods presented in this work may be ineffective for optimization problems, since it is often non-trivial to correctly invert a suboptimal solution. Specialized, domain-specific triangulation methods may address these limitations.

\paragraph{Internal validity: LLM-generated test inputs.}


In our method, hyperproperties are evaluated using LLM-generated inputs, which are not fully dissociated from the code generator. Test-input quality could therefore affect the effectiveness of property checking. Importantly, all evaluated methods use the same generated inputs, so this factor does not create bias in their relative comparison.
\Cref{tab:test-input-quality} reports the characteristics of LLM-generated test inputs. Across models, generated inputs achieve 80.9--86.5\% average coverage and 59.2--76.4\% distinguishable-test ratio (the fraction of generated tests that differentiate behaviour of at least two samples), both close to the corresponding oracle-test values, with 23--24 valid tests per program on average. Together, these statistics indicate that the generated inputs provide substantial, though not exhaustive, behavioral evidence for checking hyperproperties.


To investigate whether input- and code-generation failures are correlated, we compute task-level Spearman correlations between the probability of generating correct code and two proxies for input quality: coverage and Diff. Ratio (\Cref{tab:input-corr}). Intuitively, tasks with lower code-correctness probability should also receive lower-quality inputs, yielding positive correlations. We observe the opposite pattern. Coverage is negatively correlated with correctness for GPT-4o and Gemini-2.5-Flash, while the correlation is not significant for DeepSeek-V3. Diff. Ratio is negatively and significantly correlated with correctness for all three models. Thus, there is no evidence that weak input generation systematically co-occurred with weak code generation.

\begin{table*}[t]
\centering
\begin{minipage}[t]{0.55\textwidth}
\centering
\footnotesize
\setlength{\tabcolsep}{3pt}
\caption{Characteristics of LLM-generated test inputs. Parenthesized values are measured using oracle tests. Diff. Ratio is the percentage of generated inputs that distinguish at least two sampled programs.\label{tab:test-input-quality}}
\vspace{-2mm}
\begin{tabular}{lccc}
\toprule
Source & Coverage (\%) & Diff. Ratio (\%) & Valid Inputs \\
\midrule
GPT-4o           & 86.5 (86.9) & 76.4 (73.7) & 23 (37) \\
DeepSeek-V3      & 86.0 (85.8) & 63.8 (62.6) & 24 (38) \\
Gemini-2.5-Flash & 80.9 (81.0) & 59.2 (58.4) & 23 (39) \\
\bottomrule
\end{tabular}
\end{minipage}\hfill
\begin{minipage}[t]{0.42\textwidth}
\centering
\scriptsize
\setlength{\tabcolsep}{3pt}
\caption{Spearman correlations between the probability of generating correct code and, respectively, coverage and Diff. Ratio. Asterisks denote permutation-test significance: *$p<0.05$, **$p<0.01$, ***$p<0.001$.\label{tab:input-corr}}
\vspace{-3mm}
\begin{tabular}{lcc}
\toprule
Source & \(\rho_{\mathrm{cov}}\) & \(\rho_{\mathrm{diff}}\) \\
\midrule
GPT-4o           & -0.204** & -0.222** \\
DeepSeek-V3      & -0.095   & -0.518*** \\
Gemini-2.5-Flash & -0.303*** & -0.527*** \\
\bottomrule
\end{tabular}
\end{minipage}
\vspace{-2mm}
\end{table*}

\begin{wraptable}{r}{0.5\textwidth}
\vspace{-5mm}
\centering
\scriptsize
\setlength{\tabcolsep}{2.5pt}
\caption{Stability of \toolName\ and Majority0.5 over three independent reruns
on a 50-task subset. Entries report mean \(\pm\) standard deviation.\label{tab:rerun-stability}}
\vspace{-2mm}
\begin{tabular}{llccc}
\toprule
Source & Method & Rel. acc & Overall acc & Abs. F1 \\
\midrule
\multirow{2}{*}{GPT-4o}
& JTI & $58.9 \pm 1.1$ & $65.9 \pm 1.3$ & $71.1 \pm 1.8$ \\
& MV0.5 & $55.3 \pm 2.5$ & $62.4 \pm 2.5$ & $57.6 \pm 3.5$ \\
\midrule
\multirow{2}{*}{DeepSeek-V3}
& JTI & $81.5 \pm 3.2$ & $68.1 \pm 3.2$ & $70.0 \pm 4.1$ \\
& MV0.5 & $72.4 \pm 3.9$ & $66.0 \pm 2.4$ & $58.8 \pm 2.6$ \\
\midrule
\multirow{2}{*}{Gemini-2.5-Flash}
& JTI & $86.9 \pm 1.1$ & $69.6 \pm 1.6$ & $53.1 \pm 2.2$ \\
& MV0.5 & $82.1 \pm 0.5$ & $73.1 \pm 1.0$ & $50.9 \pm 1.6$ \\
\bottomrule
\end{tabular}
\vspace{-4mm}
\end{wraptable}
\paragraph{Statistical conclusion validity.}
LLM sampling introduces stochastic variation into all evaluation metrics. Except for reporting task-level paired bootstrap confidence intervals for the main comparisons, we also repeat the evaluation on a 50-task subset for three independent runs (\Cref{tab:rerun-stability}). The results show that the improvements in reliable accuracy and abstention F1 persist across reruns, although their absolute magnitude varies by model and benchmark.


\section{Related Work}


\textit{Theoretical Analysis of LLMs:} To explain the mechanisms underlying semantic triangulation, we use a theoretical model of LLM hallucinations (\Cref{sec:theory}). No widely accepted LLM model exists: PAC learning~\cite{valiant2013probably} assumes supervised learning with fixed, bounded-complexity concept classes, while merely assigning LLMs a non-zero error probability is inadequate---it follows from computability limits~\cite{suzuki2025hallucinations, xu2024hallucination} and does not explain specific error patterns. Our model is inspired by work on spurious correlations~\cite{ye2024spurious}.


\textit{Reliability of LLM-Generated Code:} Methods for improving LLM-generated code control (1) inputs, e.g., prompt engineering such as CoT~\cite{wei2022chain}, RAG~\cite{lewis2020retrieval}, and their enhancements~\cite{li2025structured}; (2) inference, e.g., grammar-~\cite{park2024grammar} and type-constrained decoding~\cite{mundler2025type}; or (3) outputs, which is our focus. Previous plurality-based techniques~\cite{shi2022natural,fan2024oracle} focus on selection of correct samples, but do not tackle abstention~\cite{kim2025correlated} when no correct programs are present in the sample. Algo~\cite{zhang2023algo} generates a slow reference solution using enumerative search to check against an optimized solution, which can be considered as a special case of semantic triangulation. Another line of work, e.g., Cycles~\cite{ding2024cycle}, self-refines generated code; using triangulation as its refinement signal is a promising direction.


\textit{LLM-Based Testing \& Verification:} LLMs created opportunities to automate testing and verification, for example, generate Dafny programs with specifications and assertions to automate proofs~\cite{misu2024towards,mugnier2025laurel}, and infer postconditions from natural language~\cite{endres2024can,ma2024specgen}. For error detection, these methods are limited because LLMs may repeat program errors in tests and specifications, are less effective at predicting test outputs than generating code~\cite{gao2023pal,bouras2025hoareprompt}, and face hallucination-prone low-resource specification languages~\cite{li2025language}; even human-written Eiffel contracts are as error-prone as code~\cite{ciupa2011number}. To address it, CodeT~\cite{chen2022codet} employs a dual execution agreement between sampled programs and tests. Similarly, Huang et al.~\cite{huang2023enhancing} checked consistency between code, specifications and tests. \toolName significantly outperformed these methods in our experiments, and pinpointed a fundamental limitation of naive inference of formal specifications---the lack of a bijective mapping between equivalence classes of solutions and specifications.


\textit{Confidence \& Uncertainty:} Confidence and uncertainty measures~\cite{spiess2024calibration,ye2024benchmarking} can help filter unreliable LLM outputs, such as incorrect programs. In our experiments, majority voting with a probability threshold of 0.5---used here as a simple uncertainty proxy---proved dataset-dependent: it was too permissive on LiveCodeBench and too conservative on CodeElo. Valentin et al.~\cite{valentin2025estimating} proposed a theoretically grounded estimator of correctness as an uncertainty measure, but it relies on the assumption that ``if two LLM-generated programs behave differently on the same input, at least one must be incorrect''. This assumption fails for inexact problems. Our semantic triangulation approach overcomes this limitation via ENUM-SINV (see \Cref{prop:full_enum_sinv}). ValTest~\cite{taherkhani2024valtest} uses a confidence measure (semantic
entropy) to evaluate correctness of generated tests. In contrast, we generate plausibility witnesses whose agreement increases the likelihood of correctness; confidence measures such as ValTest's can further strengthen them.


\textit{Metamorphic Testing:} Metamorphic testing (MT)~\cite{chen2018metamorphic,lee2025metamon} evaluates properties across related programs or executions but does not prescribe the properties or tests; Drowzee~\cite{li2024drowzee} encodes such relations in logic programs to detect hallucinations in general-knowledge QA. Although triangulation is a special case of MT, it significantly differs from typical MT approaches in its formal requirements (\Cref{def:triangulation}), dissociative transformations and bijection-inducing properties. In our evaluation, \toolName significantly outperformed a metamorphic approach by Wang et al.~\cite{wang2024validating} that applies non-dissociative transformations (paraphrasing).

\textit{Dual Learning:} Dual learning~\cite{he2016dual,dognin2020dualtkb} exploits intrinsic task symmetry, such as English-French translation, to form a feedback mechanism for enhanced learning. Our work differs in that we use problem inversion to decorrelate model errors, rather than providing reinforcement learning feedback, and we identified and theoretically justified requirements for effectively reducing hallucinations. These ideas can be potentially applied to dual learning for code models.



\section{Conclusion}

Semantic triangulation provides a principled, black-box way to raise confidence in LLM-generated code by checking cross-task solution consistency. Implemented via partial inversion, answer enumeration, and problem decomposition, it improves reliability measures on LiveCodeBench and CodeElo---achieving correct selection even at low sampling probabilities and outperforming auto-formalized specifications. These results highlight triangulation’s value as a robust consensus mechanism for hard-to-solve, or multi-solution programming tasks.

\section{Acknowledgements}

This work was supported by the National Natural Science Foundation of China (NSFC) under Grant No. 92582202. We thank Dimitrios Stamatios Bouras, Haoxiang Jia, Marcel B{\"o}hme and Earl T. Barr for insightful feedback.

\bibliographystyle{ACM-Reference-Format}
\bibliography{refs}

\pagebreak
\appendix

\section{Proof for Proposition~\ref{prop:triangulation_vs_plurality_easy}}
\label{sec:proof_prop_general}

\begin{proof}
Assume semantic triangulation $(\tau, \phi)$ is defined as in the proof of \Cref{prop:triangulation_vs_plurality} and in accordance with \Cref{assume:easy_pattern}. Continuing the derivation for $\mathds{E}_{d\sim\mathrm{Unif}([d]_\simeq)}(\Delta \mathds{P})$ in \Cref{eq:delta_elaborated},
\begin{eqnarray*}
\mathds{E}_{d\sim\mathrm{Unif}([d]_\simeq)} \left[\frac{\pi_c \pi_{\sigma(c)}}{\sum\limits_{i}\pi_i \pi_{\sigma(i)}} - \frac{\pi_c^2}{\sum\limits_{i}\pi_i^2} \right]
&=& \frac{1}{|[d]_{\simeq}|}\sum_{c\in C_d}\left(\frac{\pi_c \pi_{\sigma(c)}}{\sum\limits_{i}\pi_i \pi_{\sigma(i)}} - \frac{\pi_c^2}{\sum\limits_{i}\pi_i^2} \right)\\
= \frac{1}{|[d]_{\simeq}|}\left(\frac{\sum_{c\in C_d}\pi_c \pi_{\sigma(c)}}{\sum\limits_{i}\pi_i \pi_{\sigma(i)}} - \frac{\sum_{c\in C_d}\pi_c^2}{\sum\limits_{i}\pi_i^2}\right)
&=& \frac{1}{|[d]_{\simeq}|}\left(\frac{1}{1+\frac{\sum_{b\in B_d}\pi_b \pi_{\sigma(b)}}{\sum_{c\in C_d}\pi_c \pi_{\sigma(c)}}} - \frac{1}{1+\frac{\sum_{b\in B_d}\pi_b^2}{\sum_{c\in C_d}\pi_c^2}}\right) \\
\end{eqnarray*}
In order that this quantity gets positive, it suffices to have
\begin{equation*}
\frac{\sum_{b\in B_d}\pi_b \pi_{\sigma(b)}}{\sum_{c\in C_d}\pi_c \pi_{\sigma(c)}}
<\frac{\sum_{b\in B_d}\pi_b^2}{\sum_{c\in C_d}\pi_c^2}
\end{equation*}
in turn,
\begin{equation*}
\frac{\sum_{b\in B_d}\pi_b \pi_{\sigma(b)}}{\sum_{b\in B_d}\pi_b^2}
<\frac{\sum_{c\in C_d}\pi_c \pi_{\sigma(c)}}{\sum_{c\in C_d}\pi_c^2}
\end{equation*}
yet in turn,
\begin{equation*}
\frac{\sum_{b\in B_d}(\pi_b-\pi_{\sigma(b)})^2}{\sum_{b\in B_d}\pi_b^2}
>\frac{\sum_{c\in C_d}(\pi_c-\pi_{\sigma(c)})^2}{\sum_{c\in C_d}\pi_c^2}
\end{equation*}
having in mind an equality like $(a^2+b^2+c^2)-(ab+bc+ca)=\frac{1}{2}((a-b)^2+(b-c)^2+(c-a)^2)$. This is exactly the requirement of a dissociative pattern from \Cref{assume:easy_pattern}.
\end{proof}

\section{Problem Decomposition}
\label{app:problem-decomposition}

Problems involving non-trivial data structures are hard for LLMs to invert. To simplify inversion, we apply a domain-specific decomposition. In particular, we define a triangulation method STREAM for stateless stream-processing problems, i.e., that process sequences of elements without depending on global state. We designate solutions to such problems $p: \mathrm{Seq}[I] \rightarrow \mathrm{Seq}[O]$. Let $[\cdot]$ be a sequence constructor, $\mathrm{head}: \mathrm{Seq}[I] \rightarrow I$ return the first element of a sequence.

\begin{proposition}[STREAM]\label{prop:stream}
  Let $d$ be a stateless stream-processing problem, $(\tau, \phi)$ be semantic triangulation, $\tau^{\circ}$ transform a stateless stream-processing problem into a corresponding pointwise problem that processes individual elements. Let the hyperproperty $\psi$ be defined as
\[
\phi(p^{\circ}, q) \wedge \forall\ [i_1, \cdots, i_n] \in \mathrm{Seq}[I].\ p([i_1, \cdots, i_n]) = [p^{\circ}(i_1), \cdots, p^{\circ}(i_n)]
\]
where $p$ is a sample for $d$, $p^{\circ}\triangleq i \mapsto \mathrm{head}(p([i]))$. If $\tau$ is dissociative, then $(\tau \circ \tau^{\circ}, \psi)$ is semantic triangulation. We refer to this triangulation meta-method as STREAM.
\end{proposition}

\begin{proof}
  A proof that $\psi$ is bijection-inducing and correctness-coupling is mechanized in Lean.
\end{proof}

\noindent Effectively, STREAM applies a given triangulation method pointwise and then triangulates the lifting from pointwise to sequence-level processing. This demonstrates the modularity of triangulation.

\section{Operational Semantics for Practical Property Checking and an Example}
\label{app:property-checking}

We augment the semantics of hyperproperties by introducing three special values: the angelic value $\mathcal{A}$, the demonic value $\mathcal{D}$, and the undefined value $\mathcal{U}$, and extended the domains of all property components (functions, predicates and logical connectives) as defined in \Cref{fig:error_handling}. First, if a program fails with \texttt{ValueError}, indicating an invalid input, we set its output to $\mathcal{U}$; if it fails with other errors or timeout, then we treat its output as $\mathcal{D}$. Second, we introduce a function ``$\mathrm{tolerate}$'' that is an identity on all values except for $\mathcal{U}$ where it returns $\mathcal{A}$; we use it to tolerate invalid inputs when their consistency with witnesses cannot be ensured, i.e. $q(\mathrm{tolerate}(p(i_1, i_2)), i_2) = i_1$ in \Cref{prop:partial_fwd_inv}. Third, when dealing with infinite sets, as shown in rule \textsc{in\_inf\_a}, we treat the result of membership check as angelic if the value is not in the returned subset. Finally, we ensure that we only interpret the result of universal quantification as True when there is sufficient evidence, embodied in a threshold for the maximum number of accepted angelic values in rule \textsc{FORALL}. An example of applying these rules is given in below. This design balances the need for strict property enforcement with the practical reality that validity of inputs cannot be precisely identified, and some transformed problems require computing infinite sets.

\begin{figure}[t]
{\scriptsize
\begin{mathpar}
    \inferrule*[right=exec]
    {\langle \Sigma, t \rangle \Rightarrow \overline{s}}
    {\langle \Sigma, p(t) \rangle \Rightarrow \mathtt{exec}(p, \overline{s})}

    \inferrule*[right=or]
    {\langle \Sigma, t_1 \rangle \Rightarrow \overline{s}_1\\
    \langle \Sigma, t_2 \rangle \Rightarrow \overline{s}_2}
    {\langle \Sigma, t_1 \vee t_2 \rangle \Rightarrow \overline{s}_1 \vee \overline{s}_2}

    \inferrule*[right=eq]
    {\langle \Sigma, t_1 \rangle \Rightarrow \overline{s}_1\\
    \langle \Sigma, t_2 \rangle \Rightarrow \overline{s}_2}
    {\langle \Sigma, t_1 = t_2 \rangle \Rightarrow \overline{s}_1 = \overline{s}_2}

    \inferrule*[right=in]
    {\langle \Sigma, v \rangle \Rightarrow \overline{s}\\
    \langle \Sigma, t \rangle \Rightarrow V}
    {\langle \Sigma, v\in t \rangle \Rightarrow \overline{s} \in V}

    \inferrule*[right=exec\_s]
    {\langle \Sigma, t \rangle \Rightarrow s}
    {\langle \Sigma, p(t) \rangle \Rightarrow s}

    \inferrule*[right=or\_d]
    {\langle \Sigma, t \rangle \Rightarrow \mathcal{D}}
    {\langle \Sigma, t \vee \_ \rangle \Rightarrow \mathcal{D}}

    \inferrule*[right=or\_a]
    {\langle \Sigma, t_1 \rangle \Rightarrow \mathcal{A} \\
    \langle \Sigma, t_2 \rangle \Rightarrow \overline{\mathcal{D}}}
    {\langle \Sigma, t_1 \vee t_2 \rangle \Rightarrow \mathcal{A}}

    \inferrule*[right=or\_u]
    {\langle \Sigma, t_1 \rangle \Rightarrow \mathcal{U} \\
    \langle \Sigma, t_2 \rangle \Rightarrow ?\\
    ? \notin \{\mathcal{A}, \mathcal{D}\}}
    {\langle \Sigma, t_1 \vee t_2 \rangle \Rightarrow \mathcal{U}}

    \inferrule*[right=eq\_d]
    {\langle \Sigma, t \rangle \Rightarrow \mathcal{D}}
    {\langle \Sigma, t = \_ \rangle \Rightarrow \mathcal{D}}

    \inferrule*[right=eq\_uu]
    {\langle \Sigma, t_1 \rangle \Rightarrow \mathcal{U} \\
    \langle \Sigma, t_2 \rangle \Rightarrow \mathcal{U}}
    {\langle \Sigma, t_1 = t_2 \rangle \Rightarrow \mathrm{True}}

    \inferrule*[right=eq\_a]
    {\langle \Sigma, t_1 \rangle \Rightarrow \mathcal{A} \\
    \langle \Sigma, t_2 \rangle \Rightarrow \overline{\mathcal{D}}}
    {\langle \Sigma, t_1 = t_2 \rangle \Rightarrow \mathcal{A}}

    \inferrule*[right=eq\_u]
    {\langle \Sigma, t_1 \rangle \Rightarrow \mathcal{U} \\
    \langle \Sigma, t_2 \rangle \Rightarrow \overline{s}}
    {\langle \Sigma, t_1 = t_2 \rangle \Rightarrow \mathcal{U}}

    \inferrule*[right=in\_inf]
    {\langle \Sigma, v \rangle \Rightarrow \overline{s}\\
    \langle \Sigma, t \rangle \Rightarrow V^\ast\\
    \overline{s} \in V^\ast}
    {\langle \Sigma, v\in t \rangle \Rightarrow \mathrm{True}}

    \inferrule*[right=in\_inf\_a]
    {\langle \Sigma, v \rangle \Rightarrow \overline{s}\\
    \langle \Sigma, t \rangle \Rightarrow V^\ast\\
    \overline{s} \notin V^\ast}
    {\langle \Sigma, v\in t \rangle \Rightarrow \mathcal{A}}

    \inferrule*[right=in\_uu]
    {\langle \Sigma, v \rangle \Rightarrow \mathcal{U}\\
    \langle \Sigma, t \rangle \Rightarrow \mathcal{U}}
    {\langle \Sigma, v\in t \rangle \Rightarrow \mathrm{True}}

    \inferrule*[right=in\_s]
    {\langle \Sigma, v \rangle \Rightarrow ?_1\\
    \langle \Sigma, t \rangle \Rightarrow ?_2\\
    \{?_1, ?_2\} \cup S \neq \emptyset \wedge \{?_1, ?_2\} \setminus \{\mathcal{U}\} \neq \emptyset}
    {\langle \Sigma, v\in t \rangle \Rightarrow\ \uparrow(\{?_1, ?_2\})}

    \inferrule*[right=forall]
    {\langle \Sigma, t_1 \rangle \Rightarrow \{\overline{s}_i\}_i\\
    \forall i.\ \langle \Sigma[v\mapsto \overline{s}_i], t_2 \rangle \Rightarrow\ ?_i}
    {\langle \Sigma, \forall\ v\in t_1.\ t_2 \rangle \Rightarrow \forall i.\ ?_i \in \{\mathcal{A}, \mathrm{True}\} \wedge |\{?_i\mid ?_i = \mathcal{A}\}_i| < T}

    \inferrule*[right=forall\_s]
    {\langle \Sigma, t \rangle \Rightarrow s}
    {\langle \Sigma, \forall\ v\in t.\ \_ \rangle \Rightarrow s = \mathcal{A}}

\end{mathpar}
}
\vspace{-6mm}
  \caption{Operational semantics of property checking. The first four rules (\textsc{exec}, \textsc{or}, \textsc{eq}, \textsc{in}) handle normal values; the remaining rules handle errors and infinite sets. Here $\Rightarrow$---the evaluation relation, $\Sigma$---a variable assignment, $t$---a term, $x$---a variable, $S$---all special values $\{\mathcal{A}, \mathcal{U}, \mathcal{D}\}$, $s$---a special value, $\overline{\mathcal{D}}$---any value except for $\mathcal{D}$, $\overline{s}$---a non-special value, $V$---a finite set of non-special values, $V^\ast$---a marked subset of an unknown set of non-special values, $p$---a program, $\mathtt{exec}$---an execution procedure, $?$---any value or set, $\uparrow$ selects the ``strongest'' special value, where $\mathcal{D}$ is stronger than $\mathcal{A}$ and $\mathcal{A}$ is stronger than $\mathcal{U}$, $T$---threshold parameter. Logical connectives are handled analogously to $\vee$, all rules respect commutativity.\label{fig:error_handling}}
  \vspace{-2mm}
\end{figure}

Here, we present an example derivation for the ENUM-SINV hyperproperty under the \Cref{fig:error_handling} semantics for the programs in \Cref{fig:inexact_challenge}. Assume we have the following programs: \( p \triangleq i \mapsto \{i{+}1,\,i{+}2\}\),\quad \( q \triangleq o \mapsto \{o{-}1,\,o{-}2\}\). The property consists of two conjunctive clauses, per input:
\[
\begin{aligned}
L_1(i):&\quad \forall\, o \in p(i).\ \ i \in q(o),\\
L_2(o):&\quad \forall\, i' \in q(o).\ \ o \in p(i').
\end{aligned}
\]
Executions may yield special values $\U$ (invalid/undefined), $\A$ (angelic), and $\D$ (demonic crash/failure) with precedence $\D>\A>\U$. The rule \textsc{forall} uses a threshold $T$ to tolerate a bounded number of $\A$ branches; if $|\{\A_i\}|<T$ it returns $\True$, otherwise it fails.

The following example derivation is for $L_1(-1)$:

\begin{figure}[h]
{\scriptsize
\begin{mathpar}
  \inferrule*[right=\textsc{forall}]
  {{\judg{ p(-1)}{\{0,1\}}}
 \and
  \inferrule*[right=\textsc{in}]
    {\judg{-1}{-1} \\ \judg{ q(0)}{\{-1,-2\}}}
    {\judg{-1\in q(0)}{\True}}
    \and
  \inferrule*[right=\textsc{in}]
    {\judg{-1}{-1} \\ \judg{ q(1)}{\{0,-1\}}}
    {\judg{-1\in q(1)}{\True}}}
  {\judg{\forall\, o \in p(-1).\ {-1} \in q(o)}{\True}}
\end{mathpar}
}
\end{figure}

Imagine now that the LLM decided that it is intractable to enumerate all input values for the set-valued inverse function, and instead produces only the largest such number, i.e. \( q \triangleq o \mapsto \{o{-}1\}^\ast\). In this case, the derivation will change as follows:

\begin{figure}[h]
{\scriptsize
\begin{mathpar}
  \inferrule*[right=\textsc{forall}]
  {{...}
 \and
  \inferrule*[right=\textsc{in\_inf}]
    {\judg{-1}{-1} \\ \judg{ q(0)}{\{-1\}^\ast}}
    {\judg{-1\in q(0)}{\True}}
    \and
  \inferrule*[right=\textsc{in\_inf\_a}]
    {\judg{-1}{-1} \\ \judg{ q(1)}{\{0\}^\ast}}
    {\judg{-1\in q(1)}{\mathcal{A}}}}
  {\judg{\forall\, o \in p(-1).\ {-1} \in q(o)}{\False}}
\end{mathpar}
}
\end{figure}

\noindent That is, the property check will fail, because the number of angelic values $\mathcal{A}$ in \textsc{forall} rule exceeds our threshold $\frac{1}{3}$, indicating insufficient evidence of validity in the presence of intractable enumeration.

\section{Choosing Hyperparameters}
\label{sec:entropy_stabilize}

\paragraph{Model Temperature Setting} Essentially, as the model’s temperature approaches 0, its outputs for the same prompt become highly concentrated, yielding more conservative and repetitive samples. Raising the temperature increases sampling variability, which can produce more creative outputs (at the cost of greater randomness), a trade-off that is often necessary for challenging inputs~\cite{zhu2023hotcoldadaptivetemperature}. Since the two datasets we use (LCB and CEI) consist of recent competition problems and are quite challenging (as also evident in \Cref{fig:prob_correct_dist} with a high proportion of problems that can't be solved over 100 trials), we need the model to explore a broader range of possibilities. Thus, the temperature was set to a relatively high value, 1.0, which is usually the default temperature for API calls.

\begin{wrapfigure}{r}{0.45\textwidth}
  \includegraphics[width=0.45\textwidth]{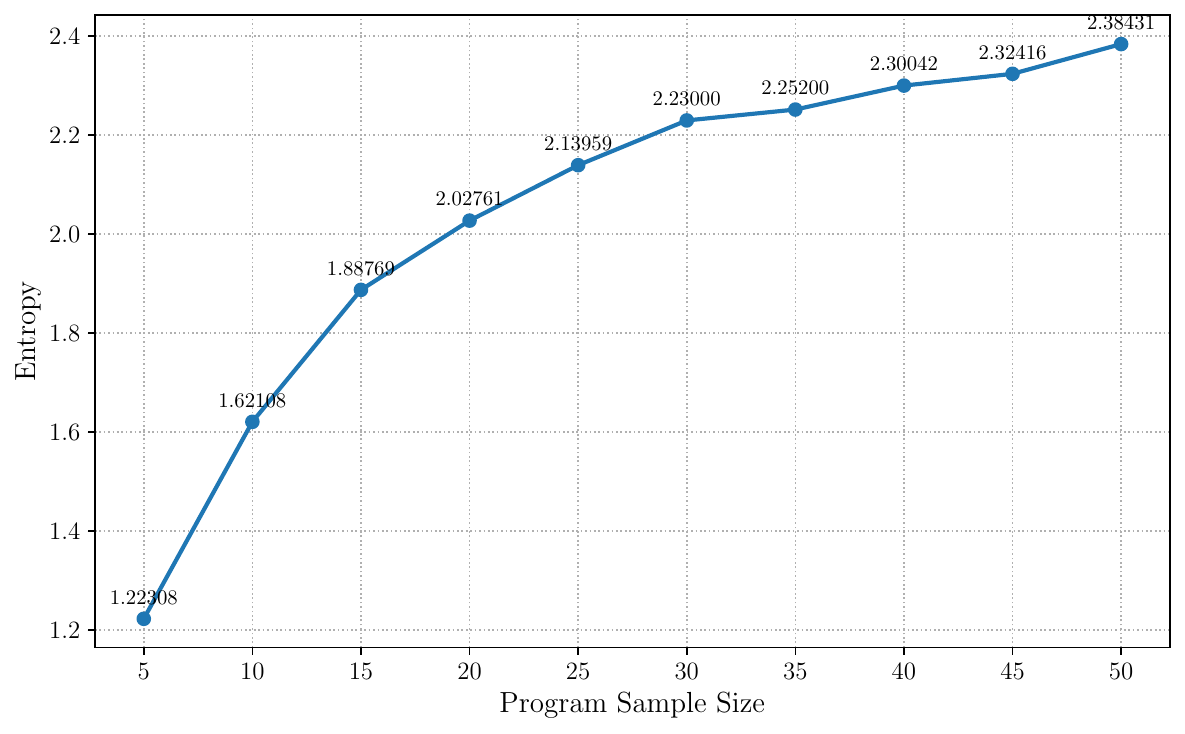}
  \caption{\label{fig:entropy}Semantic entropy of distribution of program equivalence classes with sample size for GPT-4o across 175 LCB problems.} 
  \vspace{-8mm}
\end{wrapfigure}

\paragraph{Sample Size Setting} As noted earlier, the plausibility witnesses underlying Plurality (and its special case, Majority0.5) that we compare against are derived from the semantic equivalence relation over programs. Therefore, to enable a fair evaluation, when choosing the sample size, we require the distribution over program equivalence classes to exhibit statistical stability. To assess this stability, we use semantic entropy~\cite{kuhn2023semantic}, defined as follows:

\vspace{2mm}
\noindent\ \ \ $
\mathrm{SE}\big(m_{\equiv}(\,\cdot \mid x)\big) \mathrel{\triangleq}
- \sum_{y} m_{\equiv}(y \mid x)\,\log m_{\equiv}(y \mid x),
$
\vspace{2mm}

\noindent where $m_{\equiv}$ is the conditional distribution over equivalence classes.

Varying the sample size from 5 to 50 in steps of 5, we show in \Cref{fig:entropy} the resulting entropy change on the 175 LCB problems using GPT-4o (temperature=1.0). It is observed that the entropy increases rapidly from 1.22 (N=5) to 2.23 (N=30), after which the growth becomes markedly slower. Therefore, as a practical trade-off between precision and cost, we select a sample size of 30.

\section{RQ1 Result for DeepSeek-V3 and Gemini 2.5 Flash}
\label{sec:deepseek-V3}

The figures below shows the probability of correctness under agreement for various methods with DeepSeek-V3 and Gemini 2.5 Flash, mirroring the trends observed with GPT-4o, albeit with a smaller gap.

\begin{figure}[h]
    \begin{subfigure}[t]{0.49\textwidth}
    \includegraphics[width=\textwidth]{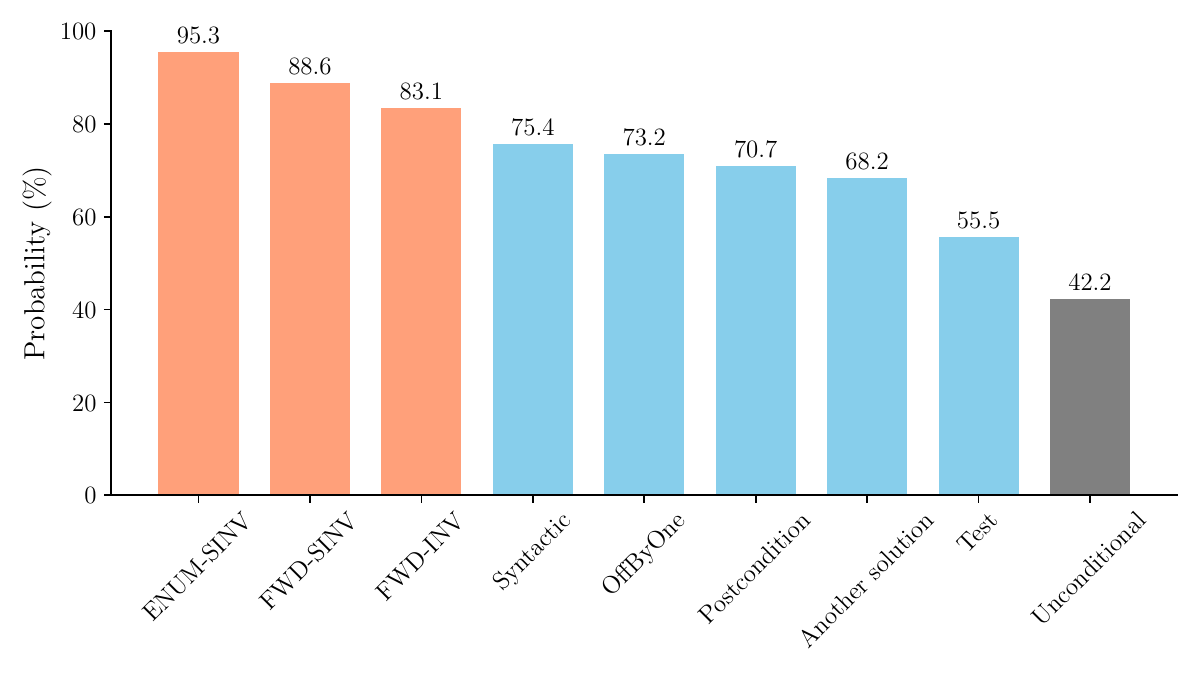}    
    \end{subfigure}%
    \hfill
    \begin{subfigure}[t]{0.49\textwidth} 
    \includegraphics[width=\textwidth]{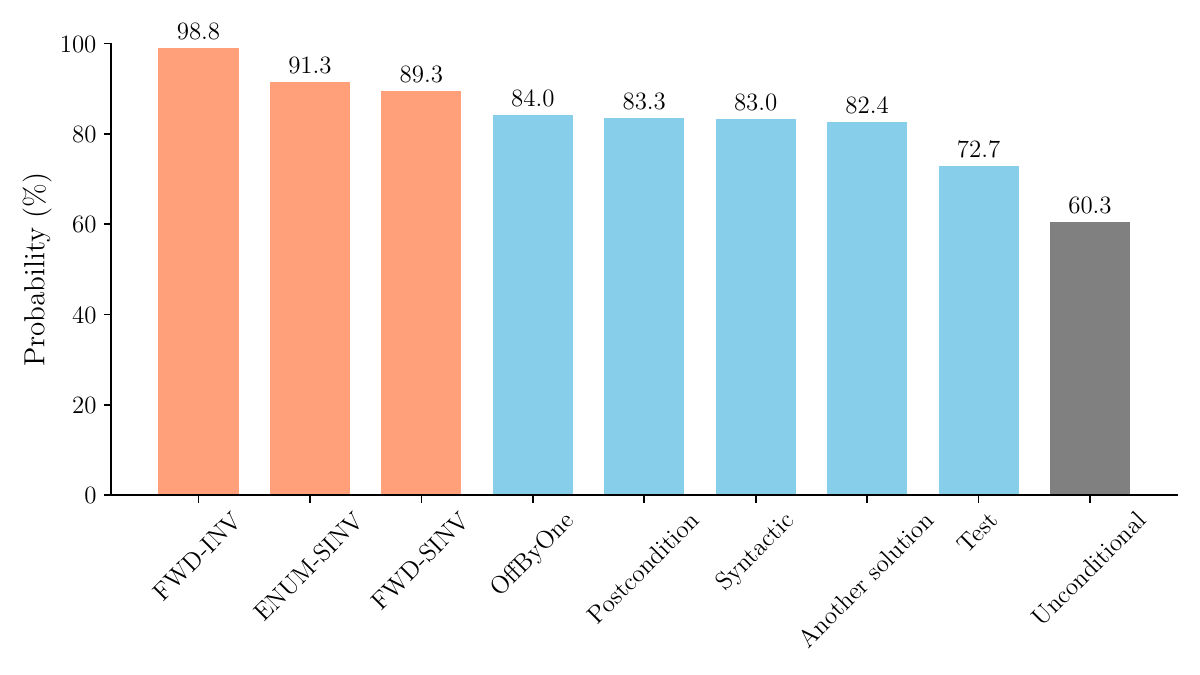}   
    \end{subfigure}%
  \caption{\label{fig:conditional_prob_deepseek}The probability of DeepSeek-V3 (left) and Gemini 2.5 Flash (right) sample correctness under agreement for different methods across 175 LCB problems. ``Unconditional'' refers to the probability of correctness irrespective of an agreement.} 
\end{figure}
\vspace{-5mm}

\section{Lean Proof}
\label{app:lean-proofs}

This appendix gives the Lean mechanization of the semantic triangulation
framework and the concrete triangulation instances used in the paper.
Table~\ref{tab:lean-proof-map} maps paper results to Lean declarations.

\begin{center}
\small
\captionof{table}{Correspondence between paper results and Lean declarations.}
\label{tab:lean-proof-map}
\begin{tabular}{ll}
\toprule
Paper result & Lean declaration \\
\midrule
Definition~\ref{def:triangulation}
  & \texttt{Triangulation} \\
Auxiliary result
  & \texttt{pluralityTriangulation} \\
Proposition~\ref{prop:full_fwd_inv}
  & \texttt{tri\_FULL\_FWD\_INV} \\
Proposition~\ref{prop:partial_fwd_inv}
  & \texttt{tri\_partial\_FWD\_INV} \\
Proposition~\ref{prop:full_fwd_sinv}
  & \texttt{tri\_FWD\_SINV} \\
Proposition~\ref{prop:partial_fwd_sinv}
  & \texttt{triang\_PARTIAL\_FWD\_SINV} \\
Proposition~\ref{prop:full_enum_sinv}
  & \texttt{tri\_Full\_ENUM\_SINV} \\
Proposition~\ref{prop:partial_enum_sinv}
  & \texttt{tri\_Partial\_ENUM\_SINV} \\
Proposition~\ref{prop:stream}
  & \texttt{triangulation\_STREAM} \\
\bottomrule
\end{tabular}
\end{center}

{\scriptsize

\begin{tcolorbox}[mybox,breakable,title={Semantic triangulation framework}]
\begin{minted}{lean}
import Mathlib.Data.Set.Basic

/- Definition 2.1 -/

-- D: type of problem descriptions
-- P: type of "forward" programs (solve D)
-- Q: type of "transformed" programs (solve T D)
structure Triangulation (D D' P Q : Type) where
  T     : D → D'
  Phi   : P → Q → Prop
  /-semantic equivalence on the original side -/
  EP    : Setoid P

  /- Injectivity on classes: if the same q matches two p’s by Φ, then those two p’s are EP-equivalent. -/
  injective  :
    ∀ {p₁ p₂ q}, Phi p₁ q → Phi p₂ q → EP.r p₁ p₂

  /- Surjectivity on classes (onto the relevant image): for every q we can find a p with Φ -/
  surjective :
    ∀ q, ∃ p, Phi p q

  /-- Existence of a mate for every p. -/
  totality   :
    ∀ p, ∃ q, Phi p q

  /-- Correctness predicates (what it means to "solve" a task). -/
  CorrectP : D → P → Prop
  CorrectQ : D' → Q → Prop

  /- Correctness coupling : if q correctly solves T d and Φ(p,q) holds,
      then p correctly solves d. -/
  correctness_coupling :
    ∀ {d p q}, CorrectQ (T d) q → Phi p q → CorrectP d p
\end{minted}
\end{tcolorbox}

\begin{tcolorbox}[mybox,breakable,title={Auxiliary result: triangulation generalizes plurality}]
\begin{minted}{lean}
/-
   Auxiliary result - triangulation generalizes plurality
-/

-- We will need that correctness is stable under ≡ on P. if p ≡ q,both correct or incorrect
def CorrectRespects (D P : Type) (EP : Setoid P)
    (Correct : D → P → Prop) : Prop :=
  ∀ {d p q}, EP.r p q → (Correct d p ↔ Correct d q)

def pluralityTriangulation
  (D P : Type)(EP : Setoid P)(Correct : D → P → Prop)(resp : CorrectRespects D P EP Correct)
  : Triangulation D D P P :=
{ T := id , Phi := fun p q => EP.r p q , EP  := EP ,
    injective := by
    intro p₁ p₂ q h1 h2
    -- h1: Phi p₁ q (p₁ ≡ q) , h2: Phi p₂ q (p₂ ≡ q)
    exact EP.trans h1 (EP.symm h2),

    surjective := by
    intro q;
    exact ⟨q, EP.refl q⟩,

    totality := by
    intro p;
    exact ⟨p, EP.refl p⟩,

    CorrectP := Correct, CorrectQ := fun d' p => Correct d' p ,

    correctness_coupling := by
    intro d p q hQ hPhi
    have hiff := resp (d:=d) (p:=p) (q:=q) hPhi
    -- hiff : Correct d p ↔ Correct d q
    -- from Correct d q, conclude Correct d p
    exact Iff.mpr hiff hQ }
\end{minted}
\end{tcolorbox}

\begin{tcolorbox}[mybox,breakable,title={Proposition~\ref{prop:full_fwd_inv}: Full FWD-INV}]
\begin{minted}{lean}
/-
  Proposition 5.2 — Full FWD_INV
-/

section Full_FWD_INV

variable {A B : Type}

abbrev Spec (A B : Type) := A → B → Prop
def T_inv (R : Spec A B) : Spec B A := fun b a => R a b

def CorrectF (R : Spec A B) (e : Equiv A B) : Prop := ∀ a, R a (e a)
def CorrectG (R' : Spec B A) (e' : Equiv B A) : Prop := ∀ b, R' b (e' b)

def tri_FULL_FWD_INV
  : Triangulation (Spec A B) (Spec B A) (Equiv A B) (Equiv B A) :=
{ T := T_inv, Phi := fun e e' => e' = e.symm
, EP :=
{ r := fun (e₁ e₂ : Equiv A B) => e₁ = e₂
, iseqv := by
    refine ⟨?refl, ?symm, ?trans⟩
    · intro x; rfl
    · intro x y h; simpa using h.symm
    · intro x y z h₁ h₂; simpa using h₁.trans h₂
}

, injective := by
    intro e1 e2 e' h1 h2
    have h : e2 = e1 := by
      simpa [h2] using congrArg Equiv.symm h1
    exact h.symm

, surjective := by
    intro e'
    exact ⟨e'.symm, rfl⟩

, totality := by
    intro e
    exact ⟨e.symm, rfl⟩

, CorrectP := fun R e => CorrectF R e
, CorrectQ := fun R' e' => CorrectG R' e'

, correctness_coupling := by
    intro R e e' hG hPhi a
    subst hPhi
    simpa using hG (e a)
}

end Full_FWD_INV
\end{minted}
\end{tcolorbox}

\begin{tcolorbox}[mybox,breakable,title={Proposition~\ref{prop:partial_fwd_inv}: Partial FWD-INV}]
\begin{minted}{lean}
/-
  Proposition 5.4 — Partial FWD_INV
-/

section Partial_FWD_INV

variable {I₁ I₂ O : Type}

def T_inv1 (d : Spec (I₁ × I₂) O) : Spec (O × I₂) I₁ :=
  fun (oi : O × I₂) (i₁ : I₁) => d (i₁, oi.2) oi.1

abbrev PFwd (I₁ I₂ O : Type) := I₂ → Equiv I₁ O
abbrev PInv (I₁ I₂ O : Type) := I₂ → Equiv O I₁

def setoidPFwd {I₁ I₂ O : Type} : Setoid (PFwd I₁ I₂ O) where
  r p₁ p₂ := ∀ i₂ i₁, (p₁ i₂) i₁ = (p₂ i₂) i₁
  iseqv := {
    refl := fun _ _ _ => rfl
    symm := fun h i₂ i₁ => (h i₂ i₁).symm
    trans := fun h₁ h₂ i₂ i₁ => (h₁ i₂ i₁).trans (h₂ i₂ i₁)
  }

def Phi_partial (p : PFwd I₁ I₂ O) (q : PInv I₁ I₂ O) : Prop :=
  ∀ i₁ i₂, (q i₂) ((p i₂) i₁) = i₁

def CorrectP (d : Spec (I₁ × I₂) O) (p : PFwd I₁ I₂ O) : Prop :=
  ∀ i₁ i₂, d (i₁, i₂) ((p i₂) i₁)
def CorrectQ (d' : Spec (O × I₂) I₁) (q : PInv I₁ I₂ O) : Prop :=
  ∀ o i₂, d' (o, i₂) ((q i₂) o)

def tri_partial_FWD_INV
  : Triangulation (Spec (I₁ × I₂) O) (Spec (O × I₂) I₁) (PFwd I₁ I₂ O) (PInv I₁ I₂ O) :=
{ T := T_inv1, Phi := Phi_partial, EP := setoidPFwd

, injective := by
    intro p₁ p₂ q h1 h2 i₂ i₁
    have inj : Function.Injective (q i₂) := (q i₂).injective
    have e : (q i₂) ((p₁ i₂) i₁) = (q i₂) ((p₂ i₂) i₁) := by
      calc
        (q i₂) ((p₁ i₂) i₁) = i₁ := h1 i₁ i₂
        _ = (q i₂) ((p₂ i₂) i₁) := (h2 i₁ i₂).symm
    exact inj e

, surjective := by
    intro q
    refine ⟨(fun i₂ => (q i₂).symm), ?_⟩
    intro i₁ i₂
    exact Equiv.apply_symm_apply (q i₂) i₁

, totality := by
    intro p
    refine ⟨(fun i₂ => (p i₂).symm), ?_⟩
    intro i₁ i₂
    exact Equiv.symm_apply_apply (p i₂) i₁

, CorrectP := CorrectP, CorrectQ := CorrectQ
, correctness_coupling := by
    intro d p q hQ hPhi i₁ i₂
    have := hQ ((p i₂) i₁) i₂
    simpa [T_inv1, hPhi i₁ i₂] using this
}
end Partial_FWD_INV
\end{minted}
\end{tcolorbox}

\begin{tcolorbox}[mybox,breakable,title={Proposition~\ref{prop:full_fwd_sinv}: Full FWD-SINV}]
\begin{minted}{lean}
/-
  Proposition 5.6 — Full FWD_SINV
-/

section Full_FWD_SINV

variable {I O : Type}

def T_swap (d : Spec I O) : Spec O I := fun o i => d i o

structure SINV (I O : Type) where
  q      : O → Set I
  cover  : ∀ i : I, ∃ o : O, i ∈ q o
  unique : ∀ i o₁ o₂, i ∈ q o₁ → i ∈ q o₂ → o₁ = o₂

def setoidFun (X Y : Type) : Setoid (X → Y) where
  r f g := ∀ x, f x = g x
  iseqv :=
    { refl  := by intro f x; rfl
      symm  := by intro f g h x; simpa using (h x).symm
      trans := by intro f g h h₁ h₂ x; exact (h₁ x).trans (h₂ x) }

-- φ：L1 ∧ L2
def Phi_SINV (p : I → O) (q : SINV I O) : Prop :=
  (∀ i, i ∈ q.q (p i)) ∧ (∀ i i', i' ∈ q.q (p i) → p i' = p i)

def CorrectPSINV (d : Spec I O) (p : I → O) : Prop :=
  ∀ i, d i (p i)
def CorrectQSINV (d' : Spec O I) (q : SINV I O) : Prop :=
  ∀ o i, i ∈ q.q o → d' o i

def sinvOfP (p : I → O) : SINV I O :=
{ q := fun o => { i | p i = o }
, cover := by
    intro i; exact ⟨p i, rfl⟩
, unique := by
    intro i o₁ o₂ hi₁ hi₂
    -- hi₁ : p i = o₁, hi₂ : p i = o₂
    exact hi₁.symm.trans hi₂ }

noncomputable def pOfSINV (q : SINV I O) : I → O :=
  fun i => Classical.choose (q.cover i)

lemma pOfSINV_L1 (q : SINV I O) :
  ∀ i, i ∈ q.q (pOfSINV q i) := by
  intro i
  exact Classical.choose_spec (q.cover i)
lemma pOfSINV_L2 (q : SINV I O) :
  ∀ i i', i' ∈ q.q (pOfSINV q i) → pOfSINV q i' = pOfSINV q i := by
  intro i i' hi'
  have hi'cov : i' ∈ q.q (pOfSINV q i') := pOfSINV_L1 q i'
  have := q.unique i' (pOfSINV q i') (pOfSINV q i) hi'cov hi'
  simp [this]

def tri_FWD_SINV
  : Triangulation (Spec I O) (Spec O I) (I → O) (SINV I O) :=
{ T := T_swap, Phi := Phi_SINV, EP := setoidFun I O

, injective := by
    intro p₁ p₂ q h1 h2 i
    have hi1 : i ∈ q.q (p₁ i) := h1.left i
    have hi2 : i ∈ q.q (p₂ i) := h2.left i
    exact q.unique i (p₁ i) (p₂ i) hi1 hi2

, surjective := by
    intro q
    refine ⟨pOfSINV q, ?_⟩
    refine ⟨?_, ?_⟩
    · exact pOfSINV_L1 q
    · exact pOfSINV_L2 q

, totality := by
    intro p
    refine ⟨sinvOfP p, ?_⟩
    constructor
    · intro i
      -- i ∈ { i | p i = p i }
      simp [sinvOfP]
    · intro i i' hi'
      -- hi' : p i' = p i
      simpa [sinvOfP] using hi'

, CorrectP := CorrectPSINV, CorrectQ := CorrectQSINV

, correctness_coupling := by
    intro d p q hQ hPhi i
    have hi : i ∈ q.q (p i) := hPhi.left i
    simpa [T_swap] using hQ (p i) i hi
}

end Full_FWD_SINV
\end{minted}
\end{tcolorbox}

\begin{tcolorbox}[mybox,breakable,title={Proposition~\ref{prop:partial_fwd_sinv}: Partial FWD-SINV}]
\begin{minted}{lean}
/-
  Proposition 5.8 — Partial FWD_SINV
-/

namespace PARTIAL_FWD_SINV

variable {A Y B : Type}

abbrev Fwd  (A Y B : Type) := A → Y → B
abbrev SINV  (A Y B : Type) := B → Y → Set A
abbrev SpecPSINV (A Y B : Type) := B → Y → Set A → Prop
abbrev Spec (A Y B : Type) := A → Y → B → Prop

def T_pfib_partial (R : Spec A Y B) : SpecPSINV A Y B :=
  fun b y S => ∀ x, x ∈ S → R x y b

def Phi_pSINV (f : Fwd A Y B) (g : SINV A Y B) : Prop :=
  (∀ x y, x ∈ g (f x y) y) ∧
  (∀ x y x', x' ∈ g (f x y) y → f x' y = f x y)

structure Resonator (A Y B : Type) where
  g    : SINV A Y B
  spec : ∃ f : Fwd A Y B, Phi_pSINV f g
  unique : ∀ (y : Y) (x : A) (b₁ b₂ : B), x ∈ g b₁ y → x ∈ g b₂ y → b₁ = b₂

def EP_Fwd2 : Setoid (Fwd A Y B) :=
{ r := fun f₁ f₂ => ∀ x y, f₁ x y = f₂ x y,
  iseqv := ⟨(fun _ _ _ => rfl), (fun h x y => (h x y).symm),
            (fun h₁ h₂ x y => (h₁ x y).trans (h₂ x y))⟩ }

def CorrectF2 (R : Spec A Y B) (f : Fwd A Y B) : Prop :=
  ∀ x y, R x y (f x y)
def CorrectG2 (R' : SpecPSINV A Y B) (g : Resonator A Y B) : Prop :=
  ∀ b y, R' b y (g.g b y)

def triang_PARTIAL_FWD_SINV :
  Triangulation (Spec A Y B) (SpecPSINV A Y B)
                (Fwd A Y B) (Resonator A Y B) :=
{ T := T_pfib_partial
, Phi := fun f g => Phi_pSINV f g.g, EP := EP_Fwd2

, injective := by
    intro f₁ f₂ g h1 h2 x y
    -- L1: x ∈ g (f₁ x y) y
    have hx : x ∈ g.g (f₁ x y) y := h1.left x y
    have hx2 : x ∈ g.g (f₂ x y) y := h2.left x y
    exact g.unique y x (f₁ x y) (f₂ x y) hx hx2

, surjective := by
    intro g
    rcases g.spec with ⟨f, hf⟩
    exact ⟨f, hf⟩

, totality := by
    intro f
    let g : SINV A Y B := fun b y => {x | f x y = b}
    have hΦ : Phi_pSINV f g := by
      constructor
      · intro x y; simp [g]
      · intro x y x' hx'; simpa [g] using hx'
    have hunique : ∀ (y : Y) (x : A) (b₁ b₂ : B), x ∈ g b₁ y → x ∈ g b₂ y → b₁ = b₂ := by
      intro y x b₁ b₂ hb₁ hb₂
      exact hb₁.symm.trans hb₂
    exact ⟨⟨g, ⟨f, hΦ⟩, hunique⟩, hΦ⟩

, CorrectP := fun R f => CorrectF2 R f, CorrectQ := fun R' g => CorrectG2 R' g

, correctness_coupling := by
    intro d f g hG hPhi x y
    have hsound := hG (f x y) y
    have hx : x ∈ g.g (f x y) y := hPhi.left x y
    exact hsound x hx
}

end PARTIAL_FWD_SINV
\end{minted}
\end{tcolorbox}

\begin{tcolorbox}[mybox,breakable,title={Proposition~\ref{prop:full_enum_sinv}: Full ENUM-SINV}]
\begin{minted}{lean}
/-
 Proposition 5.10 - Full ENUM_SINV
-/

namespace Full_ENUM_SINV

variable {I O : Type}

abbrev SpecRel   (I O : Type) := I → O → Prop
abbrev SpecEnum  (I O : Type) := I → Set O → Prop   -- d^≺
abbrev SpecSInv  (I O : Type) := O → Set I → Prop   -- d^≻

-- τ₁, τ₂
def tau1 (R : SpecRel I O) : SpecEnum I O :=
  fun i S => S = { o | R i o }
def tau2 (R : SpecRel I O) : SpecSInv I O :=
  fun o S => S = { i | R i o }

structure EnumDesc (I O : Type) where
  R : SpecRel I O
  d : SpecEnum I O
  repr : d = tau1 (I:=I) (O:=O) R
structure SInvDesc (I O : Type) where
  R : SpecRel I O
  d : SpecSInv I O
  repr : d = tau2 (I:=I) (O:=O) R

-- T = τ₂ ∘ τ₁⁻¹
def T_enum_sinv (de : EnumDesc I O) : SInvDesc I O :=
  { R := de.R, d := tau2 (I:=I) (O:=O) de.R, repr := rfl }

def Phi (p : I → Set O) (q : O → Set I) : Prop :=
  ∀ i o, o ∈ p i ↔ i ∈ q o

def EP : Setoid (I → Set O) :=
{ r := fun p₁ p₂ => ∀ i, p₁ i = p₂ i
, iseqv := ⟨(fun _ _ => rfl),
            (fun h i => (h i).symm),
            (fun h₁ h₂ i => (h₁ i).trans (h₂ i))⟩ }

def CorrectEnum (de : EnumDesc I O) (p : I → Set O) : Prop :=
  ∀ i, de.d i (p i)
def CorrectSInv (ds : SInvDesc I O) (q : O → Set I) : Prop :=
  ∀ o, ds.d o (q o)

def tri_Full_ENUM_SINV
  : Triangulation (EnumDesc I O) (SInvDesc I O) (I → Set O) (O → Set I) :=
{ T := T_enum_sinv, Phi := Phi, EP := EP

, injective := by
    intro p₁ p₂ q h1 h2 i
    apply Set.ext
    intro o
    exact (h1 i o).trans (Iff.symm (h2 i o))

, surjective := by
    intro q
    refine ⟨(fun i => { o | i ∈ q o }), ?_⟩
    intro i o
    rfl

, totality := by
    intro p
    refine ⟨(fun o => { i | o ∈ p i }), ?_⟩
    intro i o
    rfl

, CorrectP := CorrectEnum, CorrectQ := CorrectSInv

, correctness_coupling := by
    intro de p q hQ hPhi
    let R := de.R
    have hde : de.d = tau1 (I:=I) (O:=O) R := de.repr
    have hQ' : ∀ o, q o = { i | R i o } := by
      intro o
      have : (T_enum_sinv de).d o (q o) := hQ o
      simp [T_enum_sinv, tau2] at this
      exact this
    intro i
    have : p i = { o | R i o } := by
      apply Set.ext; intro o
      have φ := hPhi i o
      have iInQ : i ∈ q o ↔ R i o := by
        rw [hQ' o]
        rfl
      exact φ.trans iInQ
    rw [hde]
    simp [tau1, this]
}

end Full_ENUM_SINV
\end{minted}
\end{tcolorbox}

\begin{tcolorbox}[mybox,breakable,title={Proposition~\ref{prop:partial_enum_sinv}: Partial ENUM-SINV}]
\begin{minted}{lean}
/-
 Proposition 5.11 - Partial ENUM_SINV
-/

namespace Partial_ENUM_SINV

variable {I₁ I₂ O : Type}

abbrev SpecRel   (I₁ I₂ O : Type) := I₁ → I₂ → O → Prop
abbrev SpecEnum  (I₁ I₂ O : Type) := I₁ → I₂ → Set O  → Prop
abbrev SpecSInv  (I₁ I₂ O : Type) := O   → I₂ → Set I₁ → Prop

-- τ₁, τ₂
def tau1 (R : SpecRel I₁ I₂ O) : SpecEnum I₁ I₂ O :=
  fun i₁ i₂ S => S = { o | R i₁ i₂ o }
def tau2 (R : SpecRel I₁ I₂ O) : SpecSInv I₁ I₂ O :=
  fun o i₂ S => S = { i₁ | R i₁ i₂ o }

structure EnumDesc (I₁ I₂ O : Type) where
  R : SpecRel I₁ I₂ O
  d : SpecEnum I₁ I₂ O
  repr : d = tau1 (I₁:=I₁) (I₂:=I₂) (O:=O) R
structure SInvDesc (I₁ I₂ O : Type) where
  R : SpecRel I₁ I₂ O
  d : SpecSInv I₁ I₂ O
  repr : d = tau2 (I₁:=I₁) (I₂:=I₂) (O:=O) R

-- T = τ₂ ∘ τ₁⁻¹
def T_enum_sinv (de : EnumDesc I₁ I₂ O) : SInvDesc I₁ I₂ O :=
  { R := de.R, d := tau2 (I₁:=I₁) (I₂:=I₂) (O:=O) de.R, repr := rfl }

def Phi (p : I₁ → I₂ → Set O) (q : O → I₂ → Set I₁) : Prop :=
  ∀ i₁ i₂ o, o ∈ p i₁ i₂ ↔ i₁ ∈ q o i₂

def EP : Setoid (I₁ → I₂ → Set O) :=
{ r := fun p₁ p₂ => ∀ i₁ i₂, p₁ i₁ i₂ = p₂ i₁ i₂
, iseqv := ⟨(fun _ _ _ => rfl),
            (fun h i₁ i₂ => (h i₁ i₂).symm),
            (fun h₁ h₂ i₁ i₂ => (h₁ i₁ i₂).trans (h₂ i₁ i₂))⟩ }

def CorrectEnum (de : EnumDesc I₁ I₂ O) (p : I₁ → I₂ → Set O) : Prop :=
  ∀ i₁ i₂, de.d i₁ i₂ (p i₁ i₂)
def CorrectSInv (ds : SInvDesc I₁ I₂ O) (q : O → I₂ → Set I₁) : Prop :=
  ∀ o i₂, ds.d o i₂ (q o i₂)

def tri_Partial_ENUM_SINV
  : Triangulation (EnumDesc I₁ I₂ O) (SInvDesc I₁ I₂ O) (I₁ → I₂ → Set O) (O → I₂ → Set I₁) :=
{ T := T_enum_sinv, Phi := Phi, EP := EP

, injective := by
    intro p₁ p₂ q h1 h2 i₁ i₂
    apply Set.ext; intro o
    exact (h1 i₁ i₂ o).trans (Iff.symm (h2 i₁ i₂ o))

, surjective := by
    intro q
    refine ⟨(fun i₁ i₂ => { o | i₁ ∈ q o i₂ }), ?_⟩
    intro i₁ i₂ o; rfl

, totality := by
    intro p
    refine ⟨(fun o i₂ => { i₁ | o ∈ p i₁ i₂ }), ?_⟩
    intro i₁ i₂ o; rfl

, CorrectP := CorrectEnum, CorrectQ := CorrectSInv

, correctness_coupling := by
    intro de p q hQ hPhi
    let R := de.R
    have hde : de.d = tau1 (I₁:=I₁) (I₂:=I₂) (O:=O) R := de.repr
    have hQ' : ∀ o i₂, q o i₂ = { i₁ | R i₁ i₂ o } := by
      intro o i₂
      have : (T_enum_sinv de).d o i₂ (q o i₂) := hQ o i₂
      simp [T_enum_sinv, tau2] at this
      exact this
    intro i₁ i₂
    have p_eq : p i₁ i₂ = { o | R i₁ i₂ o } := by
      apply Set.ext; intro o
      have φ := hPhi i₁ i₂ o
      have iInQ : i₁ ∈ q o i₂ ↔ R i₁ i₂ o := by
        rw [hQ' o i₂]
        rfl
      exact φ.trans iInQ
    rw [hde]
    simp [tau1, p_eq]
}

end Partial_ENUM_SINV
\end{minted}
\end{tcolorbox}

\begin{tcolorbox}[mybox,breakable,title={Proposition~\ref{prop:stream}: STREAM}]
\begin{minted}{lean}
/-
 Proposition 5.12 - STREAM
-/

namespace STREAM

variable {I O : Type}

structure SpecDesc (I O : Type) where
  R : I → O → Prop

structure SpecSeqDesc (I O : Type) where
  R : List I → List O → Prop
  stateless : ∀ (p0 : I → O), (∀ i, R [i] [p0 i]) → ∀ l, R l (l.map p0)

structure StreamHom (I O : Type) where
  p : List I → List O
  p0   : I → O
  hom  : ∀ l, p l = l.map p0

def EP_STREAM {D' Q : Type}
  (base_tri : Triangulation (SpecDesc I O) D' (I → O) Q) :
  Setoid (StreamHom I O) where
  r h1 h2 := base_tri.EP.r h1.p0 h2.p0
  iseqv :=
  { refl  := by
      intro h
      exact base_tri.EP.refl h.p0
    symm  := by
      intro h1 h2 h
      exact base_tri.EP.symm h
    trans := by
      intro h1 h2 h3 h12 h23
      exact base_tri.EP.trans h12 h23 }

variable {D' Q : Type}
variable (tri_pt : Triangulation (SpecDesc I O) D' (I → O) Q)

def T_pointwise (d : SpecSeqDesc I O) : SpecDesc I O := { R := fun i o => d.R [i] [o] }
def T := fun d => tri_pt.T (T_pointwise d)
def Phi (h : StreamHom I O) (q : Q) := tri_pt.Phi h.p0 q ∧ ∀ l, h.p l = l.map h.p0

def CorrectP_STREAM (d : SpecSeqDesc I O) (h : StreamHom I O) : Prop := ∀ l, d.R l (h.p l)
def CorrectQ_STREAM (d' : D') (q : Q) : Prop := tri_pt.CorrectQ d' q

def triangulation_STREAM (base_tri : Triangulation (SpecDesc I O) D' (I → O) Q)

(h_CorrectP_def : ∀ (d : SpecDesc I O) (p : I → O), base_tri.CorrectP d p ↔ (∀ i, d.R i (p i)))

  : Triangulation (SpecSeqDesc I O) D' (StreamHom I O) Q :=
{ T  := fun d => base_tri.T (T_pointwise d)
, Phi := fun h q => base_tri.Phi h.p0 q ∧ ∀ l, h.p l = l.map h.p0
, EP  := EP_STREAM base_tri

, injective := by
    intro h1 h2 q h1q h2q
    exact base_tri.injective (p₁:=h1.p0) (p₂:=h2.p0) (q:=q) h1q.left h2q.left

, surjective := by
    intro q; rcases base_tri.surjective q with ⟨p0, hp⟩
    exact ⟨⟨fun l => l.map p0, p0, by intro l; rfl⟩, And.intro hp (by intro l; rfl)⟩

, totality := by
    intro h; rcases base_tri.totality h.p0 with ⟨q, hp⟩
    exact ⟨q, And.intro hp h.hom⟩

, CorrectP := CorrectP_STREAM
, CorrectQ := fun d' q => base_tri.CorrectQ d' q

, correctness_coupling := by
      intro d h q hQ hΦ l
      have maplaw : ∀ l, h.p l = l.map h.p0 := hΦ.right
      have pc0 : base_tri.CorrectP (T_pointwise d) h.p0 :=
        base_tri.correctness_coupling hQ hΦ.left
      rw [h_CorrectP_def] at pc0
      have singletons : ∀ i, d.R [i] [h.p0 i] := by
        intro i
        simpa [T_pointwise] using pc0 i
      have stream_ok : d.R l (l.map h.p0) :=
        d.stateless h.p0 singletons l
      simpa [maplaw l] using stream_ok

}

end STREAM
\end{minted}
\end{tcolorbox}

}

\section{Motivating Example Details}
\label{sec:motivating_details}

The motivating example in \Cref{fig:enumsinv} was simplified for presentation purpose. Here, we provide complete problem descriptions, their transformations and example solutions. Note that for this example \toolName also applied STREAM meta-rule (\Cref{prop:stream}) on top of ENUM-SINV.

{\scriptsize
\begin{tcolorbox}[mybox,breakable,title={Original problem 1999D from CodeElo}]
\begin{minted}[escapeinside=||]{python}
def slavics_exam(t: int, test_cases: list[tuple[str, str]]) -> list[Union[str,tuple[str,str]]]
\end{minted}
\begin{minted}[breaklines]{Markdown}
Slavic has a very tough exam and needs your help in order to pass it. Here is the question he is struggling with:

There exists a string $s$, which consists of lowercase English letters and possibly zero or more "?".

Slavic is asked to change each "?" to a lowercase English letter such that string $t$ becomes a subsequence (not necessarily continuous) of the string $s$.

Output any such string, or say that it is impossible in case no string that respects the conditions exists.

# Input:
The first line contains a single integer $T$ ($1 \leq T \leq 10^4$) — the number of test cases.

The first line of each test case contains a single string $s$ ($1 \leq |s| \leq 2 \cdot 10^5$, and $s$ consists only of lowercase English letters and "?"-s)  – the original string you have.

The second line of each test case contains a single string $t$ ($1 \leq |t| \leq |s|$, and $t$ consists only of lowercase English letters)  – the string that should be a subsequence of string $s$.

The sum of $|s|$ over all test cases doesn't exceed $2 \cdot 10^5$, where $|x|$ denotes the length of the string $x$.

# Output:
For each test case, if no such string exists as described in the statement, output "NO" (without quotes).

Otherwise, output "YES" (without quotes). Then, output one line — the string that respects all conditions.

You can output "YES" and "NO" in any case (for example, strings "yEs", "yes", and "Yes" will be recognized as a positive response).

If multiple answers are possible, you can output any of them.

# Example input:
5
?????
xbx
ab??e
abcde
ayy?x
a
ab??e
dac
paiu
mom

# Example output:
YES
xabax
YES
abcde
YES
ayyyx
NO
NO
\end{minted}
\end{tcolorbox}
}

{\scriptsize
\begin{tcolorbox}[mybox,breakable,title={Dominant error (probability 0.23)}]
\begin{minted}{python}
def slavics_exam(t: int, test_cases: List[Tuple[str, str]]) -> List[Union[str, Tuple[str, str]]]:
    def is_subsequence(sub: str, full: str) -> bool:
        """Check if sub is a subsequence of full."""
        it = iter(full)
        return all(char in it for char in sub)

    results = []
    for s, t in test_cases:
        if len(t) > len(s):
            results.append("NO")
            continue

        # Attempt to replace '?' in s to make t a subsequence
        modified_s = []
        t_idx = 0
        for char in s:
            if char == '?':
                if t_idx < len(t):
                    modified_s.append(t[t_idx])
                    t_idx += 1
                else:
                    modified_s.append('a') # Any character when t is already matched
            else:
                modified_s.append(char)
        
        modified_s = ''.join(modified_s)

        # Check if t is now a subsequence of modified_s
        if is_subsequence(t, modified_s):
            results.append(("YES", modified_s))
        else:
            results.append("NO")

    return results
\end{minted}
\end{tcolorbox}
}

{\scriptsize
\begin{tcolorbox}[mybox,breakable,title={Low-probability (0.07) correct solution}]
\begin{minted}{python}
def slavics_exam(t: int, test_cases: List[Tuple[str, str]]) -> List[Union[str, Tuple[str, str]]]:
    def can_form_subsequence(s: str, t: str) -> str:
        # Fill '?' in `s` such that `t` can be a subsequence of `s`
        s_list = list(s)
        t_index = 0
        t_len = len(t)
        
        for i in range(len(s_list)):
            if t_index < t_len and s_list[i] == t[t_index]:
                # Advance t_index if matches
                t_index += 1
            elif s_list[i] == '?':
                # Replace '?' with a character from `t` if needed
                if t_index < t_len:
                    s_list[i] = t[t_index]
                    t_index += 1
                else:
                    # Replace remaining '?' with 'a' or any other letter
                    s_list[i] = 'a'
        
        if t_index == t_len:
            return "YES", "".join(s_list)
        else:
            return "NO"
    
    results = []
    for s, t in test_cases:
        if not(1 <= len(s) <= 2 * 10**5) or not(1 <= len(t) <= len(s)):
            raise ValueError('Invalid input')
        
        result = can_form_subsequence(s, t)
        results.append(result)
    
    return results
\end{minted}
\end{tcolorbox}
}

{\scriptsize
\begin{tcolorbox}[mybox,breakable,title={Transformed pointwise problem}]
\begin{minted}[escapeinside=||]{python}
def slavics_exam_pointwise_unpacked(s: str, t: str) -> Union[str, tuple[str, str]]
\end{minted}
\begin{minted}[breaklines]{Markdown}
Slavic has a very tough exam and needs your help in order to pass it. Here is the question he is struggling with:

There exists a string $s$, which consists of lowercase English letters and possibly zero or more "?".

Slavic is asked to change each "?" to a lowercase English letter such that another string becomes a subsequence (not necessarily continuous) of the string $s$.

Output any such string, or say that it is impossible in case no string that respects the conditions exists.

# Input:
The input consists of two parameters:

The first parameter contains a single string $s$ ($1 \leq |s| \leq 2 \cdot 10^5$, and $s$ consists only of lowercase English letters and "?"-s) — the original string you have.

The second parameter contains a single string $t$ ($1 \leq |t| \leq |s|$, and $t$ consists only of lowercase English letters) — the string that should be a subsequence of string $s$.

# Output:
For the parameters, if no such string exists as described in the statement, output "NO" (without quotes).

Otherwise, output "YES" (without quotes). Then, output one line — the string that respects all conditions.

You can output "YES" and "NO" in any case (for example, strings "yEs", "yes", and "Yes" will be recognized as a positive response).

If multiple answers are possible, you can output any of them.
\end{minted}
\end{tcolorbox}
}

{\scriptsize
\begin{tcolorbox}[mybox,breakable,title={Answer enumeration problem}]
\begin{minted}[escapeinside=||]{python}
def slavics_exam_pointwise_unpacked_enum(s: str, t: str) -> list[Union[str, tuple[str, str]]]
\end{minted}
\begin{minted}[breaklines]{Markdown}
Slavic has a very tough exam and needs your help in order to pass it. Here is the question he is struggling with:

There exists a string $s$, which consists of lowercase English letters and possibly zero or more "?".

Slavic is asked to change each "?" to a lowercase English letter such that another string becomes a subsequence (not necessarily continuous) of the string $s$.

Output all such strings, along with "YES" if at least one solution is found, or say that it is impossible in case no string that respects the conditions exists.

# Input:
The input consists of two parameters:

The first parameter contains a single string $s$ ($1 \leq |s| \leq 2 \cdot 10^5$, and $s$ consists only of lowercase English letters and "?"-s) — the original string you have.

The second parameter contains a single string $t$ ($1 \leq |t| \leq |s|$, and $t$ consists only of lowercase English letters) — the string that should be a subsequence of string $s$.

# Output:
For the parameters, output a list of all possible outputs as described in the statement: either "NO" if no such string exists, or one or more tuples containing "YES" (in any case, for example, strings "yEs", "yes", and "Yes" will be recognized as positive responses) followed by a valid string that respects all conditions.

If multiple answers are possible, all of them must be included in the list.
\end{minted}
\end{tcolorbox}
}

{\scriptsize
\begin{tcolorbox}[mybox,breakable,title={ENUM solution}]
\begin{minted}{python}
def slavics_exam_simp_single_unpacked_enum(s: str, t: str) -> List[Union[str, Tuple[str, str]]]:
    if not s or not t or len(t) > len(s):
        raise ValueError('Invalid input')

    def is_subsequence(s: str, t: str) -> bool:
        """Check if t is a subsequence of s."""
        t_index = 0
        for char in s:
            if t_index < len(t) and char == t[t_index]:
                t_index += 1
        return t_index == len(t)

    results = []
    queue = [(s, t, 0)]
    
    while queue:
        current_s, target_t, index = queue.pop(0)
        
        # If '?' is found, try all possible replacements
        if index < len(current_s):
            if current_s[index] == '?':
                for replacement in 'abcdefghijklmnopqrstuvwxyz':
                    new_string = current_s[:index] + replacement + current_s[index+1:]
                    queue.append((new_string, target_t, index + 1))
            else:
                queue.append((current_s, target_t, index + 1))
        else:
            if is_subsequence(current_s, target_t):
                results.append(('YES', current_s))

    if not results:
        return ['NO']
    
    return results
\end{minted}
\end{tcolorbox}
}

{\scriptsize
\begin{tcolorbox}[mybox,breakable,title={Set-valued partial inverse problem w.r.t. NO}]
\begin{minted}[escapeinside=||]{python}
def slavics_exam_pointwise_unpacked_sinv_no(s: str) -> list[str]
\end{minted}
\begin{minted}[breaklines]{Markdown}
Slavic has a challenging exam and needs your help to pass it. Here is a new question he needs assistance with:

There exists a string $s$ which is composed of lowercase English letters and may include zero or more "?". The task is to replace each "?" with a lowercase English letter such that a second given string, $t$, is not a subsequence (not necessarily continuous) of string $s$.

You need to identify the set of all possible strings $t$ that would result in an output of "NO" within reasonable bounds.

# Input:
The input consists of one parameter:

The parameter contains a single string $s$ ($1 \leq |s| \leq 2 \cdot 10^5$, and $s$ consists only of lowercase English letters and "?"-s) — the original string you have.

# Output:
Output a list of all strings $t$ ($1 \leq |t| \leq |s|$) which do not exist as a subsequence of string $s$ after replacing every "?" in $s$ with appropriate lowercase English letters. If no such strings exist within the reasonable bound, return an empty list.

The reasonable bound should be considered as any string $t$ with a length up to that of $s$.

The output list must include all possible strings that satisfy the condition, respecting the bounds.
\end{minted}
\end{tcolorbox}
}

{\scriptsize
\begin{tcolorbox}[mybox,breakable,title={Set-valued partial inverse problem w.r.t. Yes}]
\begin{minted}[escapeinside=||]{python}
def slavics_exam_pointwise_unpacked_sinv_yes(s_with_replaced_marks: str, s: str) -> list[str]
\end{minted}
\begin{minted}[breaklines]{Markdown}
Slavic has a very tough exam and needs your help. Here is the modified question he is struggling with:

There exists a string \( s \), which consists of lowercase English letters and possibly zero or more "?".

You are given a string \( s_{\text{with\_replaced\_marks}} \), which represents \( s \) with all the "?" replaced by lowercase English letters such that a string could become a subsequence of this \( s_{\text{with\_replaced\_marks}} \).

Your task is to determine all the possible values of string \( t \) such that \( t \) is a subsequence of \( s_{\text{with\_replaced\_marks}} \).

# Input:
The input consists of two parameters:

The first parameter contains a single string \( s_{\text{with\_replaced\_marks}} \) (same length and character restrictions as \( s \)) — the original string with "?" replaced.

The second parameter contains a single string \( s \) (\( 1 \leq |s| \leq 2 \cdot 10^5 \), and \( s \) consists only of lowercase English letters and "?"-s) — the original string you have.

# Output:
Output an exhaustive list of strings \( t \) such that if the function were called with any of these \( t \) values and \( s \), it would produce the result: "YES" and the \( s_{\text{with\_replaced\_marks}} \) as the satisfying string.

If no such string \( t \) exists, return an empty list.

Ensure that the solution includes every possible string \( t \) that meets the condition.
\end{minted}
\end{tcolorbox}
}

{\scriptsize
\begin{tcolorbox}[mybox,breakable,title={Composed SINV solution}]
\begin{minted}{python}
def slavics_exam_pointwise_unpacked_sinv_no(s: str) -> list[str]:
    from itertools import product
    import string
    
    def is_subsequence(small: str, large: str) -> bool:
        it = iter(large)
        return all(char in it for char in small)

    if not isinstance(s, str) or any(c not in string.ascii_lowercase + '?' for c in s):
        raise ValueError('Invalid input')
    
    potential_chars = string.ascii_lowercase
    replacements = [c if c != '?' else potential_chars for c in s]
    result_set = set()
    
    # Generate candidate strings whose lengths are in the reasonable bound
    for length in range(1, len(s) + 1):
        # Generate all possible combinations of letters with the given length
        for candidate in product(potential_chars, repeat=length):
            candidate_str = ''.join(candidate)
            can_form_subsequence = False
            # Check if the candidate is a subsequence of any possible version of `s`
            for version in product(*replacements):
                version_str = ''.join(version)
                if is_subsequence(candidate_str, version_str):
                    can_form_subsequence = True
                    break
            
            if not can_form_subsequence:
                result_set.add(candidate_str)
    
    return list(result_set)

from itertools import combinations

def slavics_exam_pointwise_unpacked_sinv_yes(s_with_replaced_marks: str, s: str) -> list[str]:
    def is_subsequence(t: str, s_rep: str) -> bool:
        it = iter(s_rep)
        return all(char in it for char in t)

    if len(s_with_replaced_marks) != len(s):
        raise ValueError('Invalid input')
    
    possible_ts = set()

    # Generate all unique possible 't' strings by considering '?' to be any character
    # that would make it a subsequence.
    for i in range(1, len(s_with_replaced_marks) + 1):
        for comb in combinations(range(len(s_with_replaced_marks)), i):
            candidate = ''.join(s_with_replaced_marks[idx] for idx in comb)
            if is_subsequence(candidate, s_with_replaced_marks):
                possible_ts.add(candidate)
    
    return sorted(possible_ts)


def slavics_exam_pointwise_unpacked_sinv(*args):
    args = list(args)
    if isinstance(args[0], str):
        return (False, slavics_exam_pointwise_unpacked_sinv_no(*args[1:]))
    else:
        yes_args = [args[0][1]] + args[1:]
        return (True, slavics_exam_pointwise_unpacked_sinv_yes(*yes_args))
\end{minted}
\end{tcolorbox}
}

\section{Proof of the Strict Rearrangement Step}
\label{app:strict-rearrangement}

We justify the strict inequality used in the proof of Proposition~4.2.
Let $S$ be a finite index set and let $\sigma:S\to S$ be a permutation.
For any real numbers $\{\pi_i\}_{i\in S}$, consider the gap
\[
\sum_{i\in S}\pi_i^2-\sum_{i\in S}\pi_i\pi_{\sigma(i)}.
\]
Since $\sigma$ is a permutation, it only reorders the indices, and therefore
\[
\sum_{i\in S}\pi_{\sigma(i)}^2=\sum_{i\in S}\pi_i^2.
\]
Thus,
\[
\begin{aligned}
\frac{1}{2}\sum_{i\in S}\left(\pi_i-\pi_{\sigma(i)}\right)^2
&=
\frac{1}{2}\sum_{i\in S}
\left(\pi_i^2+\pi_{\sigma(i)}^2-2\pi_i\pi_{\sigma(i)}\right) \\
&=
\sum_{i\in S}\pi_i^2-\sum_{i\in S}\pi_i\pi_{\sigma(i)}.
\end{aligned}
\]
The left-hand side is a sum of squared terms, so it is nonnegative. Hence
\[
\sum_{i\in S}\pi_i^2\ge
\sum_{i\in S}\pi_i\pi_{\sigma(i)}.
\]

Moreover, the inequality is strict whenever at least one squared term is
positive, i.e., whenever there exists $i\in S$ such that
$\pi_i\ne\pi_{\sigma(i)}$. In our application, we take $S=B_d$. By construction,
$\sigma$ has no fixed points on $B_d$, and by Assumption~4 the probabilities
$\{\pi_b:b\in B_d\}$ are distinct. Therefore, for every $b\in B_d$,
$\pi_b\ne\pi_{\sigma(b)}$, and at least one squared term above is strictly
positive. Consequently,
\[
\sum_{b\in B_d}\pi_b^2>
\sum_{b\in B_d}\pi_b\pi_{\sigma(b)}.
\]
This establishes the strictness condition used in \Cref{eq:delta_elaborated}.

\end{document}